\documentclass[11pt]{article}%
\usepackage{amsfonts}
\usepackage{amsmath}
\usepackage{hyperref}
\usepackage{amssymb}
\usepackage{graphicx}%
\providecommand{\U}[1]{\protect\rule{.1in}{.1in}}
\newtheorem{theorem}{Theorem}

\newtheorem{definition}[theorem]{Definition}

\newtheorem{remark}[theorem]{Remark}

\numberwithin{equation}{section}
\numberwithin{theorem}{section}
\numberwithin{table}{section}
\newenvironment{proof}[1][Proof]{\noindent\textbf{#1.} }{\ \rule{0.5em}{0.5em}}
\begin{document}

\title{New improved estimators for overdispersion in models with clustered
multinomial data and unequal cluster sizes}
\author{Alonso-Revenga, J. M.$^{1}$, Mart\'{\i}n, N.$^{2}$\thanks{Corresponding
author, E-mail: \href{mailto: nirian.martin@uc3m.es}{nirian.martin@uc3m.es}.},
Pardo, L.$^{3}$\\$^{1}${\small Department of Statistics and O.R. III, Complutense University of
Madrid, Spain}\\$^{2}${\small Department of Statistics, Carlos III University of Madrid,
Spain}\\$^{3}${\small Department of Statistics and O.R. I, Complutense University of
Madrid, Spain}}
\maketitle

\begin{abstract}
It is usual to rely on the quasi-likelihood methods for deriving statistical
methods applied to clustered multinomial data with no underlying distribution.
Even though extensive literature can be encountered for these kind of data
sets, there are few investigations to deal with unequal cluster sizes. This
paper aims to contribute to fill this gap by proposing new estimators for the
intracluster correlation coefficient.

\end{abstract}


\section{Introduction\label{sec1}}

When categorical data arise from individuals classified into groups of
individuals or cluster of objects, the major issue is that observations within
a cluster are not independent and the conventional methods of inference for
multinomial sampling, are inappropriate. The strength of similarity of two
observations within a cluster is typically measured by the intracluster
correlation coefficient, whereas observations from separate clusters are
regarded as independent. In most situations, the intracluster correlation
tends to be positive and this induces that the variances of the counts under
clustered sampling to be greater than the ones under multinomial sampling,
namely extra variation with respect to the multinomial sampling (for the
technical details, see page \pageref{page}). This kind of observations are
referred to as overdispersed multinomial clustered data. Some models in the
literature have been considered for this type of \textquotedblleft complex
sampling\textquotedblright. See Altham (1976), Brier (1980), Cohen (1976),
Hall (2000), Men\'{e}ndez et al. (1995, 1996), Morel and Nagaraj (1993),
Neerchal and Morel (1998) and references therein.

A sample of size $n>1$ is taken in each of the $N$ independent clusters,%
\[
\boldsymbol{X}^{\left(  \ell\right)  }=(X_{1}^{\left(  \ell\right)
},...,X_{n}^{\left(  \ell\right)  })^{T}\text{,\quad}\ell=1,...,N,
\]
with realizations in the sampling space $\mathcal{X=}\left\{  1,...,M\right\}
$, i.e.%
\[
\forall\ell=1,...,N,\qquad p_{r}\left(  \boldsymbol{\theta}\right)  =\Pr
(X_{s}^{\left(  \ell\right)  }=r)>0\text{,\quad}s=1,...,n\text{,\quad
}r=1,...,M,
\]
with $\sum_{r=1}^{M}p_{r}\left(  \boldsymbol{\theta}\right)  =1$. This
distribution,%
\begin{equation}
\boldsymbol{p}\left(  \boldsymbol{\theta}\right)  =\left(  p_{1}\left(
\boldsymbol{\theta}\right)  ,...,p_{M}\left(  \boldsymbol{\theta}\right)
\right)  ^{T}, \label{prob}%
\end{equation}
is assumed to be unknown but belonging to a known family of discrete
distributions on $\mathcal{X}$, $\mathcal{P}=\{\boldsymbol{p}\left(
\boldsymbol{\theta}\right)  :\boldsymbol{\theta}\in\Theta\}$, with
$\Theta\subset\mathbb{R}^{M_{0}}$ ($M_{0}\leq M$). In other words, the true
value of parameter $\boldsymbol{\theta}=\left(  \theta_{1},...,\theta_{M_{0}%
}\right)  ^{T}$, $\boldsymbol{\theta}_{0}$, is assumed to be unknown. We
denote the number of units in the $\ell$-th cluster that are classified into
the $r$-th category by
\begin{equation}
Y_{r}^{\left(  \ell\right)  }=%
{\textstyle\sum\limits_{s=1}^{n}}
I_{\left\{  r\right\}  }(X_{s}^{\left(  \ell\right)  })\text{,\quad
}r=1,....,M\text{,\quad}\ell=1,...,N, \label{samples}%
\end{equation}
with $I_{\left\{  r\right\}  }(X_{s}^{\left(  \ell\right)  })$ being equal to
$1$ if $X_{s}^{\left(  \ell\right)  }=r$ and $0$ otherwise. Therefore,
$Y_{1}^{\left(  \ell\right)  }+...+Y_{M}^{\left(  \ell\right)  }=n$, i.e., all
clusters contain the same number of units, $n$. In Section \ref{sec4} a
generalization for unequal cluster sizes is presented. The $M$-dimensional
vector of cell counts associated with the $\ell$-th cluster,%
\begin{equation}
\boldsymbol{Y}^{\left(  \ell\right)  }=(Y_{1}^{\left(  \ell\right)
},...,Y_{M}^{\left(  \ell\right)  })^{T}, \label{CT}%
\end{equation}
is the so-called contingency table.

In what is to follow, we shall assume that $\boldsymbol{p}\left(
\boldsymbol{\theta}\right)  $ belongs to the general class of log-linear
models with full column rank $M\times M_{0}$ design matrix $\boldsymbol{W}$,%
\begin{equation}
\boldsymbol{p}\left(  \boldsymbol{\theta}\right)  =\frac{\exp
\{\boldsymbol{W\theta}\}}{\boldsymbol{1}_{M}^{T}\exp\{\boldsymbol{W\theta}%
\}}\text{,} \label{model}%
\end{equation}
where the $M$ linearly independent column vectors of $\boldsymbol{W}$, are
also linearly independent with respect to the $M$-dimensional vector of ones,
$\boldsymbol{1}_{M}=\left(  1,...,1\right)  ^{T}$. The assumption established
by (\ref{model}) is the condition needed to define the parametric space of
$\boldsymbol{\theta}$, $\Theta$, for log-linear models.

The assumption about $\boldsymbol{p}\left(  \boldsymbol{\theta}\right)  $,
belonging to the general class of log-linear models, covers important models.
We are going to clarify this point for a two dimensional log-linear models,
undestanding that it is easily generalized for any other dimension. If the
$\ell$-th cluster's sample come from a bidimensional variable $(X_{1},X_{2})$
with $I$ and $J$ categories respectively, we have%
\begin{align*}
&  \boldsymbol{X}^{\left(  \ell\right)  }=((X_{1,1}^{\left(  \ell\right)
},X_{2,1}^{\left(  \ell\right)  }),...,(X_{1,n}^{\left(  \ell\right)
},X_{2,n}^{\left(  \ell\right)  }))^{T}\text{,\quad}\ell=1,...,N,\\
&  (X_{1,s}^{\left(  \ell\right)  },X_{2,s}^{\left(  \ell\right)  }%
)\in\mathcal{X}=\{1,...,I\}\times\{1,...,J\}\text{,\quad}s=1,...,n,
\end{align*}
and the single index probability vector (\ref{prob}) matches the double index
probability vector, in lexicographic order,%
\begin{align*}
\boldsymbol{p}\left(  \boldsymbol{\theta}\right)   &  =\left(  p_{11}\left(
\boldsymbol{\theta}\right)  ,p_{12}\left(  \boldsymbol{\theta}\right)
,...,p_{IJ}\left(  \boldsymbol{\theta}\right)  \right)  ^{T},\\
p_{ij}\left(  \boldsymbol{\theta}\right)   &  =\Pr(X_{1}=i,X_{2}%
=j)\text{,\quad}i=1,...,I\text{,\quad}j=1,...,J,
\end{align*}
i.e. in this case, we have $M=I\times J$ cells. The sample of counts given in
(\ref{samples}) for each cluster $\ell=1,...,N$\ can be denoted using the
double index notation, through%
\begin{equation}
Y_{ij}^{\left(  \ell\right)  }=%
{\textstyle\sum\limits_{s=1}^{n}}
I_{\left\{  (i,j)\right\}  }(X_{1,s}^{\left(  \ell\right)  },X_{2,s}^{\left(
\ell\right)  })\text{,\quad}i=1,....,I,\text{\quad}j=1,....,J.\label{sample2}%
\end{equation}
In this setting, we have a two-way contingency table with $I$ rows and $J$
columns for each cluster,%
\[
\boldsymbol{Y}^{\left(  \ell\right)  }=(Y_{11}^{\left(  \ell\right)  }%
,Y_{12}^{\left(  \ell\right)  },...,Y_{IJ}^{\left(  \ell\right)  })^{T},
\]
corresponding to the cells counts of two variables $X_{1}$ and $X_{2}$,
respectively. The independence model between $X_{1}$ and $X_{2}$\ is the most
important model for two-way contingency tables, defined primarily as
\[
p_{ij}\left(  \boldsymbol{\theta}\right)  =p_{i\bullet}\left(
\boldsymbol{\theta}\right)  p_{\bullet j}\left(  \boldsymbol{\theta}\right)
\text{,\quad}i=1,...,I\text{,\quad}j=1,...,J,
\]
where $p_{i\bullet}\left(  \boldsymbol{\theta}\right)  =%
{\textstyle\sum_{j=1}^{J}}
p_{ij}\left(  \boldsymbol{\theta}\right)  $, $p_{\bullet j}\left(
\boldsymbol{\theta}\right)  =%
{\textstyle\sum_{i=1}^{I}}
p_{ij}\left(  \boldsymbol{\theta}\right)  $, and expressed as%
\begin{equation}
\log p_{ij}\left(  \boldsymbol{\theta}\right)  =u+\theta_{1(i)}+\theta
_{2(j)}\text{,\quad}i=1,...,I,\text{\quad}j=1,...,J,\label{In1}%
\end{equation}
in terms of log-linear models, jointly with the restrictions to avoid
overparemeterization,
\[%
{\textstyle\sum\limits_{i=1}^{I}}
\theta_{1(i)}=%
{\textstyle\sum\limits_{j=1}^{J}}
\theta_{2(j)}=0.
\]

For the traditional multinomial log-linear models, the first and second order
moments of $\boldsymbol{Y}^{\left(  \ell\right)  }$ are%
\[
\mathrm{E}[\boldsymbol{Y}^{\left(  \ell\right)  }]=n\boldsymbol{p}\left(
\boldsymbol{\theta}\right)  \text{\quad and\quad}\mathrm{Var}[\boldsymbol{Y}%
^{\left(  \ell\right)  }]=n\boldsymbol{\Sigma}_{\boldsymbol{p}\left(
\boldsymbol{\theta}\right)  },
\]
where%
\begin{equation}
\boldsymbol{\Sigma}_{\boldsymbol{p}\left(  \boldsymbol{\theta}\right)
}=\boldsymbol{D}_{\boldsymbol{p}\left(  \boldsymbol{\theta}\right)
}-\boldsymbol{p}\left(  \boldsymbol{\theta}\right)  \boldsymbol{p}\left(
\boldsymbol{\theta}\right)  ^{T}, \label{var}%
\end{equation}
and $\boldsymbol{D}_{\boldsymbol{p}\left(  \boldsymbol{\theta}\right)  }$\ is
the diagonal matrix of $\boldsymbol{p}\left(  \boldsymbol{\theta}\right)  $.
In this paper, we shall assume the components of sample vectors
$\boldsymbol{Y}^{\left(  \ell\right)  }$ to be overdispersed with respect to
the model with multinomial sampling, i.e.,%
\begin{equation}
\mathrm{E}[\boldsymbol{Y}^{\left(  \ell\right)  }]=n\boldsymbol{p}\left(
\boldsymbol{\theta}\right)  \text{\quad and\quad}\mathrm{Var}[\boldsymbol{Y}%
^{\left(  \ell\right)  }]=\vartheta_{n}n\boldsymbol{\Sigma}_{\boldsymbol{p}%
\left(  \boldsymbol{\theta}\right)  }, \label{moments2}%
\end{equation}
with%
\begin{equation}
\vartheta_{n}=1+\left(  n-1\right)  \rho^{2}\in(1,n] \label{DF}%
\end{equation}
referred to as \textquotedblleft design effect\textquotedblright\ and
$\rho^{2}\in(0,1]$ to as \textquotedblleft intracluster correlation
coefficient\textquotedblright. Notice that $\vartheta_{n}=1$ would correspond
to the multinomial sampling with parameters $n$ and $\boldsymbol{p}\left(
\boldsymbol{\theta}\right)  $, which means that either the components of
$\boldsymbol{X}^{\left(  \ell\right)  }$ are mutually independent ($\rho
^{2}=0$) or there is a unique observation without possibility of being
correlated ($n=1$).

\label{page}In order to interpret $\rho^{2}$, we could consider $(Y_{r}%
^{\left(  \ell\right)  }|Z_{r}=p_{r}\left(  \boldsymbol{\theta}\right)  )\sim
Bin(n,p_{r}\left(  \boldsymbol{\theta}\right)  )$, with $Z_{r}$ being a
generic latent random variable which models the probability of success for
each of the individuals associated with $Y_{r}^{\left(  \ell\right)  }$, with
$\mathrm{E}[Z_{r}]=p_{r}\left(  \boldsymbol{\theta}\right)  $ and
$\mathrm{Var}[Z_{r}]=\mathrm{E}[Z_{r}^{2}]-\mathrm{E}^{2}[Z_{r}]$ has a
general shape. Since the support of $Z_{r}$ is $[0,1]$, it holds that
$Z_{r}\geq Z_{r}^{2}$ and so $\mathrm{E}[Z_{r}]\geq\mathrm{E}[Z_{r}^{2}]$ or%
\begin{align}
\mathrm{E}[Z_{r}]-\mathrm{E}^{2}[Z_{r}]  &  \geq\mathrm{E}[Z_{r}%
^{2}]-\mathrm{E}^{2}[Z_{r}]\nonumber\\
p_{r}\left(  \boldsymbol{\theta}\right)  (1-p_{r}\left(  \boldsymbol{\theta
}\right)  )  &  \geq\mathrm{Var}[Z_{r}]. \label{var2}%
\end{align}
From (\ref{var2}), there exists $\rho_{r}^{2}\in\lbrack0,1]$ such that
$\mathrm{Var}[Z_{r}]=\rho_{r}^{2}p_{r}\left(  \boldsymbol{\theta}\right)
(1-p_{r}\left(  \boldsymbol{\theta}\right)  )$ and thus%
\begin{align*}
\mathrm{E}[Y_{r}^{\left(  \ell\right)  }]  &  =\mathrm{E}[\mathrm{E}%
[Y_{r}^{\left(  \ell\right)  }|Z_{r}]]=np_{r}\left(  \boldsymbol{\theta
}\right)  ,\\
\mathrm{Var}[Y_{r}^{\left(  \ell\right)  }]  &  =\mathrm{E}[\mathrm{Var}%
[Y_{r}^{\left(  \ell\right)  }|Z_{r}]]+\mathrm{Var}[\mathrm{E}[Y_{r}^{\left(
\ell\right)  }|Z_{r}]]=\vartheta_{n}^{(r)}np_{r}\left(  \boldsymbol{\theta
}\right)  (1-p_{r}\left(  \boldsymbol{\theta}\right)  ),
\end{align*}
where $\vartheta_{n}^{(r)}=1+\left(  n-1\right)  \rho_{r}^{2}$. Since the same
degree of overdispersion is assumed over the $M$ categories, it holds
$\rho_{1}^{2}=\cdots=\rho_{M}^{2}=\rho^{2}$, $\vartheta_{n}^{(1)}%
=\cdots=\vartheta_{n}^{(M)}=\vartheta_{n}$, and now%
\[
\mathrm{Var}[Y_{r}^{\left(  \ell\right)  }]=\vartheta_{n}np_{r}\left(
\boldsymbol{\theta}\right)  (1-p_{r}\left(  \boldsymbol{\theta}\right)
),\text{\quad}r=1,....,M,\text{\quad}\ell=1,...,N,
\]
match the diagonal elements of the inflated variance-covariance matrix given
in (\ref{moments2}).

Ann and James (1995) presented an algorithm for generating overdispersed
binomial distributions. Some examples of distributions for $\boldsymbol{Y}%
^{\left(  \ell\right)  }$, with expectation vector and variance-covariance
given in (\ref{moments2}), are the following: the Dirichlet-multinomial, the
random-clumped multinomial and $n$-inflated multinomial distributions. The
Dirichlet-multinomial distribution%
\begin{equation}
\Pr(Y_{1}^{\left(  \ell\right)  }=y_{1},...,Y_{M}^{\left(  \ell\right)
}=y_{M})=\binom{n}{y_{1}\cdots y_{M}}\frac{\Gamma(c)}{\Gamma(n+c)}\frac{%
{\textstyle\prod_{r=1}^{M}}
\Gamma(y_{r}+cp_{r}\left(  \boldsymbol{\theta}\right)  )}{%
{\textstyle\prod_{r=1}^{M}}
\Gamma(cp_{r}\left(  \boldsymbol{\theta}\right)  )}, \label{model1}%
\end{equation}
where $y_{s}\in\mathcal{%
\mathbb{Z}
}^{+}$, $\sum_{r=1}^{M}y_{r}=n$, $c=\rho^{-2}(1-\rho^{2})$, $\binom{n}%
{y_{1}\cdots y_{M}}=n!/%
{\textstyle\prod\nolimits_{r=1}^{M}}
y_{r}!$ and $\Gamma(\cdot)$ denotes the gamma function, is due to Mosimann (1962).

The random-clumped multinomial distribution%
\begin{equation}
\Pr(\boldsymbol{Y}^{\left(  \ell\right)  }=\boldsymbol{y})=\sum_{r=1}^{M}%
p_{r}\left(  \boldsymbol{\theta}\right)  \Pr(\boldsymbol{U}^{(r)}%
=\boldsymbol{y}),\label{model2}%
\end{equation}
where $\boldsymbol{y}=(y_{1},...,y_{M})^{T}$, $y_{r}\in\mathcal{%
\mathbb{Z}
}^{+}$, $\sum_{r=1}^{M}y_{r}=n$, $\boldsymbol{U}^{(r)}$, $r=1,...,M$ are
multinomial random vectors%
\[
\boldsymbol{U}^{(r)}\sim\mathcal{M}(n,(1-\rho)\boldsymbol{p}\left(
\boldsymbol{\theta}\right)  +\rho\boldsymbol{e}_{r}),\text{\quad}r=1,...,M,
\]
$\rho\in\lbrack0,1]$ and $\boldsymbol{e}_{r}$ is $r$-th the unit vector of
dimension $M$ ($1$ in the $r$-th position and the rest elements are zero), is
due to Morel and Nagaraj (1993). The $n$-inflated multinomial distribution%
\begin{equation}
\Pr(\boldsymbol{Y}^{\left(  \ell\right)  }=\boldsymbol{y})=(1-\rho^{2}%
)\Pr(\boldsymbol{U}=\boldsymbol{y})+\rho^{2}%
{\textstyle\sum_{r=1}^{M}}
I_{\left\{  n\right\}  }(y_{r})p_{r}\left(  \boldsymbol{\theta}\right)
,\label{model3}%
\end{equation}
where $\boldsymbol{y}=(y_{1},...,y_{M})^{T}$, $y_{r}\in\mathcal{%
\mathbb{Z}
}^{+}$, $\sum_{r=1}^{M}y_{r}=n$,%
\[
\boldsymbol{U}\sim\mathcal{M}(n,\boldsymbol{p}\left(  \boldsymbol{\theta
}\right)  ),
\]
is due to Cohen (1976) and Altham (1976). The multinomial distribution, with
zero inflation in the first $M-1$ cells, i.e. $n$-inflation in the $M$-th cell%
\begin{equation}
\Pr(\boldsymbol{Y}^{\left(  \ell\right)  }=\boldsymbol{y})=w\Pr(\boldsymbol{U}%
=\boldsymbol{y})+(1-w)I_{\left\{  n\right\}  }(y_{r})(n),\label{modelo4}%
\end{equation}
for any $w\in(0,1)$, cannot be considered in general as a distribution with
expectation vector and variance-covariance given in (\ref{moments2}), but does
satisfy both moments for the special case of $M=2$, i.e. for the zero-inflated
binomial distribution. The details are given in Section \ref{SecA.0} in the
Appendix. The advantage of not using the distributional assumption in the
model is that we can address the estimation problem in a similar way done for
the multinomial sampling, if we correct the estimator of the variance-variance
covariance matrix through an appropriate estimator of the design effect or
intracluster correlation coefficient. The consistency of both estimators is an
important property in order to make statistical inference about the goodness
of fit or other kind of hypothesis testing.

Throughout this paper we shall assume at the beginning, that all the
contingency tables $\boldsymbol{Y}^{\left(  \ell\right)  }$, $\ell=1,...,N$,
have a common sample size, $n$. This assumption is often violated (e.g., due
to missing values). The extension of the results from equal cluster sizes to
the unequal cluster sizes is not difficult, nevertheless, as we are aware,
even for the quasi-likelihood methodology, no paper has previously provided an
explicit expression for a consistent estimator of the design effect
($\vartheta_{n}$) or intracluster correlation coefficient ($\rho^{2}$). We
shall present this extension in Section \ref{sec4}.

The contents of this paper are organized as follows. For two-way contingency
tables with overdispersion, Brier (1980) analyzed the independence model but
using a parametrization different from the log-linear modeling given in
(\ref{In1}). An advantage of using the log-linear model parametrization is
that the estimation of the interaction parameter could provide some insight on
the appropriate log-linear model before considering the goodness-of-fit test.
In three-way contingency tables with overdispersion, the log-linear modeling
makes clearly simpler the statistical inference needed for model fitting.
Motivated by these facts, the second purpose of this paper is to present a new
family of estimators useful for log-linear modeling with overdispersion, under
the mild assumption that the distribution of the contingency tables\ is not
specified but it is suppose to hold (\ref{moments2}). These new estimators are
the quasi\ minimum divergence estimators. We shall refer them in Section
\ref{sec2}. One member of these estimators is the so-called quasi maximum
likelihood estimator. Their corresponding asymptotic properties are also
shown. We shall propose in Section \ref{sec5prev} new estimators for
$\vartheta_{n}$ and $\rho^{2}$. The assumption of equal cluster sizes is
generalized to unequal cluster sizes in Section \ref{sec4}. In this setting on
one hand, a new family of consistent estimators of the design effect or
intracluster correlation coefficient is provided, and on the other hand a new
estimator is proposed for the special case of large cluster sizes. Two
numerical examples illustrate the practical application of the new proposed
estimators in Section \ref{sec5} and a simulation study is presented in
Section \ref{sec6} by using distributions for the contingency tables, related
to (\ref{model1}), (\ref{model2}), (\ref{model3}). Finally, in Section
\ref{sec7} some concluding remarks are provided.

\section{Quasi minimum $\phi$-divergence estimator for log-linear models with
complex sampling\label{sec2}}

The nonparametric estimator of $\boldsymbol{p}\left(  \boldsymbol{\theta
}\right)  $ based on $N$ clusters is%
\begin{equation}
\widehat{\boldsymbol{p}}=\frac{1}{nN}\sum\limits_{\ell=1}^{N}\boldsymbol{Y}%
^{\left(  \ell\right)  }, \label{pHat}%
\end{equation}
i.e. $\widehat{\boldsymbol{p}}=(\widehat{p}_{1},...,\widehat{p}_{M})^{T}$,
with $\widehat{p}_{r}=\frac{1}{nN}\sum\limits_{\ell=1}^{N}Y_{r}^{\left(
\ell\right)  }$. This global estimator can also be expressed through the
average of the nonparametric estimators of $\boldsymbol{p}\left(
\boldsymbol{\theta}\right)  $, based on the $\ell$-th cluster,
$\widehat{\boldsymbol{p}}^{(\ell)}=\frac{1}{n}\boldsymbol{Y}^{\left(
\ell\right)  }$, $\ell=1,...,N$,
\[
\widehat{\boldsymbol{p}}=\frac{1}{N}\sum\limits_{\ell=1}^{N}%
\widehat{\boldsymbol{p}}^{(\ell)},
\]
i.e. $\widehat{\boldsymbol{p}}^{(\ell)}=(\widehat{p}_{1}^{(\ell)}%
,...,\widehat{p}_{M}^{(\ell)})^{T}$, with $\widehat{p}_{r}^{(\ell)}=\frac
{1}{n}Y_{r}^{(\ell)}$, $\ell=1,...,N$.

In the particular case of $\boldsymbol{Y}^{\left(  \ell\right)  }$ with
multinomial distribution ($\vartheta_{n}=1$), $\ell=1,...,N$, $\boldsymbol{Y}%
=\sum\limits_{\ell=1}^{N}\boldsymbol{Y}^{\left(  \ell\right)  }$ is also
multinomial, $\boldsymbol{Y}=(Y_{1},...,Y_{M})^{T}\boldsymbol{\sim}%
\mathcal{M}(nN,\boldsymbol{p}\left(  \boldsymbol{\theta}\right)  )$. Obtaining
the MLE of $\boldsymbol{\theta}$\ consists in maximizing%
\begin{equation}
\Pr\left(  Y_{1}=y_{1},...,Y_{M}=y_{M}\right)  =\binom{nN}{y_{1}\cdots y_{M}%
}p_{1}\left(  \boldsymbol{\theta}\right)  ^{y_{1}}\times\cdots\times
p_{M}\left(  \boldsymbol{\theta}\right)  ^{y_{M}} \label{mult}%
\end{equation}
or equivalently
\[
\log\Pr\left(  Y_{1}=y_{1},...,Y_{M}=y_{M}\right)  =-nNd_{Kullback}%
(\widehat{\boldsymbol{p}},\boldsymbol{p}\left(  \boldsymbol{\theta}\right)
)+k,
\]
where $k$ is a constant independent from the parameter $\boldsymbol{\theta,}$
$d_{Kullback}(\widehat{\boldsymbol{p}},\boldsymbol{p}\left(
\boldsymbol{\theta}\right)  )$ is the Kullback-Leibler divergence between the
probability vectors $\widehat{\boldsymbol{p}}$ and $\boldsymbol{p}\left(
\boldsymbol{\theta}\right)  $, i.e.,
\[
d_{Kullback}(\widehat{\boldsymbol{p}},\boldsymbol{p}\left(  \boldsymbol{\theta
}\right)  )=\sum\limits_{r=1}^{M}\widehat{p}_{r}\log\frac{\widehat{p}_{r}%
}{p_{r}\left(  \boldsymbol{\theta}\right)  }.
\]
Therefore the MLE of $\boldsymbol{\theta}$ for the multinomial model is given
by the value $\widehat{\boldsymbol{\theta}}=\widehat{\boldsymbol{\theta}%
}\left(  \boldsymbol{Y}\right)  $ such that
\begin{equation}
\widehat{\boldsymbol{\theta}}\left(  \boldsymbol{Y}\right)  =\arg\min
_{\theta\in\Theta}d_{Kullback}(\widehat{\boldsymbol{p}},\boldsymbol{p}\left(
\boldsymbol{\theta}\right)  ). \label{Kull}%
\end{equation}
Let $\phi$ a convex function $\phi\left(  x\right)  $, $x>0$, such that at
$x=1$,\ $\phi\left(  1\right)  =0$,\ $\phi^{\prime}\left(  1\right)  =0$,
$\phi^{\prime\prime}\left(  1\right)  >0$, at $x=0$,\ $0\phi\left(
0/0\right)  =0$ and $0\phi\left(  p/0\right)  =\underset{u\rightarrow
\infty}{\lim}p\phi\left(  u\right)  /u$. It is well-known that the
Kullback-Leibler divergence is a particular case of the so-called
phi-divergence measures between the probability vectors
$\widehat{\boldsymbol{p}}$ and $\boldsymbol{p}\left(  \boldsymbol{\theta
}\right)  $, given by%
\begin{equation}
d_{\phi}\left(  \widehat{\boldsymbol{p}},\boldsymbol{p}\left(
\boldsymbol{\theta}\right)  \right)  =\sum\limits_{r=1}^{M}p_{r}\left(
\boldsymbol{\theta}\right)  \phi\!\left(  \frac{\widehat{p}_{r}}{p_{r}\left(
\boldsymbol{\theta}\right)  }\right)  . \label{6}%
\end{equation}
More thoroughly, taking
\begin{equation}
\phi(x)=x\log x-x+1, \label{phiKull}%
\end{equation}
(\ref{6}) is the Kullback-Leibler divergence between $\widehat{\boldsymbol{p}%
}$ and $\boldsymbol{p}\left(  \boldsymbol{\theta}\right)  $. For more details
about phi-divergence measures see Pardo (2006). The phi-divergence based
estimators for multinomial log-linear models is not new, see for instance
Cressie and Pardo (2000, 2003), Cressie et al. (2003), Mart\'{\i}n and Pardo
(2008a, 2008b, 2010, 2011, 2012).

When $\rho^{2}=0$, notice that a unique contingency table associated with one
cluster ($N=1$), $\boldsymbol{Y=\boldsymbol{Y}}^{(1)}\boldsymbol{\sim
}\mathcal{M}(n,\boldsymbol{p}\left(  \boldsymbol{\theta}\right)  )$, is enough
for having a suitable sample for making asymptotic statistical inference,
since the Weak Law of Large Numbers and the Central Limit Theorem can be
applied directly inside the unique cluster, making the number independent
observation inside, $n$, large enough. Nevertheless, when $\rho^{2}>0$, $n$ is
fixed and $N$ must be large enough. The following definition shows that the
phi-divergences, given in (\ref{6}), are also useful for the general case of
$\rho^{2}\geq0$. Unlike the multinomial sampling ($\rho^{2}=0$), for clustered
multinomial log-linear models ($\rho^{2}>0$) the knowledge of the shape of the
moments, given in (\ref{moments2}), is only assumed. Since no underlying
distribution is being assumed and only mild assumptions on the first two
moments of a distribution are taken into account, the estimator of
$\boldsymbol{\theta}$ is termed \textquotedblleft quasi minimum $\phi
$-divergence estimator\textquotedblright\ of $\boldsymbol{\theta}$ (in the
sequel, QM$\phi$E), defined for the first time for a more general setting in
Vos (1992).\bigskip

\begin{definition}
\label{DefMin}We consider a statistical model verifying (\ref{moments2}). The
QM$\phi$E of $\boldsymbol{\theta}$, $\widehat{\boldsymbol{\theta}}_{\phi
}=\widehat{\boldsymbol{\theta}}_{\phi}\left(  \boldsymbol{Y}\right)  $, is
defined as
\begin{equation}
\widehat{\boldsymbol{\theta}}_{\phi}\left(  \boldsymbol{Y}\right)  =\arg
\min_{\theta\in\Theta}d_{\phi}(\widehat{\boldsymbol{p}},\boldsymbol{p}\left(
\boldsymbol{\theta}\right)  ). \label{70}%
\end{equation}
where $d_{\phi}(\widehat{\boldsymbol{p}},\boldsymbol{p}\left(
\boldsymbol{\theta}\right)  )$, the phi-divergence measure between the
probability vectors $\widehat{\boldsymbol{p}}$ and $\boldsymbol{p}\left(
\boldsymbol{\theta}\right)  $,\ is given by (\ref{6}).
\end{definition}

From a practical point of view, in order to find the quasi minimum $\phi$
divergence estimator of $\boldsymbol{\theta}$ for clustered multinomial
log-linear models, we have to solve the following system of equations%
\begin{equation}
\boldsymbol{W}^{T}\boldsymbol{\Sigma}_{p(\boldsymbol{\theta})}\boldsymbol{D}%
_{\boldsymbol{p}(\theta)}^{-1}\boldsymbol{\Psi}^{\phi}(\boldsymbol{\theta
})=\boldsymbol{0}_{M_{0}}, \label{eq}%
\end{equation}
where $\boldsymbol{\Sigma}_{p(\boldsymbol{\theta})}\boldsymbol{D}%
_{\boldsymbol{p}(\theta)}^{-1}=\boldsymbol{I}_{M}-\boldsymbol{p}\left(
\boldsymbol{\theta}\right)  \boldsymbol{1}_{M}^{T}$,%
\begin{align*}
\boldsymbol{\Psi}^{\phi}(\boldsymbol{\theta})  &  =(\Psi_{1}^{\phi
}(\boldsymbol{\theta}),...,\Psi_{M}^{\phi}(\boldsymbol{\theta}))^{T},\\
\Psi_{r}^{\phi}(\boldsymbol{\theta})  &  =\widehat{p}_{r}\phi^{\prime
}\!\left(  \frac{\widehat{p}_{r}}{p_{r}\left(  \boldsymbol{\theta}\right)
}\right)  -p_{r}\left(  \boldsymbol{\theta}\right)  \phi\!\left(
\frac{\widehat{p}_{r}}{p_{r}\left(  \boldsymbol{\theta}\right)  }\right)
,\quad r=1,...,M.
\end{align*}
This expression arises from considering%
\[
\dfrac{\partial}{\partial\theta_{g}}d_{\phi}\left(  \widehat{\boldsymbol{p}%
},\boldsymbol{p}\left(  \boldsymbol{\theta}\right)  \right)  =-\sum
\limits_{r=1}^{M}\dfrac{\partial p_{r}(\boldsymbol{\theta})}{\partial
\theta_{g}}\frac{\Psi_{r}^{\phi}(\boldsymbol{\theta})}{p_{r}%
(\boldsymbol{\theta})},\quad g=1,...,M_{0}.
\]
These equations are nonlinear functions of the unknown parameter,
$\boldsymbol{\theta}$. In order to solve these equations numerically the
Newton-Raphson method is used, in such a way that the $(t+1)$th-step estimate,
$\widehat{\boldsymbol{\theta}}_{\phi}^{\left(  t+1\right)  }$, is obtained
from $\widehat{\boldsymbol{\theta}}_{\phi}^{\left(  t\right)  }$ as
\[
\widehat{\boldsymbol{\theta}}_{\phi}^{\left(  t+1\right)  }%
=\widehat{\boldsymbol{\theta}}_{\phi}^{\left(  t\right)  }-\boldsymbol{G}%
_{\phi}^{-1}(\widehat{\boldsymbol{\theta}}_{\phi}^{\left(  t\right)
})\boldsymbol{W}^{T}\boldsymbol{\Sigma}_{p(\widehat{\boldsymbol{\theta}}%
_{\phi}^{\left(  t\right)  })}\boldsymbol{D}_{\boldsymbol{p}%
(\widehat{\boldsymbol{\theta}}_{\phi}^{\left(  t\right)  })}^{-1}%
\boldsymbol{\Psi}^{\phi}(\widehat{\boldsymbol{\theta}}_{\phi}^{\left(
t\right)  }),
\]
where
\begin{align*}
\boldsymbol{G}_{\phi}(\boldsymbol{\theta})  &  =\left(  G_{\phi,g,h}%
(\boldsymbol{\theta})\right)  _{g,h=1,...,M_{0}}=\boldsymbol{W}^{T}%
\boldsymbol{\Sigma}_{p(\boldsymbol{\theta})}\boldsymbol{D}%
_{p(\boldsymbol{\theta})}^{-1}\boldsymbol{D}_{\overline{\boldsymbol{\Psi}%
}^{\phi}(\boldsymbol{\theta})}\boldsymbol{D}_{p(\boldsymbol{\theta})}%
^{-1}\boldsymbol{\Sigma}_{p(\boldsymbol{\theta})}\boldsymbol{W}\\
&  =\boldsymbol{W}^{T}\left(  \boldsymbol{I}_{M}-\boldsymbol{p}%
(\boldsymbol{\theta})\boldsymbol{1}_{M}^{T}\right)  \boldsymbol{D}%
_{\overline{\boldsymbol{\Psi}}^{\phi}(\boldsymbol{\theta})}\left(
\boldsymbol{I}_{M}-\boldsymbol{1}_{M}\boldsymbol{p}^{T}(\boldsymbol{\theta
})\right)  \boldsymbol{W},
\end{align*}
with%
\begin{align*}
G_{\phi,g,h}(\boldsymbol{\theta})  &  =\dfrac{\partial^{2}}{\partial\theta
_{g}\partial\theta_{h}}d_{\phi}\left(  \widehat{\boldsymbol{p}},\boldsymbol{p}%
\left(  \boldsymbol{\theta}\right)  \right) \\
&  =\sum\limits_{r=1}^{M}{\phi}^{\prime\prime}\!\left(  \frac{\widehat{p}_{r}%
}{p_{r}(\boldsymbol{\theta})}\right)  \left(  \frac{\widehat{p}_{r}}%
{p_{r}(\boldsymbol{\theta})}\right)  ^{2}\dfrac{1}{p_{r}(\mathbf{\theta}%
)}{\dfrac{\partial p_{r}(\boldsymbol{\theta})}{\partial\theta_{g}}%
\dfrac{\partial p_{r}(\boldsymbol{\theta})}{\partial\theta_{h}}}%
-\sum\limits_{r=1}^{M}\dfrac{\partial^{2}p_{r}(\boldsymbol{\theta})}%
{\partial\theta_{g}\partial\theta_{h}}\frac{\Psi^{\phi}(\boldsymbol{\theta}%
)}{p_{r}(\boldsymbol{\theta})}\\
&  =\sum\limits_{r=1}^{M}\overline{\Psi}_{r}^{\phi}(\boldsymbol{\theta}%
)\frac{1}{p_{r}^{2}(\boldsymbol{\theta})}{\dfrac{\partial p_{r}%
(\boldsymbol{\theta})}{\partial\theta_{g}}\dfrac{\partial p_{r}%
(\boldsymbol{\theta})}{\partial\theta_{h}}},
\end{align*}%
\begin{align*}
\overline{\boldsymbol{\Psi}}^{\phi}(\boldsymbol{\theta})  &  =(\overline{\Psi
}_{1}^{\phi}(\boldsymbol{\theta}),...,\overline{\Psi}_{M}^{\phi}%
(\boldsymbol{\theta}))^{T},\\
\overline{\Psi}_{r}^{\phi}(\boldsymbol{\theta})  &  ={\phi}^{\prime\prime
}\!\left(  \dfrac{\widehat{p}_{r}}{p_{r}(\boldsymbol{\theta})}\right)
\dfrac{\widehat{p}_{r}^{2}}{p_{r}(\boldsymbol{\theta})}-\Psi_{r}^{\phi
}(\boldsymbol{\theta}),\quad r=1,...,M.
\end{align*}
It is worthwhile of mentioning that the quasi maximum\ likelihood estimators
(QMLE), introduced by Wedderburn (1974), are very useful for clustered
multinomial models, in particular for multinomial log-linear models. The QMLEs
of $\boldsymbol{\theta}$ are obtained by solving the system of non-linear
equations (\ref{eq}) with ${\phi}$ given by (\ref{phiKull}), i.e.,
\[
\boldsymbol{\Psi}^{\phi}(\boldsymbol{\theta})=\widehat{\boldsymbol{p}%
}-\boldsymbol{p}\left(  \boldsymbol{\theta}\right)  ,
\]
and the expression of these estimators match the ones of the minimum
Kullback-Leibler divergence estimators given in (\ref{Kull}). Hence, the
QM$\phi$Es are generalizations of the QMLEs, for clustered multinomial
log-linear models.

Under the assumption that the $QM\phi E$ exists, Newton-Raphson method tends
always to converge with any initialization or guess of parameters. However,
for some samples, such as contingency tables with outlying cells or
high-dimensional parameters, $d_{\phi}(\widehat{\boldsymbol{p}},\boldsymbol{p}%
(\boldsymbol{\theta}))$, can be quite flat and can create troubles related to
precision, even with no null frequencies. Such problems are related to Hessian
matrices near to be singular. In addition, if the $QM\phi E$ fails to exist,
the algorithm does not converge, since the Hessian matrix of $d_{\phi
}(\widehat{\boldsymbol{p}},\boldsymbol{p}(\boldsymbol{\theta}))$ becomes
iteratively to being singular. The Newton-Raphson method converges in general
in few iterations, but is convenient to begin properly the iterative process
with the weighted least squares method of fitting log-linear models (Grizzle,
Stramer and Koch (1969)), i.e. taking
\[%
\begin{pmatrix}
\widehat{u}\\
\widehat{\boldsymbol{\theta}}%
\end{pmatrix}
=\left(  \boldsymbol{X}^{T}\boldsymbol{D}_{\widehat{\boldsymbol{p}}%
}\boldsymbol{X}\right)  ^{-1}\boldsymbol{X}^{T}\boldsymbol{D}%
_{\widehat{\boldsymbol{p}}}\log\widehat{\boldsymbol{p}}%
\]
with $\boldsymbol{X}=(\boldsymbol{1}_{M},\boldsymbol{W})$ and removing from it
$\widehat{u}$, the independent term. A count of $\widehat{p}_{r}=0$ is
problematic, so one could set $\widehat{p}_{r}=\frac{1}{2}$.

In order to ensure an algorithm with proper convergence properties in both
cases at the same time, with zero frequencies or not, there are several
possibilities. Quasi-Newton and conjugate gradient methods require only
evaluating gradients, being quasi-Newton methods faster but more storage
demanding. For this reason, we have applied in the simulation study the
Fortran \texttt{NAG} subroutine \texttt{C05PBF}, which is based on a modified
version of the Powell's hybrid algorithm, a combination of quasi-Newton and
conjugate gradient methods. It is also worth of mentioning that the derivative
free algorithms (Nelder-Mead, Hooke-Jeeves, Torczon) constitute a robust
choice with respect to the initial point, and could be particularly useful for
contingency tables with outlying cells.

\begin{theorem}
\label{Th1}Let $\widehat{\boldsymbol{\theta}}_{\phi}$ be the QM$\phi$E for the
unknown parameter $\boldsymbol{\theta}$ of the clustered multinomial
log-linear models, then it holds

\begin{enumerate}
\item[i)]
\begin{equation}
\sqrt{N}(\widehat{\boldsymbol{\theta}}_{\phi}-\boldsymbol{\theta}%
_{0})\overset{\mathcal{L}}{\underset{N\rightarrow\infty}{\longrightarrow}%
}\mathcal{N}(\boldsymbol{0}_{M_{0}},\tfrac{\vartheta_{n}}{n}\left(
\boldsymbol{\boldsymbol{W}}^{T}\boldsymbol{\Sigma\boldsymbol{_{\boldsymbol{p}%
\left(  \theta_{0}\right)  }}W}\right)  ^{-1}),\label{9}%
\end{equation}

\item[ii)]
\begin{equation}
\sqrt{N}(\boldsymbol{p}(\widehat{\boldsymbol{\theta}}_{\phi})-\boldsymbol{p}%
(\boldsymbol{\theta}_{0}))\overset{\mathcal{L}}{\underset{N\rightarrow
\infty}{\longrightarrow}}\mathcal{N}(\boldsymbol{0}_{M},\tfrac{\vartheta_{n}%
}{n}\boldsymbol{\Sigma\boldsymbol{_{\boldsymbol{p}\left(  \theta_{0}\right)
}}W}\left(  \boldsymbol{\boldsymbol{W}}^{T}\boldsymbol{\Sigma
\boldsymbol{_{\boldsymbol{p}\left(  \theta_{0}\right)  }}W}\right)
^{-1}\boldsymbol{W}^{T}\boldsymbol{\Sigma}_{p\left(  \boldsymbol{\theta}%
_{0}\right)  }),\label{10}%
\end{equation}
with $\boldsymbol{\theta}_{0}$\ being the true and unknown value of
$\boldsymbol{\theta}$.
\end{enumerate}
\end{theorem}

\begin{proof}
The proof is given in Section \ref{SecA.2} of the Appendix.
\end{proof}

\section{Consistent estimator for $\vartheta_{n}$ and $\rho^{2}$%
\label{sec5prev}}

We consider $\widehat{\boldsymbol{p}}$ and $\boldsymbol{p}\left(
\boldsymbol{\theta}\right)  $ defined in (\ref{pHat}) and (\ref{model})
respectively. By the Weak Law of Large Numbers, it holds%
\[
\widehat{\boldsymbol{p}}\overset{P}{\underset{N\rightarrow\infty
}{\longrightarrow}}\boldsymbol{p}\left(  \boldsymbol{\theta}_{0}\right)  ,
\]
and applying the Central Limit Theorem, it follows that
\begin{equation}
\sqrt{N}\left(  \widehat{\boldsymbol{p}}-\boldsymbol{p}\left(
\boldsymbol{\theta}_{0}\right)  \right)  \overset{\mathcal{L}%
}{\underset{N\rightarrow\infty}{\longrightarrow}}\mathcal{N(}\boldsymbol{0}%
_{M},\tfrac{\vartheta_{n}}{n}\boldsymbol{\Sigma}_{\boldsymbol{p}\left(
\boldsymbol{\theta}_{0}\right)  }), \label{CLT}%
\end{equation}
where $\boldsymbol{\Sigma}_{\boldsymbol{p}\left(  \boldsymbol{\theta}\right)
}$ was given in (\ref{var}).

\begin{remark}
Notice that $\tfrac{\vartheta_{n}}{n}=\frac{1}{n}+\frac{n-1}{n}\rho^{2}$ is an
increasing function of the intracluster correlation, $\rho^{2}$: with
$\rho^{2}=\frac{k-1}{n-1}$ we obtain $\tfrac{\vartheta_{n}}{n}=\frac{k}{n}$,
for $k\in\{2,...,n\}$ and with $\frac{k-2}{n-1}<\rho^{2}<\frac{k-1}{n-1}$ we
obtain $\frac{k-1}{n}<\tfrac{\vartheta_{n}}{n}<\frac{k}{n}$, for
$k\in\{2,...,n\}$. On the other hand, if the cluster size ($n$) were large,
$\frac{1}{n}(1-\rho^{2})$\ would be small, and $\tfrac{\vartheta_{n}}{n}%
=\frac{1}{n}(1-\rho^{2})+\rho^{2}$ would become similar to $\rho^{2}$.
\end{remark}

Now, we shall consider the $N$ contingency tables expressed jointly in a
unique $NM$-dimensional vector,%
\[
\widetilde{\boldsymbol{Y}}=(\boldsymbol{Y}^{\left(  1\right)  T}%
,...,\boldsymbol{Y}^{\left(  N\right)  T})^{T},
\]
and we can define its corresponding vector of probabilities,
$\widetilde{\boldsymbol{p}}$, as follows%
\[
\widetilde{\boldsymbol{p}}=\frac{1}{nN}\widetilde{\boldsymbol{Y}}.
\]
In addition, the inter-cluster-level homogeneous version of the probability
vector is given by%
\begin{equation}
\widetilde{\boldsymbol{p}}^{\ast}=(\tfrac{1}{N}\widehat{\boldsymbol{p}}%
^{T},...,\tfrac{1}{N}\widehat{\boldsymbol{p}}^{T})^{T}. \label{pTilde*}%
\end{equation}
Brier (1980) proposed a consistent estimator of $\vartheta_{n}$\ based on
comparing the discrepancy between $\widetilde{\boldsymbol{p}}$\ and
$\widetilde{\boldsymbol{p}}^{\ast}$ in the following way%
\begin{equation}
X^{2}(\widetilde{\boldsymbol{Y}})=\sum_{\ell=1}^{N}\left(  \boldsymbol{Y}%
^{(\ell)}-n\widehat{\boldsymbol{p}}\right)  ^{T}\tfrac{1}{n}\boldsymbol{D}%
_{\widehat{\boldsymbol{p}}}^{-1}\left(  \boldsymbol{Y}^{(\ell)}%
-n\widehat{\boldsymbol{p}}\right)  =n\sum_{\ell=1}^{N}\sum_{r=1}^{M}%
\frac{(\widehat{p}_{r}^{(\ell)}-\widehat{p}_{r})^{2}}{\widehat{p}_{r}},
\label{chi}%
\end{equation}
with $\boldsymbol{D}_{\boldsymbol{a}}$ being the diagonal matrix of vector
$\boldsymbol{a}$. The shape of this estimator reminds the expression of the
chi-square test-statistic for inter-cluster level homogeneity.

The following theorem permit us to define estimators for $\vartheta_{n}$\ and
$\rho^{2}$ through the same expression proposed by Brier (1980). Nevertheless,
these estimators are valid not only for the Dirichlet-multinomial distribution
given in (\ref{model1}), as desired by Brier, but also for other distributions
with overdispersion such as (\ref{model2}) and (\ref{model3}). For this reason
we refer them as the Brier's estimators.

\begin{theorem}
\label{Dist}For (\ref{chi}) divided by $(N-1)(M-1)$, as $N$ tends to infinity,
it holds
\begin{equation}
\frac{X^{2}(\widetilde{\boldsymbol{Y}})}{(N-1)(M-1)}%
\overset{P}{\underset{N\rightarrow\infty}{\longrightarrow}}\vartheta
_{n},\qquad\frac{\frac{X^{2}(\widetilde{\boldsymbol{Y}})}{(N-1)(M-1)}-1}%
{n-1}\overset{P}{\underset{N\rightarrow\infty}{\longrightarrow}}\rho^{2}.
\label{cons2}%
\end{equation}

\end{theorem}

\begin{proof}
The proof is given in Section \ref{SecA.1} of the Appendix.
\end{proof}

Since in this paper no specific distribution is assumed, the proof of Theorem
\ref{Dist} is completely new and more general than the one given in Brier
(1980) and is the basis for considering the second of the following consistent
estimators, for the design effect as well as the intracluster correlation coefficient.

\begin{definition}
[Nonparametric estimators of
$\vartheta$ and $\rho^{2}$%
]\label{Def0}The Brier's consistent estimator of the design effect,
$\vartheta_{n}$, is%
\begin{equation}
\widetilde{\vartheta}_{n,N}(\widetilde{\boldsymbol{Y}})=\frac{X^{2}%
(\widetilde{\boldsymbol{Y}})}{(N-1)(M-1)}, \label{psi}%
\end{equation}
where $X^{2}(\widetilde{\boldsymbol{Y}})$ is defined in (\ref{chi}).
Similarly, the the Brier's consistent estimator of the intracluster
correlation coefficient, $\rho^{2}$, is%
\begin{equation}
\widetilde{\rho}_{n,N}^{2}(\widetilde{\boldsymbol{Y}})=\frac
{\widetilde{\vartheta}_{n,N}(\widetilde{\boldsymbol{Y}})-1}{n-1}. \label{rro}%
\end{equation}

\end{definition}

The estimator for the design effect, $\widetilde{\vartheta}_{n,N}%
(\widetilde{\boldsymbol{Y}})$, as well as for the intracluster correlation
coefficient, $\widetilde{\rho}_{n,N}^{2}(\widetilde{\boldsymbol{Y}})$ are
fully non-parametric. Based on the proof of Theorem \ref{Dist} it is possible
to give the following definition based on the consistent estimator
$\boldsymbol{p}(\widehat{\boldsymbol{\theta}}_{\phi})$ of $\boldsymbol{p}%
(\boldsymbol{\theta})$\ for a log-linear model with complex sampling, with
$\widehat{\boldsymbol{\theta}}_{\phi}$ being the QM$\phi$E\ given in
(\ref{70}). This could be a semi-parametric version of the estimator, and is
proposed for the first time in this paper.

\begin{definition}
[Semiparametric estimators of
$\vartheta$ and $\rho^{2}$%
]\label{Def01}The parametric extension of the Brier's consistent estimator of
$\vartheta_{n}$ is%
\[
\widetilde{\vartheta}_{n,N}(\widetilde{\boldsymbol{Y}}%
,\widehat{\boldsymbol{\theta}}_{\phi})=\frac{X^{2}(\widetilde{\boldsymbol{Y}%
},\widehat{\boldsymbol{\theta}}_{\phi})}{(N-1)(M-1)},
\]
where%
\[
X^{2}(\widetilde{\boldsymbol{Y}},\widehat{\boldsymbol{\theta}}_{\phi}%
)=\sum_{\ell=1}^{N}\left(  \boldsymbol{Y}^{(\ell)}-n\widehat{\boldsymbol{p}%
}\right)  ^{T}\tfrac{1}{n}\boldsymbol{D}_{\boldsymbol{p}%
(\widehat{\boldsymbol{\theta}}_{\phi})}^{-1}\left(  \boldsymbol{Y}^{(\ell
)}-n\widehat{\boldsymbol{p}}\right)  =n\sum_{\ell=1}^{N}\sum_{r=1}^{M}%
\frac{(\widehat{p}_{r}^{(\ell)}-\widehat{p}_{r})^{2}}{p_{r}%
(\widehat{\boldsymbol{\theta}}_{\phi})}.
\]
Similarly, the parametric extension of the Brier's consistent estimator of
$\rho^{2}$ is%
\[
\widetilde{\rho}_{n,N}^{2}(\widetilde{\boldsymbol{Y}}%
,\widehat{\boldsymbol{\theta}}_{\phi})=\frac{\widetilde{\vartheta}%
_{n,N}(\widetilde{\boldsymbol{Y}},\widehat{\boldsymbol{\theta}}_{\phi}%
)-1}{n-1}.
\]

\end{definition}

\section{Generalization for unequal cluster sizes\label{sec4}}

\subsection{Notation and basic results\label{sec4.1}}

Let us consider $G$ groups of clusters in such a way that all the contingency
tables,%
\[
\boldsymbol{Y}^{\left(  g,\ell\right)  }=(Y_{1}^{\left(  g,\ell\right)
},...,Y_{M}^{\left(  g,\ell\right)  })^{T},\quad\ell=1,...,N_{g},
\]
of the same group of clusters have the same sample size $n_{g}$, $g=1,...,G$,
and $N=%
{\textstyle\sum\nolimits_{g=1}^{G}}
N_{g}$. It is assumed having at least an index $g$ such that $n_{g}>1$. If we
replace the assumption $N\rightarrow\infty$ by $N_{g}\rightarrow\infty$ for
each group of clusters, then all above stated results hold separately for each
group of clusters.

By following (\ref{pHat}), the nonparametric estimator of $\boldsymbol{p}%
\left(  \boldsymbol{\theta}\right)  $, based on $N_{g}$ clusters, is now given
by%
\[
\widehat{\boldsymbol{p}}^{(g)}=\frac{1}{n_{g}N_{g}}\sum\limits_{\ell=1}%
^{N_{g}}\boldsymbol{Y}^{\left(  g,\ell\right)  },
\]
i.e. $\widehat{\boldsymbol{p}}^{(g)}=(\widehat{p}_{1}^{(g)},...,\widehat{p}%
_{M}^{(g)})^{T}$, with $\widehat{p}_{r}^{(g)}=\frac{1}{n_{g}N_{g}}%
\sum\limits_{\ell=1}^{N_{g}}Y_{r}^{\left(  g,\ell\right)  }$, $r=1,...,M$.
This global estimator can be also expressed through the average of the
nonparametric estimators of $\boldsymbol{p}\left(  \boldsymbol{\theta}\right)
$, based on the $\ell$-th cluster, $\widehat{\boldsymbol{p}}^{(g,\ell
)}=(\widehat{p}_{1}^{(g,\ell)},...,\widehat{p}_{M}^{(g,\ell)})^{T}=\frac
{1}{n_{g}}\boldsymbol{Y}^{\left(  g,\ell\right)  }$, $\ell=1,...,N_{g}$, as%
\begin{align*}
\widehat{\boldsymbol{p}}^{\left(  g\right)  } &  =\frac{1}{N_{g}}%
\sum\limits_{\ell=1}^{N_{g}}\widehat{\boldsymbol{p}}^{(g,\ell)},\\
\widehat{p}_{r}^{(g,\ell)} &  =\frac{1}{n_{g}}Y_{r}^{(g,\ell)}.
\end{align*}
On the other hand, the nonparametric estimator of $\boldsymbol{p}\left(
\boldsymbol{\theta}\right)  $, based on $G$ groups of $N_{1}$, ..., $N_{G}%
$\ clusters with sample size $n_{1}$, ..., $n_{G}$\ respectively, is now given
by%
\begin{equation}
\widehat{\boldsymbol{p}}=\frac{\sum\limits_{g=1}^{G}\sum\limits_{\ell
=1}^{N_{g}}\boldsymbol{Y}^{\left(  g,\ell\right)  }}{\sum\limits_{g=1}%
^{G}n_{g}N_{g}}=\frac{\sum\limits_{g=1}^{G}n_{g}N_{g}\frac{1}{N_{g}}%
\sum\limits_{\ell=1}^{N_{g}}\frac{1}{n_{g}}\boldsymbol{Y}^{\left(
g,\ell\right)  }}{\sum\limits_{g=1}^{G}n_{g}N_{g}}=\sum\limits_{g=1}^{G}%
w_{g}\widehat{\boldsymbol{p}}^{\left(  g\right)  },\label{forp}%
\end{equation}
where%
\begin{equation}
w_{g}=\frac{n_{g}N_{g}}{\sum\limits_{h=1}^{G}n_{h}N_{h}}>0,\quad
g=1,...,G,\label{w}%
\end{equation}
and $\sum\limits_{g=1}^{G}w_{g}=1$.

Through the Central Limit Theorem, similarly to (\ref{CLT}), for the $g$-th
group, it follows that
\begin{equation}
\sqrt{N_{g}}(\widehat{\boldsymbol{p}}^{(g)}-\boldsymbol{p}\left(
\boldsymbol{\theta}_{0}\right)  )\overset{\mathcal{L}}{\underset{N_{g}%
\rightarrow\infty}{\longrightarrow}}\mathcal{N(}\boldsymbol{0}_{M}%
,\tfrac{\vartheta_{n_{g}}}{n_{g}}\boldsymbol{\Sigma}_{\boldsymbol{p}\left(
\boldsymbol{\theta}_{0}\right)  }),\label{TCL0}%
\end{equation}
and thus%
\begin{equation}
\frac{\sum\limits_{h=1}^{G}n_{h}N_{h}}{\sqrt{\sum\limits_{g=1}^{G}n_{g}%
N_{g}\vartheta_{n_{g}}}}(\widehat{\boldsymbol{p}}\boldsymbol{-p}%
(\boldsymbol{\theta}_{0}))\overset{\mathcal{L}}{\underset{N_{1},...,N_{G}%
\rightarrow\infty}{\longrightarrow}}\mathcal{N}(\boldsymbol{0}_{M}%
,\boldsymbol{\Sigma}_{\boldsymbol{p}\left(  \boldsymbol{\theta}_{0}\right)
}).\label{CLT2}%
\end{equation}
See Section \ref{SecA.3} in the Appendix for the details of the derivation of
(\ref{CLT2}). If in addition, if we assume that there exists a sequence
$\{N_{h}^{\ast}\}_{h=1}^{G}$, such that%
\[
\frac{N_{h}}{N}\overset{P}{\underset{N\rightarrow\infty}{\longrightarrow}%
}N_{h}^{\ast}\in(0,1],
\]
(\ref{CLT2}) can be rewritten as%
\begin{equation}
\sqrt{N}(\widehat{\boldsymbol{p}}\boldsymbol{-p}(\boldsymbol{\theta}%
_{0}))\overset{\mathcal{L}}{\underset{N\rightarrow\infty}{\longrightarrow}%
}\mathcal{N}(\boldsymbol{0}_{M},\tfrac{\vartheta_{n^{\ast}}}{\bar{n}%
}\boldsymbol{\Sigma}_{\boldsymbol{p}\left(  \boldsymbol{\theta}_{0}\right)
}),\label{CLT2B}%
\end{equation}
where%
\begin{equation}
\bar{n}=\sum\limits_{g=1}^{G}N_{g}^{\ast}n_{g},\label{nMean}%
\end{equation}%
\begin{equation}
\vartheta_{n^{\ast}}=\sum\limits_{g=1}^{G}w_{g}^{\ast}\vartheta_{n_{g}%
},\label{overp}%
\end{equation}
and%
\[
w_{g}^{\ast}=\frac{N_{g}^{\ast}n_{g}}{\sum\limits_{h=1}^{G}N_{h}^{\ast}n_{h}%
}>0,\quad g=1,...,G,
\]
such that%
\[
w_{g}\overset{P}{\underset{N\rightarrow\infty}{\longrightarrow}}w_{g}^{\ast
},\quad%
{\textstyle\sum\nolimits_{g=1}^{G}}
w_{g}^{\ast}=1.
\]
Notice that%
\begin{equation}
\vartheta_{n^{\ast}}=1+\rho^{2}\left(  n^{\ast}-1\right)  \in(1,n^{\ast
}],\label{ro}%
\end{equation}
i.e. (\ref{overp}) represents the overdispersion parameter when the cluster
size is%
\begin{equation}
n^{\ast}=%
{\textstyle\sum\nolimits_{g=1}^{G}}
w_{g}^{\ast}n_{g}.\label{n2}%
\end{equation}
In particular, $\vartheta_{n^{\ast}}=1$ ($\rho^{2}=0$ or $n_{1}=\cdots
=n_{G}=1$) represents the case of multinomial sampling.

It is interesting to be mentioned that Brier (1980, Section 3.4) proposed the
unknown parameter $\vartheta_{n^{\ast}}$, given in (\ref{overp}), for the
stronger assumption of Dirichlet-multinomial distribution for $\boldsymbol{Y}%
^{(g,\ell)}$, given in (\ref{model1}). For this reason, in a future work, a
new improved consistent estimator of $\vartheta_{n^{\ast}}$ could be a useful
tool to propose appropriate test-statistics for the goodness-of-fit of
log-linear models with clustered multinomial data under overdispersion. These
test-statistics would require a weaker assumption in comparison with the
Brier's paper.

\subsection{Brier's modified estimators for $\vartheta_{n}$ and $\rho^{2}$}

We shall define a consistent estimator of the design effect and the
intracluster correlation coefficient, for unequal cluster sizes, as%
\begin{align}
\widetilde{\vartheta}_{\widehat{n}^{\ast},N} &  =\sum\limits_{g=1}^{G}%
w_{g}\widetilde{\vartheta}_{n_{g},N_{g}}(\widetilde{\boldsymbol{Y}}%
_{g}),\label{estOver}\\
\widetilde{\rho}_{\widehat{n}^{\ast},N}^{2} &  =\frac{\widetilde{\vartheta
}_{\widehat{n}^{\ast},N}-1}{\widehat{n}^{\ast}-1},\label{estRho}%
\end{align}
where%
\[
\widehat{n}^{\ast}=%
{\textstyle\sum\nolimits_{g=1}^{G}}
w_{g}n_{g},
\]
is a consistent estimator of $n^{\ast}$ given in (\ref{overp}) or (\ref{ro})
and%
\begin{align*}
\widetilde{\boldsymbol{Y}}_{g} &  =((\boldsymbol{Y}^{\left(  g,1\right)
})^{T},...,(\boldsymbol{Y}^{\left(  g,N_{g}\right)  })^{T})^{T},\\
\widetilde{\vartheta}_{n_{g},N_{g}}(\widetilde{\boldsymbol{Y}}_{g}) &
=\frac{X^{2}(\widetilde{\boldsymbol{Y}}_{g})}{(N_{g}-1)(M-1)},\\
X^{2}(\widetilde{\boldsymbol{Y}}_{g}) &  =n_{g}\sum_{\ell=1}^{N_{g}}\sum
_{r=1}^{M}\frac{(\widehat{p}_{r}^{(\ell,g)}-\widehat{p}_{r}^{(g)})^{2}%
}{\widehat{p}_{r}^{(g)}},
\end{align*}
$g=1,...,G$. Both estimators, (\ref{estOver}) and (\ref{estRho}), are
consistent estimators since
\[
\widetilde{\vartheta}_{\widehat{n}^{\ast},N}\overset{P}{\underset{N\rightarrow
\infty}{\longrightarrow}}\vartheta_{n^{\ast}},\qquad\widetilde{\rho
}_{\widehat{n}^{\ast},N}^{2}\overset{P}{\underset{N\rightarrow\infty
}{\longrightarrow}}\rho^{2}.
\]
In addition, focussed on a specific cluster size, notice that
$\widetilde{\vartheta}_{n_{g},\widehat{n}^{\ast},N}=1+\widetilde{\rho
}_{\widehat{n}^{\ast},N}^{2}(n_{g}-1)$ is an alternative consistent estimator
of $\vartheta_{n_{g}}$, $g=1,...,G$, but it requires from estimators of
$\vartheta_{n^{\ast}}$\ and $\rho^{2}$, i.e. (\ref{estOver}) and
(\ref{estRho}) respectively.

For $\boldsymbol{Y}=\sum\limits_{g=1}^{G}\sum\limits_{\ell=1}^{N_{g}%
}\boldsymbol{Y}^{\left(  g,\ell\right)  }$, it is possible to follow
Definition \ref{DefMin} to obtain the QM$\phi$E of $\boldsymbol{\theta}$,
$\widehat{\boldsymbol{\theta}}_{\phi}=\widehat{\boldsymbol{\theta}}_{\phi
}\left(  \boldsymbol{Y}\right)  $ and also Equation (\ref{eq}) replacing
properly the expression of $\widehat{\boldsymbol{p}}$, according to
(\ref{forp}). In a similar way done for Theorem \ref{Th1}, we have%
\[
\sqrt{N}(\widehat{\boldsymbol{\theta}}_{\phi}\boldsymbol{-\theta}%
_{0})\overset{\mathcal{L}}{\underset{N\rightarrow\infty}{\longrightarrow}%
}\mathcal{N}(\boldsymbol{0}_{M_{0}},\tfrac{\vartheta_{n^{\ast}}}{\bar{n}%
}\left(  \boldsymbol{\boldsymbol{W}}^{T}\boldsymbol{\Sigma
\boldsymbol{_{\boldsymbol{p}\left(  \theta_{0}\right)  }}W}\right)  ^{-1})
\]
and%
\begin{equation}
\sqrt{N}(\boldsymbol{p}(\widehat{\boldsymbol{\theta}}_{\phi})\boldsymbol{-p}%
(\boldsymbol{\theta}_{0}))\overset{\mathcal{L}}{\underset{N\rightarrow
\infty}{\longrightarrow}}\mathcal{N}(\boldsymbol{0}_{M},\tfrac{\vartheta
_{n^{\ast}}}{\bar{n}}\boldsymbol{\Sigma}_{p\left(  \boldsymbol{\theta}%
_{0}\right)  }\boldsymbol{W}\left(  \boldsymbol{\boldsymbol{W}}^{T}%
\boldsymbol{\Sigma\boldsymbol{_{\boldsymbol{p}\left(  \theta_{0}\right)  }}%
W}\right)  ^{-1}\boldsymbol{W}^{T}\boldsymbol{\Sigma}_{p\left(
\boldsymbol{\theta}_{0}\right)  }). \label{asympVar}%
\end{equation}

\subsection{New non-parametric and semi-parametric estimators for
$\vartheta_{n}$ and $\rho^{2}$\label{Sec:new}}

\subsubsection{Case 1: $N_{g}>1$, $g=1,...,G$\label{Sec:new1}}

Let $\widetilde{\boldsymbol{Y}}=(\widetilde{\boldsymbol{Y}}_{1}^{T}%
,...,\widetilde{\boldsymbol{Y}}_{G}^{T})^{T}$, be the whole sample with the
dimension of $\widetilde{\boldsymbol{Y}}_{g}$ being the corresponding
dimension, $N_{g}>1$. Based on the proof of Theorem \ref{Dist} it is possible
to propose a new non-parametric consistent estimator of $\vartheta_{n_{g}}$
with a faster convergence level by using%
\begin{align*}
\widetilde{\vartheta}_{n_{g},N_{g}}(\widetilde{\boldsymbol{Y}}_{g}%
,\widetilde{\boldsymbol{Y}}) &  =\frac{X^{2}(\widetilde{\boldsymbol{Y}}%
_{g},\widetilde{\boldsymbol{Y}})}{(N_{g}-1)(M-1)},\\
X^{2}(\widetilde{\boldsymbol{Y}}_{g},\widetilde{\boldsymbol{Y}}) &  =n_{g}%
\sum_{\ell=1}^{N_{g}}\sum_{r=1}^{M}\frac{(\widehat{p}_{r}^{(\ell
,g)}-\widehat{p}_{r}^{(g)})^{2}}{\widehat{p}_{r}}=n_{g}\sum_{r=1}^{M}\frac
{1}{\widehat{p}_{r}}\sum_{\ell=1}^{N_{g}}(\widehat{p}_{r}^{(\ell
,g)}-\widehat{p}_{r}^{(g)})^{2},
\end{align*}
rather than $X^{2}(\widetilde{\boldsymbol{Y}}_{g})$ and $\widetilde{\vartheta
}_{n_{g},N_{g}}(\widetilde{\boldsymbol{Y}}_{g})$\ respectively, $g=1,...,G$.
Moreover, if the log-linear model were correctly validated, a new
semi-parametric consistent estimator of $\vartheta_{n_{g}}$ even with a faster
convergence degree is given by%
\begin{align*}
\widetilde{\vartheta}_{n_{g},N_{g}}(\widetilde{\boldsymbol{Y}}_{g}%
,\widehat{\boldsymbol{\theta}}_{\phi}) &  =\frac{X^{2}%
(\widetilde{\boldsymbol{Y}}_{g},\widehat{\boldsymbol{\theta}}_{\phi})}%
{(N_{g}-1)(M-1)},\\
X^{2}(\widetilde{\boldsymbol{Y}}_{g},\widehat{\boldsymbol{\theta}}_{\phi}) &
=n_{g}\sum_{\ell=1}^{N_{g}}\sum_{r=1}^{M}\frac{(\widehat{p}_{r}^{(\ell
,g)}-\widehat{p}_{r}^{(g)})^{2}}{p_{r}(\widehat{\boldsymbol{\theta}}_{\phi}%
)}=n_{g}\sum_{r=1}^{M}\frac{1}{p_{r}(\widehat{\boldsymbol{\theta}}_{\phi}%
)}\sum_{\ell=1}^{N_{g}}(\widehat{p}_{r}^{(\ell,g)}-\widehat{p}_{r}^{(g)})^{2},
\end{align*}
$g=1,...,G$. Plugging either $\widetilde{\vartheta}_{n_{g},N_{g}%
}(\widetilde{\boldsymbol{Y}}_{g},\widetilde{\boldsymbol{Y}})$ or
$\widetilde{\vartheta}_{n_{g},N_{g}}(\widetilde{\boldsymbol{Y}}_{g}%
,\widehat{\boldsymbol{\theta}}_{\phi})$ into (\ref{estOver}) in the place of
$\widetilde{\vartheta}_{n_{g},N_{g}}(\widetilde{\boldsymbol{Y}}_{g})$, the new
consistent estimators of the design effect is obtained, for unequal cluster
sizes and based on phi-divergences (the intracluster correlation coefficient,
(\ref{estRho}), is similarly computed).

In the sequel we shall abbreviate by $\widetilde{\vartheta}_{n_{g},N_{g}}$,
$\widetilde{\vartheta}_{n_{g},N_{g},\bullet}$, $\widetilde{\vartheta}%
_{n_{g},N_{g},\phi}$, the three versions $\widetilde{\vartheta}_{n_{g},N_{g}%
}(\widetilde{\boldsymbol{Y}}_{g})$, $\widetilde{\vartheta}_{n_{g},N_{g}%
}(\widetilde{\boldsymbol{Y}}_{g},\widetilde{\boldsymbol{Y}})$,
$\widetilde{\vartheta}_{n_{g},N_{g}}(\widetilde{\boldsymbol{Y}}_{g}%
,\widehat{\boldsymbol{\theta}}_{\phi})$ respectively, and their corresponding
expression for (\ref{estOver}), (\ref{estRho}), $\widetilde{\vartheta
}_{\widehat{n}^{\ast},N}$, $\widetilde{\rho}_{\widehat{n}^{\ast},N}^{2}$,
$\widetilde{\vartheta}_{\widehat{n}^{\ast},N,\bullet}$, $\widetilde{\rho
}_{\widehat{n}^{\ast},N,\bullet}^{2}$, $\widetilde{\vartheta}_{\widehat{n}%
^{\ast},N,\phi}$, $\widetilde{\rho}_{\widehat{n}^{\ast},N,\phi}^{2}$.

\subsubsection{Case 2: $n_{g}$ large enough and $N_{g}\geq1$, $g=1,...,G$%
\label{Sec:new2}}

When the values of the cluster sizes are large, without any loss of generality
can be assumed that $N_{g}=1$ and $G=N$. By following Section \ref{sec4.1} and
taking into account that $\lim_{n_{g\rightarrow\infty}}\tfrac{\vartheta
_{n_{g}}}{n_{g}}=\rho^{2}$,%
\[
\widehat{\boldsymbol{p}}^{(\ell)}\overset{\mathcal{L}}{\underset{n_{\ell
\rightarrow\infty}}{\longrightarrow}}\mathcal{N(}\boldsymbol{p}\left(
\boldsymbol{\theta}_{0}\right)  ,\rho^{2}\boldsymbol{\Sigma}_{\boldsymbol{p}%
\left(  \boldsymbol{\theta}_{0}\right)  }),
\]
which means that $\widehat{\boldsymbol{p}}^{(1)}$, ...,
$\widehat{\boldsymbol{p}}^{(G)}$ are (asymptotically, as $n_{\ell}%
\rightarrow\infty$) i.i.d. $M$-dimensional random variables. Taking into
account similar arguments as the ones given in Section \ref{SecA.1} we obtain
the following consistent estimators of $\rho^{2}$ as $n_{g},N\rightarrow
\infty$
\begin{align*}
\widehat{\rho}^{2} &  =\frac{1}{(N-1)\left(  M-1\right)  }\sum_{\ell=1}%
^{N}\left(  \widehat{\boldsymbol{p}}^{(\ell)}-\frac{1}{N}%
{\textstyle\sum\limits_{s=1}^{N}}
\widehat{\boldsymbol{p}}^{(s)}\right)  ^{T}\boldsymbol{D}%
_{\widehat{\boldsymbol{p}}}^{-1}\left(  \widehat{\boldsymbol{p}}^{(\ell
)}-\frac{1}{N}%
{\textstyle\sum\limits_{s=1}^{N}}
\widehat{\boldsymbol{p}}^{(s)}\right)  \\
&  =\frac{1}{(N-1)\left(  M-1\right)  }\sum_{r=1}^{M}\frac{1}{\widehat{p}_{r}%
}\sum_{\ell=1}^{N}\left(  \widehat{p}_{r}^{(\ell)}-\frac{1}{N}%
{\textstyle\sum\limits_{s=1}^{N}}
\widehat{p}_{r}^{(s)}\right)  ^{2},
\end{align*}
for the saturated model and%
\begin{align*}
\widehat{\rho}^{2}(\widehat{\boldsymbol{\theta}}_{\phi}) &  =\frac
{1}{(N-1)\left(  M-1\right)  }\sum_{\ell=1}^{N}\left(  \widehat{\boldsymbol{p}%
}^{(\ell)}-\frac{1}{N}%
{\textstyle\sum\limits_{s=1}^{N}}
\widehat{\boldsymbol{p}}^{(s)}\right)  ^{T}\boldsymbol{D}_{\boldsymbol{p}%
(\boldsymbol{\widehat{\boldsymbol{\theta}}_{\phi}})}^{-1}\left(
\widehat{\boldsymbol{p}}^{(\ell)}-\frac{1}{N}%
{\textstyle\sum\limits_{s=1}^{N}}
\widehat{\boldsymbol{p}}^{(s)}\right)  \\
&  =\frac{1}{(N-1)\left(  M-1\right)  }\sum_{r=1}^{M}\frac{1}{p_{r}%
(\widehat{\boldsymbol{\theta}}_{\phi})}\sum_{\ell=1}^{N}\left(  \widehat{p}%
_{r}^{(\ell)}-\frac{1}{N}%
{\textstyle\sum\limits_{s=1}^{N}}
\widehat{p}_{r}^{(s)}\right)  ^{2},
\end{align*}
for the log-linear model.

\section{Numerical examples\label{sec5}}

The following two studies represent respectively the numerical examples for
cases $1$ and $2$ in Section \ref{Sec:new}. Focussed on estimating the the
intracluster correlation coefficient, $\rho^{2}$, the semiparametric
consistent estimators are considered for case $1$, and the non-parametric ones
for case $2$. The second example illustrates that the estimation of the
intracluster correlation coefficient corrects the variance we had without
overdispersion, for using same estimators we had without overdispersion. The
corresponding \texttt{Fortran} codes are available at \texttt{%
\hyperref{http://sites.google.com/site/nirianmartinswebsite/software}{}%
{}{http://sites.google.com/site/nirianmartinswebsite/software}%
}.

\subsection{Study on housing satisfaction (Brier, 1980)\label{Housing}}

From all the households located in $N=20$ neighborhoods\ around Montevideo
(Minnesota, US), some households were randomly selected: from $N_{1}=18$
neighborhoods $n_{1}=5$ houses were selected and from $N_{2}=2$ neighborhoods
$n_{2}=3$ houses. The neighborhoods are grouped into class $g=1$ or $g=2$
depending on the selected number of houses (neighborhood or cluster size),
$n_{1}=5$\ and $n_{2}=3$\ respectively. For the $\ell$-th neighborhood
($\ell=1,...,N_{g}$) of the $g$-th cluster size, in the $s$-th selected home
($s=1,...,n_{g}$), the family was questioned on two study interests:
satisfaction with the housing in the neighborhood as a whole ($X_{1s}%
^{(g,\ell)}$), and satisfaction with their own home ($X_{2s}^{(g,\ell)}$). For
both questions the responses were classified as unsatisfied ($US$), satisfied
($S$) or very satisfied ($VS$). In the sequel, we shall identify the
aforementioned categories of the ordinal variables, $X_{11}^{(g,\ell)}$ and
$X_{12}^{(g,\ell)}$, with numbers $1$, $2$, and $3$: for example, $(US,S)$\ is
associated with $(X_{11}^{(g,\ell)}$,$X_{12}^{(g,\ell)})=(1,2)$.

Under the assumption that a family's classification according to level of
personal satisfaction is independent of its classification by level of
community satisfaction, the log-linear model given in (\ref{In1}) is
considered for a $I\times J$ contingency table with $I=J=3$. The corresponding
data, given in Table \ref{t0}, are disaggregated based on the number of houses
and neighborhood identifications $(g,\ell)$ in $20$ rows, having each $M=9$
cells in lexicographical order (number of columns). The design matrix and the
unknown parameter vector are%
\[
\boldsymbol{W}=%
\begin{pmatrix}
1 & 1 & 1 & 0 & 0 & 0 & -1 & -1 & -1\\
0 & 0 & 0 & 1 & 1 & 1 & -1 & -1 & -1\\
1 & 0 & -1 & 1 & 0 & -1 & 1 & 0 & -1\\
0 & 1 & -1 & 0 & 1 & -1 & 0 & 1 & -1
\end{pmatrix}
^{T}\quad\text{and}\quad\boldsymbol{\theta}=(\theta_{1(1)},\theta
_{1(2)},\theta_{2(1)},\theta_{2(2)})^{T}\text{.}%
\]
%

\begin{table}[hbpt]  \tabcolsep2.8pt \small\centering
$%
\begin{tabular}
[c]{rcccccccccc}\hline
& $(US,US)$ & $(US,S)$ & $(US,VS)$ & $(S,US)$ & $(S,S)$ & $(S,VS)$ & $(VS,US)$
& $(VS,S)$ & $(VS,US)$ & \\
$(g,\ell)$ & $Y_{11}^{\left(  g,\ell\right)  }$ & $Y_{12}^{\left(
g,\ell\right)  }$ & $Y_{13}^{\left(  g,\ell\right)  }$ & $Y_{21}^{\left(
g,\ell\right)  }$ & $Y_{22}^{\left(  g,\ell\right)  }$ & $Y_{23}^{\left(
g,\ell\right)  }$ & $Y_{31}^{\left(  g,\ell\right)  }$ & $Y_{32}^{\left(
g,\ell\right)  }$ & $Y_{33}^{\left(  g,\ell\right)  }$ & $n_{g}$\\\hline
$(1,1)$ & $1$ & $0$ & $0$ & $2$ & $2$ & $0$ & $0$ & $0$ & $0$ & $5$\\
$(1,2)$ & $1$ & $0$ & $0$ & $2$ & $2$ & $0$ & $0$ & $0$ & $0$ & $5$\\
$(1,3)$ & $0$ & $2$ & $0$ & $0$ & $2$ & $0$ & $0$ & $1$ & $0$ & $5$\\
$(1,4)$ & $0$ & $1$ & $0$ & $2$ & $1$ & $0$ & $1$ & $0$ & $0$ & $5$\\
$(1,5)$ & $0$ & $0$ & $0$ & $0$ & $4$ & $0$ & $0$ & $1$ & $0$ & $5$\\
$(1,6)$ & $1$ & $0$ & $0$ & $3$ & $1$ & $0$ & $0$ & $0$ & $0$ & $5$\\
$(1,7)$ & $3$ & $0$ & $0$ & $0$ & $1$ & $0$ & $0$ & $1$ & $0$ & $5$\\
$(1,8)$ & $1$ & $0$ & $0$ & $1$ & $3$ & $0$ & $0$ & $0$ & $0$ & $5$\\
$(1,9)$ & $3$ & $0$ & $0$ & $0$ & $0$ & $0$ & $1$ & $0$ & $1$ & $5$\\
$(1,10)$ & $0$ & $1$ & $0$ & $0$ & $3$ & $1$ & $0$ & $0$ & $0$ & $5$\\
$(1,11)$ & $1$ & $1$ & $0$ & $0$ & $2$ & $0$ & $1$ & $0$ & $0$ & $5$\\
$(1,12)$ & $0$ & $1$ & $0$ & $4$ & $0$ & $0$ & $0$ & $0$ & $0$ & $5$\\
$(1,13)$ & $0$ & $0$ & $0$ & $4$ & $1$ & $0$ & $0$ & $0$ & $0$ & $5$\\
$(1,14)$ & $0$ & $0$ & $0$ & $1$ & $2$ & $0$ & $0$ & $0$ & $2$ & $5$\\
$(1,15)$ & $2$ & $0$ & $0$ & $2$ & $1$ & $0$ & $0$ & $0$ & $0$ & $5$\\
$(1,16)$ & $0$ & $0$ & $0$ & $1$ & $1$ & $1$ & $0$ & $2$ & $0$ & $5$\\
$(1,17)$ & $2$ & $0$ & $0$ & $2$ & $1$ & $0$ & $0$ & $0$ & $0$ & $5$\\
$(1,18)$ & $2$ & $0$ & $0$ & $2$ & $0$ & $0$ & $1$ & $0$ & $0$ & $5$\\
$(2,1)$ & $1$ & $0$ & $0$ & $1$ & $1$ & $0$ & $0$ & $0$ & $0$ & $3$\\
$(2,2)$ & $0$ & $0$ & $0$ & $1$ & $0$ & $1$ & $0$ & $0$ & $1$ & $3$\\\hline
total & $18$ & $6$ & $0$ & $28$ & $28$ & $3$ & $4$ & $5$ & $4$ &
$n=96$\\\hline
\end{tabular}
\ \ \ \ \ \ $%
\caption{Housing satisfaction in 20 neighbourhoods of Montevideo (Brier, 1980).\label{t0}}%
\end{table}%

For estimation, the power divergence measures are considered, by restricting
$\phi$ from the family of convex\ functions to the subfamily%
\[
\phi_{\lambda}(x)=\left\{
\begin{array}
[c]{ll}%
\frac{1}{\lambda(1+\lambda)}\left[  x^{\lambda+1}-x-\lambda(x-1)\right]  , &
\lambda\notin\{-1,0\}\\
\lim_{\upsilon\rightarrow\lambda}\frac{1}{\upsilon(1+\upsilon)}\left[
x^{\upsilon+1}-x-\upsilon(x-1)\right]  , & \lambda\in\{-1,0\}
\end{array}
\right.  ,
\]
where $\lambda\in%
\mathbb{R}
$ is a tuning parameter. The expression of (\ref{6}) becomes%
\[
d_{\phi_{\lambda}}(\widehat{\boldsymbol{p}},\boldsymbol{p}(\boldsymbol{\theta
}))=\left\{
\begin{array}
[c]{ll}%
\frac{1}{\lambda(\lambda+1)}%
{\displaystyle\sum\limits_{r=1}^{M}}
\left(  \frac{\widehat{p}_{r}^{\lambda+1}}{p_{r}^{\lambda}\left(
\boldsymbol{\theta}\right)  }-p_{r}\left(  \boldsymbol{\theta}\right)
\right)  , & \lambda\notin\{-1,0\}\\
d_{Kullback}(\boldsymbol{p}\left(  \boldsymbol{\theta}\right)
,\widehat{\boldsymbol{p}}), & \lambda=-1\\
d_{Kullback}(\widehat{\boldsymbol{p}},\boldsymbol{p}\left(  \boldsymbol{\theta
}\right)  ), & \lambda=0
\end{array}
\right.  ,
\]
in such a way that for each $\lambda\in%
\mathbb{R}
$\ a different divergence measure is obtained. By following Definition
\ref{DefMin}, the quasi\ minimum power-divergence estimator (QMPE) of
$\boldsymbol{\theta}$, is given by $\widehat{\boldsymbol{\theta}}%
_{\phi_{\lambda}}=\arg\min_{\theta\in\Theta}d_{\phi_{\lambda}}%
(\widehat{\boldsymbol{p}},\boldsymbol{p}\left(  \boldsymbol{\theta}\right)
)$. Notice that the case of $\lambda=0$ for the QMPE of $\boldsymbol{\theta}$,
$\widehat{\boldsymbol{\theta}}_{\phi_{0}}$, match the QMLE of
$\boldsymbol{\theta}$, $\widehat{\boldsymbol{\theta}}$, or equivalently the
QM$\phi$E\ of $\boldsymbol{\theta}$ with $\phi$ being equal to (\ref{phiKull}%
). Under the independence model, the two parameters of interest,
$\boldsymbol{\beta}=%
\begin{pmatrix}
\rho^{2}\\
\boldsymbol{p}(\boldsymbol{\theta})
\end{pmatrix}
$, are estimated through%
\[
\widehat{\boldsymbol{\beta}}_{\widehat{n}^{\ast},N,\phi_{\lambda}}=%
\begin{pmatrix}
\widetilde{\rho}_{\widehat{n}^{\ast},N,\phi_{\lambda}}^{2}\\
\boldsymbol{p}(\widehat{\boldsymbol{\theta}}_{\phi_{\lambda}})
\end{pmatrix}
=%
\begin{pmatrix}
\frac{\widetilde{\vartheta}_{\widehat{n}^{\ast},N,\phi_{\lambda}}%
-1}{\widehat{n}^{\ast}-1}\\
\frac{\exp\{\boldsymbol{W}\widehat{\boldsymbol{\theta}}_{\phi_{\lambda}}%
\}}{\boldsymbol{1}_{M}^{T}\exp\{\boldsymbol{W}\widehat{\boldsymbol{\theta}%
}_{\phi_{\lambda}}\}}%
\end{pmatrix}
,
\]
where $\widetilde{\vartheta}_{\widehat{n}^{\ast},N}$ is computed as
(\ref{estOver}), and $\widehat{\boldsymbol{\theta}}_{\phi_{\lambda}}$ as
follows. The expression of $\boldsymbol{\Psi}^{\phi_{\lambda}}%
(\boldsymbol{\theta})=(\Psi_{1}^{\phi_{\lambda}}(\boldsymbol{\theta}%
),...,\Psi_{M}^{\phi_{\lambda}}(\boldsymbol{\theta}))^{T}$, for $\lambda\in%
\mathbb{R}
-\{-1\}$, is given by%
\[
\Psi_{r}^{\phi_{\lambda}}(\boldsymbol{\theta})=\frac{1}{1+\lambda}\left(
\frac{\widehat{p}_{r}^{\lambda+1}}{p_{r}^{\lambda}\left(  \boldsymbol{\theta
}\right)  }-p_{r}\left(  \boldsymbol{\theta}\right)  \right)  ,\quad
r=1,...,M,
\]
hence the QMPE of $\boldsymbol{\theta}$, $\widehat{\boldsymbol{\theta}}%
_{\phi_{\lambda}}$, is obtained by solving%
\begin{equation}
\boldsymbol{W}^{T}\boldsymbol{\Sigma}_{p(\boldsymbol{\theta})}\boldsymbol{D}%
_{\boldsymbol{p}(\theta)}^{-(\lambda+1)}(\widehat{\boldsymbol{p}}^{\lambda
+1}-\boldsymbol{p}^{\lambda+1}\left(  \boldsymbol{\theta}\right)
)=\boldsymbol{0}_{M_{0}},\label{MPE}%
\end{equation}
where $\lambda\in%
\mathbb{R}
-\{-1\}$ and%
\begin{align*}
\boldsymbol{\Sigma}_{p(\boldsymbol{\theta})}\boldsymbol{D}_{\boldsymbol{p}%
(\theta)}^{-(\lambda+1)}(\widehat{\boldsymbol{p}}^{\lambda+1}-\boldsymbol{p}%
^{\lambda+1}\left(  \boldsymbol{\theta}\right)  ) &  =\left(  \boldsymbol{I}%
_{M}-\boldsymbol{p}\left(  \boldsymbol{\theta}\right)  \boldsymbol{1}_{M}%
^{T}\right)  \boldsymbol{D}_{\boldsymbol{p}(\theta)}^{-\lambda}%
(\widehat{\boldsymbol{p}}^{\lambda+1}-\boldsymbol{p}^{\lambda+1}\left(
\boldsymbol{\theta}\right)  )\\
&  =%
\begin{bmatrix}
\frac{1-p_{1}\left(  \boldsymbol{\theta}\right)  }{p_{1}^{\lambda}\left(
\boldsymbol{\theta}\right)  } & -\frac{p_{1}\left(  \boldsymbol{\theta
}\right)  }{p_{2}^{\lambda}\left(  \boldsymbol{\theta}\right)  } & \cdots &
-\frac{p_{1}\left(  \boldsymbol{\theta}\right)  }{p_{M}^{\lambda}\left(
\boldsymbol{\theta}\right)  }\\
-\frac{p_{2}\left(  \boldsymbol{\theta}\right)  }{p_{1}^{\lambda}\left(
\boldsymbol{\theta}\right)  } & \frac{1-p_{2}\left(  \boldsymbol{\theta
}\right)  }{p_{2}^{\lambda}\left(  \boldsymbol{\theta}\right)  } & \cdots &
-\frac{p_{2}\left(  \boldsymbol{\theta}\right)  }{p_{M}^{\lambda}\left(
\boldsymbol{\theta}\right)  }\\
\vdots & \vdots & \ddots & \vdots\\
-\frac{p_{M}\left(  \boldsymbol{\theta}\right)  }{p_{1}^{\lambda}\left(
\boldsymbol{\theta}\right)  } & -\frac{p_{M}\left(  \boldsymbol{\theta
}\right)  }{p_{2}^{\lambda}\left(  \boldsymbol{\theta}\right)  } & \cdots &
\frac{1-p_{M}\left(  \boldsymbol{\theta}\right)  }{p_{M}^{\lambda}\left(
\boldsymbol{\theta}\right)  }%
\end{bmatrix}%
\begin{bmatrix}
\widehat{p}_{1}^{\lambda}-p_{1}^{\lambda}\left(  \boldsymbol{\theta}\right)
\\
\widehat{p}_{2}^{\lambda}-p_{2}^{\lambda}\left(  \boldsymbol{\theta}\right)
\\
\vdots\\
\widehat{p}_{M}^{\lambda}-p_{M}^{\lambda}\left(  \boldsymbol{\theta}\right)
\end{bmatrix}
.
\end{align*}
Since%
\[
\overline{\boldsymbol{\Psi}}^{\phi_{\lambda}}(\boldsymbol{\theta})=\frac
{1}{1+\lambda}\left(  \boldsymbol{I}_{M}+\lambda\boldsymbol{D}%
_{\widehat{\boldsymbol{p}}}^{\lambda+1}\boldsymbol{D}_{\boldsymbol{p}%
(\boldsymbol{\theta})}^{-(\lambda+1)}\right)  \boldsymbol{p}%
(\boldsymbol{\theta}),
\]
the expression of $\boldsymbol{G}_{\phi_{\lambda}}(\boldsymbol{\theta})$\ is%
\begin{align*}
\boldsymbol{G}_{\phi_{\lambda}}(\boldsymbol{\theta}) &  =\frac{1}{1+\lambda
}\boldsymbol{W}^{T}\left(  \boldsymbol{I}_{M}-\boldsymbol{p}%
(\boldsymbol{\theta})\boldsymbol{1}_{M}^{T}\right)  \left(  \boldsymbol{I}%
_{M}+\lambda\boldsymbol{D}_{\widehat{\boldsymbol{p}}}^{\lambda+1}%
\boldsymbol{D}_{\boldsymbol{p}(\boldsymbol{\theta})}^{-(\lambda+1)}\right)
\boldsymbol{D}_{\boldsymbol{p}(\boldsymbol{\theta})}\left(  \boldsymbol{I}%
_{M}-\boldsymbol{1}_{M}\boldsymbol{p}^{T}(\boldsymbol{\theta})\right)
\boldsymbol{W}\\
&  =\frac{1}{1+\lambda}\left(  \boldsymbol{W}^{T}\boldsymbol{\Sigma
}_{\boldsymbol{p}(\boldsymbol{\theta})}\boldsymbol{W}+\lambda\boldsymbol{W}%
^{T}\boldsymbol{\Sigma}_{\boldsymbol{p}(\boldsymbol{\theta})}\boldsymbol{D}%
_{\widehat{\boldsymbol{p}}}^{\lambda+1}\boldsymbol{D}_{\boldsymbol{p}%
(\boldsymbol{\theta})}^{-(\lambda+2)}\boldsymbol{\Sigma}_{\boldsymbol{p}%
(\boldsymbol{\theta})}\boldsymbol{W}\right)  ,
\end{align*}
and the Newton-Raphson algorithm for the QMPE of $\boldsymbol{\theta}$ is%
\begin{align}
\widehat{\boldsymbol{\theta}}_{\phi_{\lambda}}^{\left(  t+1\right)  } &
=\widehat{\boldsymbol{\theta}}_{\phi_{\lambda}}^{\left(  t\right)  }-\left(
\boldsymbol{W}^{T}\boldsymbol{\Sigma}_{\boldsymbol{p}%
(\widehat{\boldsymbol{\theta}}_{\phi_{\lambda}}^{\left(  t\right)  }%
)}\boldsymbol{W}+\lambda\boldsymbol{W}^{T}\boldsymbol{\Sigma}_{\boldsymbol{p}%
(\widehat{\boldsymbol{\theta}}_{\phi_{\lambda}}^{\left(  t\right)  }%
)}\boldsymbol{D}_{\widehat{\boldsymbol{p}}}^{\lambda+1}\boldsymbol{D}%
_{\boldsymbol{p}(\widehat{\boldsymbol{\theta}}_{\phi_{\lambda}}^{\left(
t\right)  })}^{-(\lambda+2)}\boldsymbol{\Sigma}_{\boldsymbol{p}%
(\widehat{\boldsymbol{\theta}}_{\phi_{\lambda}}^{\left(  t\right)  }%
)}\boldsymbol{W}\right)  ^{-1}\nonumber\\
&  \times\boldsymbol{W}^{T}\boldsymbol{\Sigma}_{\boldsymbol{p}%
(\widehat{\boldsymbol{\theta}}_{\phi_{\lambda}}^{\left(  t\right)  }%
)}\boldsymbol{D}_{\boldsymbol{p}(\widehat{\boldsymbol{\theta}}_{\phi_{\lambda
}}^{\left(  t\right)  })}^{-(\lambda+1)}(\widehat{\boldsymbol{p}}^{\lambda
+1}-\boldsymbol{p}^{\lambda+1}(\widehat{\boldsymbol{\theta}}_{\phi_{\lambda}%
}^{\left(  t\right)  })).\label{NR}%
\end{align}
Under no model assumption, the two parameters of interest, $\boldsymbol{\beta
}$, are estimated through the saturated log-linear model%
\[
\widehat{\boldsymbol{\beta}}_{\widehat{n}^{\ast},N}=%
\begin{pmatrix}
\widetilde{\rho}_{\widehat{n}^{\ast},N}^{2}\\
\widehat{\boldsymbol{p}}%
\end{pmatrix}
\qquad\text{or}\qquad\widehat{\boldsymbol{\beta}}_{\widehat{n}^{\ast
},N,\bullet}=%
\begin{pmatrix}
\widetilde{\rho}_{\widehat{n}^{\ast},N,\bullet}^{2}\\
\widehat{\boldsymbol{p}}%
\end{pmatrix}
.
\]
Under the independence model assumption as well as no model assumption, the
estimates of $\boldsymbol{\beta}$\ are shown for $\lambda\in
\{-0.5,0,2/3,1,2\}$ in Table \ref{t00}. The intracluster correlation
coefficient exhibits the smallest value under no model assumption and under
the independence model assumption a set of quite different values is obtained.
In Section \ref{sec6}, through a simulation study, some guidance is given for
selecting the most appropriate estimate. The standard errors are also shown
based on the square root of the diagonal elements of
\begin{align*}
\widehat{\mathrm{Var}}[\widehat{\boldsymbol{p}}] &  =\tfrac
{\widetilde{\vartheta}_{\widehat{n}^{\ast},N}}{N\widehat{\bar{n}}%
}\boldsymbol{\Sigma\boldsymbol{_{\widehat{\boldsymbol{p}}}}}\text{\quad
(Brier's non-parametric),}\\
\widehat{\mathrm{Var}}[\widehat{\boldsymbol{p}}] &  =\tfrac
{\widetilde{\vartheta}_{\widehat{n}^{\ast},N,\bullet}}{N\widehat{\bar{n}}%
}\boldsymbol{\Sigma\boldsymbol{_{\widehat{\boldsymbol{p}}}}}\text{\quad
(improved Brier's non-parametric),}\\
\widehat{\mathrm{Var}}[\boldsymbol{p}(\widehat{\boldsymbol{\theta}}%
_{\phi_{\lambda}})] &  =\tfrac{\widetilde{\vartheta}_{\widehat{n}^{\ast
},N,\phi_{\lambda}}}{N\widehat{\bar{n}}}\boldsymbol{\Sigma
\boldsymbol{_{\boldsymbol{p}(\widehat{\boldsymbol{\theta}}_{\phi_{\lambda}})}%
}W}\left(  \boldsymbol{\boldsymbol{W}}^{T}\boldsymbol{\Sigma
\boldsymbol{_{\boldsymbol{p}(\widehat{\boldsymbol{\theta}}_{\phi_{\lambda}})}%
}W}\right)  ^{-1}\boldsymbol{W}^{T}\boldsymbol{\Sigma}%
_{p(\widehat{\boldsymbol{\theta}}_{\phi_{\lambda}})}.
\end{align*}
Taking into account the asymptotic normality of $\widehat{\boldsymbol{p}}%
$\ and $\boldsymbol{p}(\widehat{\boldsymbol{\theta}}_{\phi_{\lambda}})$, their
corresponding confidence intervals, with $1-\alpha$ level, could be calculated
as $\widehat{p}_{r}\mp z_{\alpha/2}\widehat{\mathrm{Var}}[\widehat{p}_{r}]$ or
$p_{r}(\widehat{\boldsymbol{\theta}}_{\phi_{\lambda}})\mp z_{\alpha
/2}\widehat{\mathrm{Var}}[p_{r}(\widehat{\boldsymbol{\theta}}_{\phi_{\lambda}%
})]$, $r=1,...,M$.%

\begin{table}[htbp]  \tabcolsep2.8pt  \centering
$%
\begin{tabular}
[c]{ccccccccccc}\hline
\multicolumn{11}{c}{No model (Brier's non-parametric)}\\\hline
& $\overset{}{\hat{p}_{1}}$ & $\hat{p}_{2}$ & $\hat{p}_{3}$ & $\hat{p}_{4}$ &
$\hat{p}_{5}$ & $\hat{p}_{6}$ & $\hat{p}_{7}$ & $\hat{p}_{8}$ & $\hat{p}_{9}$
& $\widetilde{\rho}_{\widehat{n}^{\ast},N}^{2}$\\\hline
& $0.1875$ & $0.0625$ & $0.0000$ & $0.2917$ & $0.2917$ & $0.0313$ & $0.0417$ &
$0.0521$ & $0.0417$ & $0.0172$\\
& $(0.0411)$ & $(0.0255)$ & $(0.0000)$ & $(0.0479)$ & $(0.0479)$ & $(0.0183)$
& $(0.0210)$ & $(0.0234)$ & $(0.0210)$ & \\\hline
\multicolumn{11}{c}{No model (improved Brier's non-parametric)}\\\hline
& $\overset{}{\hat{p}_{1}}$ & $\hat{p}_{2}$ & $\hat{p}_{3}$ & $\hat{p}_{4}$ &
$\hat{p}_{5}$ & $\hat{p}_{6}$ & $\hat{p}_{7}$ & $\hat{p}_{8}$ & $\hat{p}_{9}$
& $\widetilde{\rho}_{\widehat{n}^{\ast},N,\bullet}^{2}$\\\hline
& $0.1875$ & $0.0625$ & $0.0000$ & $0.2917$ & $0.2917$ & $0.0313$ & $0.0417$ &
$0.0521$ & $0.0417$ & $0.0199$\\\hline
& $(0.0413)$ & $(0.0256)$ & $(0.0000)$ & $(0.0481)$ & $(0.0481)$ & $(0.0184)$
& $(0.0212)$ & $(0.0235)$ & $(0.0212)$ & \\\hline
\multicolumn{11}{c}{Independence model}\\\hline
$\lambda$ & $\overset{}{p_{1}(\widehat{\boldsymbol{\theta}}_{\phi_{\lambda}}%
)}$ & $p_{2}(\widehat{\boldsymbol{\theta}}_{\phi_{\lambda}})$ & $p_{3}%
(\widehat{\boldsymbol{\theta}}_{\phi_{\lambda}})$ & $p_{4}%
(\widehat{\boldsymbol{\theta}}_{\phi_{\lambda}})$ & $p_{5}%
(\widehat{\boldsymbol{\theta}}_{\phi_{\lambda}})$ & $p_{6}%
(\widehat{\boldsymbol{\theta}}_{\phi_{\lambda}})$ & $p_{7}%
(\widehat{\boldsymbol{\theta}}_{\phi_{\lambda}})$ & $p_{8}%
(\widehat{\boldsymbol{\theta}}_{\phi_{\lambda}})$ & $p_{9}%
(\widehat{\boldsymbol{\theta}}_{\phi_{\lambda}})$ & $\widetilde{\rho
}_{\widehat{n}^{\ast},N,\phi_{\lambda}}^{2}$\\\hline
$-0.5$ & $0.1274$ & $0.1001$ & $0.0113$ & $0.3412$ & $0.2682$ & $0.0302$ &
$0.0649$ & $0.0510$ & $0.0057$ & $0.3109$\\
& $(0.0387)$ & $(0.0323)$ & $(0.0082)$ & $(0.0617)$ & $(0.0564)$ & $(0.0207)$
& $(0.0278)$ & $(0.0226)$ & $(0.0045)$ & \\
$0$ & $0.1302$ & $0.1016$ & $0.0182$ & $0.3201$ & $0.2497$ & $0.0448$ &
$0.0705$ & $0.0550$ & $0.0099$ & $0.1545$\\
& $(0.0331)$ & $(0.0276)$ & $(0.0093)$ & $(0.0512)$ & $(0.0464)$ & $(0.0210)$
& $(0.0245)$ & $(0.0198)$ & $(0.0055)$ & \\
$2/3$ & $0.1316$ & $0.1027$ & $0.0252$ & $0.3004$ & $0.2345$ & $0.0575$ &
$0.0751$ & $0.0586$ & $0.0144$ & $0.0872$\\
& $(0.0303)$ & $(0.0253)$ & $(0.0103)$ & $(0.0456)$ & $(0.0411)$ & $(0.0214)$
& $(0.0229)$ & $(0.0186)$ & $(0.0066)$ & \\
$1$ & $0.1319$ & $0.1033$ & $0.0280$ & $0.2931$ & $0.2296$ & $0.0622$ &
$0.0761$ & $0.0596$ & $0.0162$ & $0.0712$\\
& $(0.0296)$ & $(0.0248)$ & $(0.0108)$ & $(0.0440)$ & $(0.0397)$ & $(0.0216)$
& $(0.0225)$ & $(0.0183)$ & $(0.0070)$ & \\
$2$ & $0.1322$ & $0.1054$ & $0.0346$ & $0.2771$ & $0.2209$ & $0.0725$ &
$0.0765$ & $0.0610$ & $0.0200$ & $0.0477$\\
& $(0.0283)$ & $(0.0241)$ & $(0.0118)$ & $(0.0414)$ & $(0.0374)$ & $(0.0222)$
& $(0.0215)$ & $(0.0178)$ & $(0.0078)$ & \\\hline
\end{tabular}
\ \ \ \ \ \ $%
\caption
{Estimates and standard errors (in brackets) of probabilities and intracluster correlation coefficient: non-paramatric and semiparametric versions for the independence model.\label
{t00}}%
\end{table}%

\subsection{Study on FBI data (Weir and Hill, 2002)\label{FBI}}

In an FBI Laboratory Division Publication, article by Budowle and Moretti
(1999), genotype profile data were electronically published. Based on six US
subpopulations, allele frequencies were reported for $13$ commonly-used
forensic loci in the Combined DNA Index System (CODIS): D3S1358, vWA, FGA,
D8S1179, D21S11, D18S51, D5S818, D13S317, D7S820, CSF1PO, TPOX, THO1 and
D16S539. For the first four loci, allele frequencies are summarized in Tables
\ref{t1}-\ref{t4}, based on six clusters, African Americans ($1$), U.S.
Caucasians ($2$), Hispanics ($3$), Bahamians ($4$), Jamaicans ($5$), and
Trinidadians ($6$).%

\begin{table}[htbp]  \tabcolsep2.8pt  \centering
$%
\begin{tabular}
[c]{cccccccccc}\hline
$\ell$ & $Y_{1}^{(\ell)}$ & $Y_{2}^{(\ell)}$ & $Y_{3}^{(\ell)}$ &
$Y_{4}^{(\ell)}$ & $Y_{5}^{(\ell)}$ & $Y_{6}^{(\ell)}$ & $Y_{7}^{(\ell)}$ &
$Y_{8}^{(\ell)}$ & $n_{\ell}$\\\hline
$1$ & $1$ & $5$ & $37$ & $86$ & $99$ & $62$ & $19$ & $2$ & $311$\\
$2$ & $0$ & $0$ & $53$ & $85$ & $85$ & $79$ & $63$ & $2$ & $367$\\
$3$ & $0$ & $1$ & $28$ & $150$ & $100$ & $49$ & $33$ & $6$ & $367$\\
$4$ & $0$ & $0$ & $21$ & $88$ & $96$ & $59$ & $19$ & $1$ & $284$\\
$5$ & $2$ & $5$ & $19$ & $95$ & $81$ & $64$ & $15$ & $2$ & $283$\\
$6$ & $0$ & $0$ & $8$ & $48$ & $42$ & $32$ & $17$ & $0$ & $147$\\\hline
\end{tabular}
\ \ \ $%
\caption{Integer-valued alleles for D3S1358 loci desagregated by US subpopulations: 12, 13, 14, 15, 16, 17, 18, 19 (M=8).\label{t1}}%
\end{table}%
%

\begin{table}[htbp]  \tabcolsep2.8pt  \centering
$%
\begin{tabular}
[c]{cccccccccccc}\hline
$\ell$ & $Y_{1}^{(\ell)}$ & $Y_{2}^{(\ell)}$ & $Y_{3}^{(\ell)}$ &
$Y_{4}^{(\ell)}$ & $Y_{5}^{(\ell)}$ & $Y_{6}^{(\ell)}$ & $Y_{7}^{(\ell)}$ &
$Y_{8}^{(\ell)}$ & $Y_{9}^{(\ell)}$ & $Y_{10}^{(\ell)}$ & $n_{\ell}$\\\hline
$1$ & $0$ & $2$ & $21$ & $76$ & $84$ & $60$ & $37$ & $22$ & $9$ & $0$ &
$311$\\
$2$ & $0$ & $1$ & $35$ & $41$ & $78$ & $97$ & $79$ & $32$ & $4$ & $0$ &
$367$\\
$3$ & $1$ & $19$ & $25$ & $127$ & $89$ & $73$ & $28$ & $5$ & $0$ & $0$ &
$367$\\
$4$ & $3$ & $8$ & $16$ & $43$ & $74$ & $59$ & $51$ & $23$ & $7$ & $0$ &
$284$\\
$5$ & $1$ & $1$ & $19$ & $62$ & $81$ & $53$ & $42$ & $15$ & $7$ & $2$ &
$283$\\
$6$ & $1$ & $1$ & $13$ & $18$ & $44$ & $39$ & $21$ & $7$ & $3$ & $0$ &
$147$\\\hline
\end{tabular}
\ \ $%
\caption{Integer-valued alleles for vWA loci desagregated by US subpopulations: 11, 13, 14, 15, 16, 17, 18, 19, 20, 21 (M=10).\label{t2}}%
\end{table}%
%

\begin{table}[htbp]  \tabcolsep2.8pt  \centering
$%
\begin{tabular}
[c]{ccccccccccccccc}\hline
$\ell$ & $Y_{1}^{(\ell)}$ & $Y_{2}^{(\ell)}$ & $Y_{3}^{(\ell)}$ &
$Y_{4}^{(\ell)}$ & $Y_{5}^{(\ell)}$ & $Y_{6}^{(\ell)}$ & $Y_{7}^{(\ell)}$ &
$Y_{8}^{(\ell)}$ & $Y_{9}^{(\ell)}$ & $Y_{10}^{(\ell)}$ & $Y_{11}^{(\ell)}$ &
$Y_{12}^{(\ell)}$ & $Y_{13}^{(\ell)}$ & $n_{\ell}$\\\hline
$1$ & $3$ & $16$ & $25$ & $38$ & $74$ & $36$ & $59$ & $33$ & $12$ & $7$ & $6$
& $1$ & $1$ & $311$\\
$2$ & $12$ & $19$ & $54$ & $65$ & $68$ & $58$ & $54$ & $26$ & $7$ & $4$ & $0$
& $0$ & $0$ & $367$\\
$3$ & $1$ & $30$ & $27$ & $45$ & $67$ & $52$ & $44$ & $55$ & $32$ & $13$ & $1$
& $0$ & $0$ & $367$\\
$4$ & $0$ & $17$ & $22$ & $31$ & $42$ & $51$ & $60$ & $30$ & $10$ & $16$ & $2$
& $2$ & $0$ & $283$\\
$5$ & $1$ & $19$ & $15$ & $18$ & $61$ & $61$ & $42$ & $32$ & $9$ & $16$ & $6$
& $3$ & $0$ & $283$\\
$6$ & $2$ & $8$ & $14$ & $15$ & $25$ & $23$ & $30$ & $17$ & $6$ & $3$ & $2$ &
$1$ & $1$ & $147$\\\hline
\end{tabular}
\ \ $%
\caption{Integer-valued alleles for FGA loci desagregated by US subpopulations: 18, 19, 20, 21, 22, 23, 24, 25, 26, 27, 28, 29, 30 (M=13).\label{t3}}%
\end{table}%
%

\begin{table}[htbp]  \tabcolsep2.8pt  \centering
$%
\begin{tabular}
[c]{ccccccccccccc}\hline
$\ell$ & $Y_{1}^{(\ell)}$ & $Y_{2}^{(\ell)}$ & $Y_{3}^{(\ell)}$ &
$Y_{4}^{(\ell)}$ & $Y_{5}^{(\ell)}$ & $Y_{6}^{(\ell)}$ & $Y_{7}^{(\ell)}$ &
$Y_{8}^{(\ell)}$ & $Y_{9}^{(\ell)}$ & $Y_{10}^{(\ell)}$ & $Y_{11}^{(\ell)}$ &
$n_{\ell}$\\\hline
$1$ & $1$ & $2$ & $6$ & $12$ & $32$ & $72$ & $104$ & $65$ & $14$ & $3$ & $0$ &
$311$\\
$2$ & $7$ & $4$ & $38$ & $19$ & $53$ & $127$ & $75$ & $40$ & $3$ & $1$ & $0$ &
$367$\\
$3$ & $1$ & $1$ & $34$ & $24$ & $41$ & $117$ & $90$ & $46$ & $10$ & $3$ & $0$
& $367$\\
$4$ & $0$ & $1$ & $6$ & $16$ & $33$ & $54$ & $93$ & $55$ & $19$ & $7$ & $0$ &
$284$\\
$5$ & $0$ & $2$ & $3$ & $11$ & $32$ & $60$ & $89$ & $59$ & $25$ & $1$ & $1$ &
$283$\\
$6$ & $1$ & $0$ & $7$ & $11$ & $22$ & $35$ & $36$ & $26$ & $9$ & $0$ & $0$ &
$147$\\\hline
\end{tabular}
\ \ $%
\caption{Integer-valued alleles for D8S1179 loci desagregated by US subpopulations: 8, 9, 10, 11, 12, 13, 14, 15, 16, 17, 18 (M=11).\label{t4}}%
\end{table}%

Weir and Hill (2002) proposed estimating $\rho^{2}$ from just the first and
second moments of the allele frequency distribution, and this is the essence
of their so-called method of moments%
\begin{equation}
\overline{\rho}^{2}=\frac{\sum\limits_{r=1}^{M}(MSP_{r}-MSG_{r})}%
{\sum\limits_{r=1}^{M}\left(  MSP_{r}+(\bar{\eta}-1)MSG_{r}\right)
},\label{roWeir}%
\end{equation}
where%
\begin{align*}
\bar{\eta} &  =\frac{1}{(N-1)n}\left(  \left(  \sum_{\ell=1}^{N}n_{\ell
}\right)  ^{2}-\sum_{\ell=1}^{N}n_{\ell}^{2}\right)  ,\\
MSP_{r} &  =\frac{1}{N-1}\sum_{\ell=1}^{N}n_{\ell}(\widehat{p}_{r}^{(\ell
)}-\widehat{p}_{r})^{2},\\
MSG_{r} &  =\frac{1}{\sum_{\ell=1}^{N}n_{\ell}-N}\sum_{\ell=1}^{N}n_{\ell
}\widehat{p}_{r}^{(\ell)}(1-\widehat{p}_{r}^{(\ell)}).
\end{align*}
In Table \ref{t3} the estimates of $\boldsymbol{p}$\ and $\rho^{2}$\ are shown
for loci\ D3S1358, vWA, FGA and D8S1179. The intracluster correlation
coefficient exhibits a similar value for both methods, the new proposed
estimation of Section \ref{Sec:new2} ($\widehat{\rho}^{2}$) and Weir and Hill
estimation ($\overline{\rho}^{2}$). The standard errors are also shown based
on the square root of the diagonal elements of
\begin{align*}
\widehat{\mathrm{Var}}(\widehat{\boldsymbol{p}}) &  =\tfrac{\widehat{\rho}%
^{2}}{N^{2}}\boldsymbol{\Sigma\boldsymbol{_{\widehat{\boldsymbol{p}}}}%
}\text{\quad(with overdispersion),}\\
\widehat{\mathrm{Var}}(\widehat{\boldsymbol{p}}) &  =\tfrac{1}{\sum_{\ell
=1}^{N}n_{\ell}}\boldsymbol{\Sigma\boldsymbol{_{\widehat{\boldsymbol{p}}}}%
}\text{\quad(without overdispersion).}%
\end{align*}
In the case with overdispersion the standard errors have bigger values than in
the case without overdispersion (e1%
$>$%
e2). The explanation to this difference is based on the correction that
$\widehat{\rho}^{2}$\ inherits for having large cluster sizes and in the case
without overdispersion the formula does not inherits the assumption that the
cluster sizes are large.%

\begin{table}[htbp]  \small\tabcolsep1.5pt  \centering
$%
\begin{tabular}
[c]{ccccccccccccccc}\hline
loci &  & $\overset{}{\hat{p}_{1}}$ & $\hat{p}_{2}$ & $\hat{p}_{3}$ & $\hat
{p}_{4}$ & $\hat{p}_{5}$ & $\hat{p}_{6}$ & $\hat{p}_{7}$ & $\hat{p}_{8}$ &
$\hat{p}_{9}$ & $\hat{p}_{10}$ & $\hat{p}_{11}$ & $\hat{p}_{12}$ & $\hat
{p}_{13}$\\\hline
\multicolumn{1}{l}{D3S1358} & est & $0.0017$ & $0.0063$ & $0.0944$ & $0.3138$
& $0.2860$ & $0.1961$ & $0.0944$ & $0.0074$ & $-$ & $-$ & $-$ & $-$ & $-$\\
& e1 & $0.0010$ & $0.0019$ & $0.0070$ & $0.0111$ & $0.0108$ & $0.0095$ &
$0.0070$ & $0.0020$ & $-$ & $-$ & $-$ & $-$ & $-$\\
& e2 & $0.0007$ & $0.0014$ & $0.0051$ & $0.0081$ & $0.0079$ & $0.0069$ &
$0.0051$ & $0.0015$ & $-$ & $-$ & $-$ & $-$ & $-$\\
\multicolumn{1}{l}{vWA} & est & $0.0034$ & $0.0182$ & $0.0733$ & $0.2086$ &
$0.2558$ & $0.2166$ & $0.1467$ & $0.0591$ & $0.0171$ & $0.0011$ & $-$ & $-$ &
$-$\\
& e1 & $0.0014$ & $0.0032$ & $0.0062$ & $0.0097$ & $0.0104$ & $0.0098$ &
$0.0084$ & $0.0056$ & $0.0031$ & $0.0008$ & $-$ & $-$ & $-$\\
& e2 & $0.0011$ & $0.0026$ & $0.0050$ & $0.0078$ & $0.0084$ & $0.0079$ &
$0.0068$ & $0.0045$ & $0.0025$ & $0.0006$ & $-$ & $-$ & $-$\\
\multicolumn{1}{l}{FGA} & est & $0.0108$ & $0.0620$ & $0.0893$ & $0.1206$ &
$0.1917$ & $0.1598$ & $0.1644$ & $0.1098$ & $0.0432$ & $0.0336$ & $0.0097$ &
$0.0040$ & $0.0011$\\
& e1 & $0.0025$ & $0.0058$ & $0.0068$ & $0.0078$ & $0.0094$ & $0.0087$ &
$0.0088$ & $0.0075$ & $0.0049$ & $0.0043$ & $0.0023$ & $0.0015$ & $0.0008$\\
& e2 & $0.0016$ & $0.0038$ & $0.0045$ & $0.0052$ & $0.0062$ & $0.0058$ &
$0.0059$ & $0.0050$ & $0.0032$ & $0.0029$ & $0.0016$ & $0.0010$ & $0.0005$\\
\multicolumn{1}{l}{D8S1179} & est & $0.0057$ & $0.0057$ & $0.0534$ & $0.0529$
& $0.1211$ & $0.2644$ & $0.2769$ & $0.1654$ & $0.0455$ & $0.0085$ & $0.0006$ &
$-$ & $-$\\
& e1 & $0.0018$ & $0.0018$ & $0.0054$ & $0.0053$ & $0.0078$ & $0.0105$ &
$0.0107$ & $0.0089$ & $0.0050$ & $0.0022$ & $0.0006$ & $-$ & $-$\\
& e2 & $0.0014$ & $0.0014$ & $0.0040$ & $0.0040$ & $0.0059$ & $0.0079$ &
$0.0080$ & $0.0067$ & $0.0037$ & $0.0017$ & $0.0004$ & $-$ & $-$\\\hline
loci &  & $\overset{}{\widehat{\rho}^{2}}$ & $\overline{\rho}^{2}$ &  &  &  &
&  &  &  &  &  &  & \\\cline{1-4}%
D3S1358 & est & $0.0109$ & $0.0109$ &  &  &  &  &  &  &  &  &  &  & \\
vWA & est & $0.0133$ & $0.0156$ &  &  &  &  &  &  &  &  &  &  & \\
FGA & est & $0.0090$ & $0.0065$ &  &  &  &  &  &  &  &  &  &  & \\
D8S1179 & est & $0.0116$ & $0.0129$ &  &  &  &  &  &  &  &  &  &  &
\\\cline{1-4}%
\end{tabular}
\ \ \ \ \ \ \ \ $%
\caption{Estimates of $\boldsymbol{p}$ and $\rho^{2}$ for loci D3S1358, vWA, FGA and D8S1179 and standard errors of probabilities without (e1) and with (e2) overdispersion.\label{t3}}%
\end{table}%
\newpage

\section{Simulation study\label{sec6}}

The major issue of interest of this section is to investigate, through Monte
Carlo simulations, the improvement of the new estimators of the intracluster
correlation coefficient, $\rho^{2}$, based on $X^{2}(\widetilde{\boldsymbol{Y}%
}_{g},\widetilde{\boldsymbol{Y}})$ or $X^{2}(\widetilde{\boldsymbol{Y}}%
_{g},\widehat{\boldsymbol{\theta}}_{\phi_{\lambda}})$, in comparison with
either the Brier's classical one, based on $X^{2}(\widetilde{\boldsymbol{Y}%
}_{g})$ (see Section \ref{sec4}), or the Weir and Hill's proposal (see Section
\ref{FBI}). Such an improvement is measured through $R=15,000$ replications,
in terms of the root of the mean square error ($\mathrm{RMSE}$) and bias. The
estimates are truncated at $0$ or $1$, to restrict the parameter space of
$\rho^{2}$ to $(0,1)$. As underlying unknown distributions, three scenarios
are taken into account: the Dirichlet-multinomial (DM), the n-inflated
multinomial (NI) and the random clumped (RC) distributions. In Appendix
\ref{Alg} the algorithms to generate observations from these distributions are
provided. Initially, we tried to use the \texttt{drnbet fortran IMSL
subroutine}\ to generate the DM distributions and we saw that it does not
generate observations correctly from the beta distribution. Later, we
discovered that Ahn and James (1995) had the same problem, and for this reason
we have used the \texttt{G05FEF fortran NAG subroutine}.\newpage

\subsection{Simulation: study on housing satisfaction}

Based on a mild modification of the study of housing satisfaction (Section
\ref{Housing}), $N_{1}=18$, $N_{2}=2$, $N_{3}=5$ clusters are considered with
$G=3$ different cluster sizes, $n_{1}=5$, $n_{2}=3$, $n_{3}=7$. In this way,
the experiment can be evaluated for a value $G$ not so close to $G=1$ (equal
cluster sizes). With theoretical values for the vector of unknown parameters
$\boldsymbol{\theta}=(\theta_{1(1)},\theta_{1(2)},\theta_{2(1)},\theta
_{2(2)})^{T}=(0.1,0.2,0.4,0.3)^{T}$, the clustered multinomial distributions
are simulated under the independence log-linear model of Section \ref{Housing}.

In Figure \ref{fig1}, the plots on left hand side exhibit a greater value
going up, for the three distribution, which means that $\mathrm{RMSE}%
(\widetilde{\rho}_{\widehat{n}^{\ast},N,\phi_{\lambda}}^{2})<\mathrm{RMSE}%
(\widetilde{\rho}_{\widehat{n}^{\ast},N,\bullet}^{2})<\mathrm{RMSE}%
(\widetilde{\rho}_{\widehat{n}^{\ast},N}^{2})$ with $\lambda=\frac{2}{3}$. A
big part of the $\mathrm{RMSE}$ is due to bias, in fact $\mathrm{bias}%
(\widetilde{\rho}_{\widehat{n}^{\ast},N,\phi_{\lambda}}^{2})<\mathrm{bias}%
(\widetilde{\rho}_{\widehat{n}^{\ast},N,\bullet}^{2})<\mathrm{bias}%
(\widetilde{\rho}_{\widehat{n}^{\ast},N}^{2})$ with $\lambda=\frac{2}{3}$ and
for $\widetilde{\rho}_{\widehat{n}^{\ast},N,\phi_{\lambda}}^{2}$ with
$\lambda=\frac{2}{3}$ and $\widetilde{\rho}_{\widehat{n}^{\ast},N,\bullet}%
^{2}$ the negative bias is becoming greater as $\rho^{2}$\ increases.
Identifying the proper log-linear model makes the bias of $\widetilde{\rho
}_{\widehat{n}^{\ast},N,\phi_{\lambda}}^{2}$ with $\lambda=\frac{2}{3}$
clearly smaller and stable as $\rho^{2}$\ increases. The estimators were
constructed under no distributional assumption but from the simulation study,
but the behaviour of the estimators are appreciated to be quite different
depending on the distributional assumption. It is also worth of being
mentioned that the $\mathrm{RMSE}$ and the bias of the estimors of $\rho^{2}%
$\ tends to be smaller with the DM and RC distributions in comparison with the
NI distribution. The estimators with the DM distribution seem to be more
precise and the estimators with the RC distribution less biased. In Figure
\ref{fig2}, density plots based on the $15,000$ replications are shown for
$\rho^{2}=0.5$, and from them the same conclusions about the bias are
obtained.\ By following the results of Figure \ref{fig3}, where $\mathrm{RMSE}%
$ and the bias of $\widetilde{\rho}_{\widehat{n}^{\ast},N,\phi_{\lambda}}^{2}$
is compared for $\lambda\in\{-0.5,0,\frac{2}{3},1,2\}$, the QMPE with
$\lambda\in\{\frac{2}{3},1\}$\ tends to be more precise than the QMLE
($\lambda=0$), however the QMLE ($\lambda=0$) seems to be more unbiased. The
optimal choice of $\lambda$ for $\widetilde{\rho}_{\widehat{n}^{\ast}%
,N,\phi_{\lambda}}^{2}$ seems to be very related with the optimal choice of of
$\lambda$ for for the QMPE of $\boldsymbol{\theta}$.%

\begin{figure}[htbp]  \centering
\begin{tabular}
[c]{cc}%
{\includegraphics[
height=2.6489in,
width=3.5284in
]%
{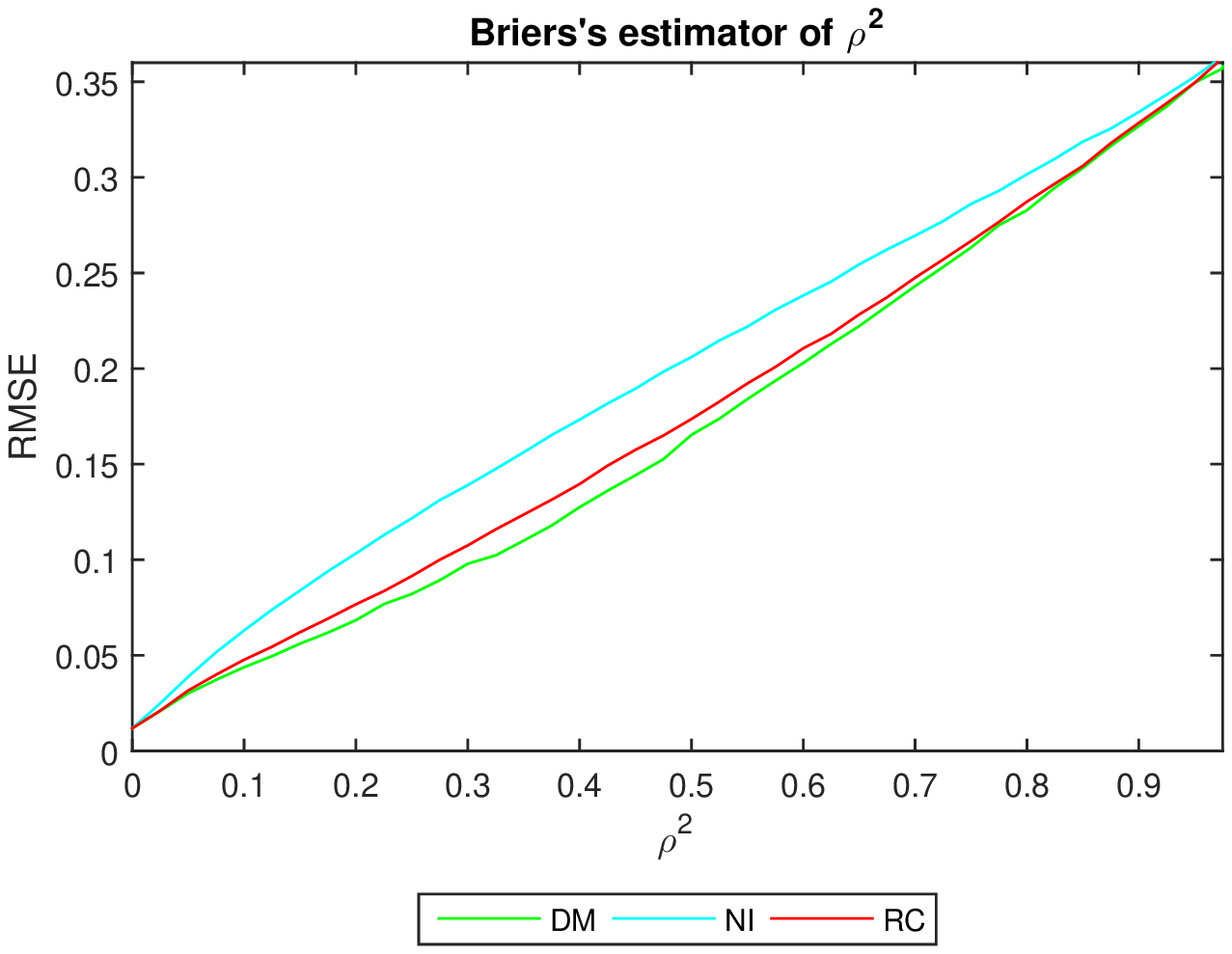}%
}
&
{\includegraphics[
height=2.6489in,
width=3.5284in
]%
{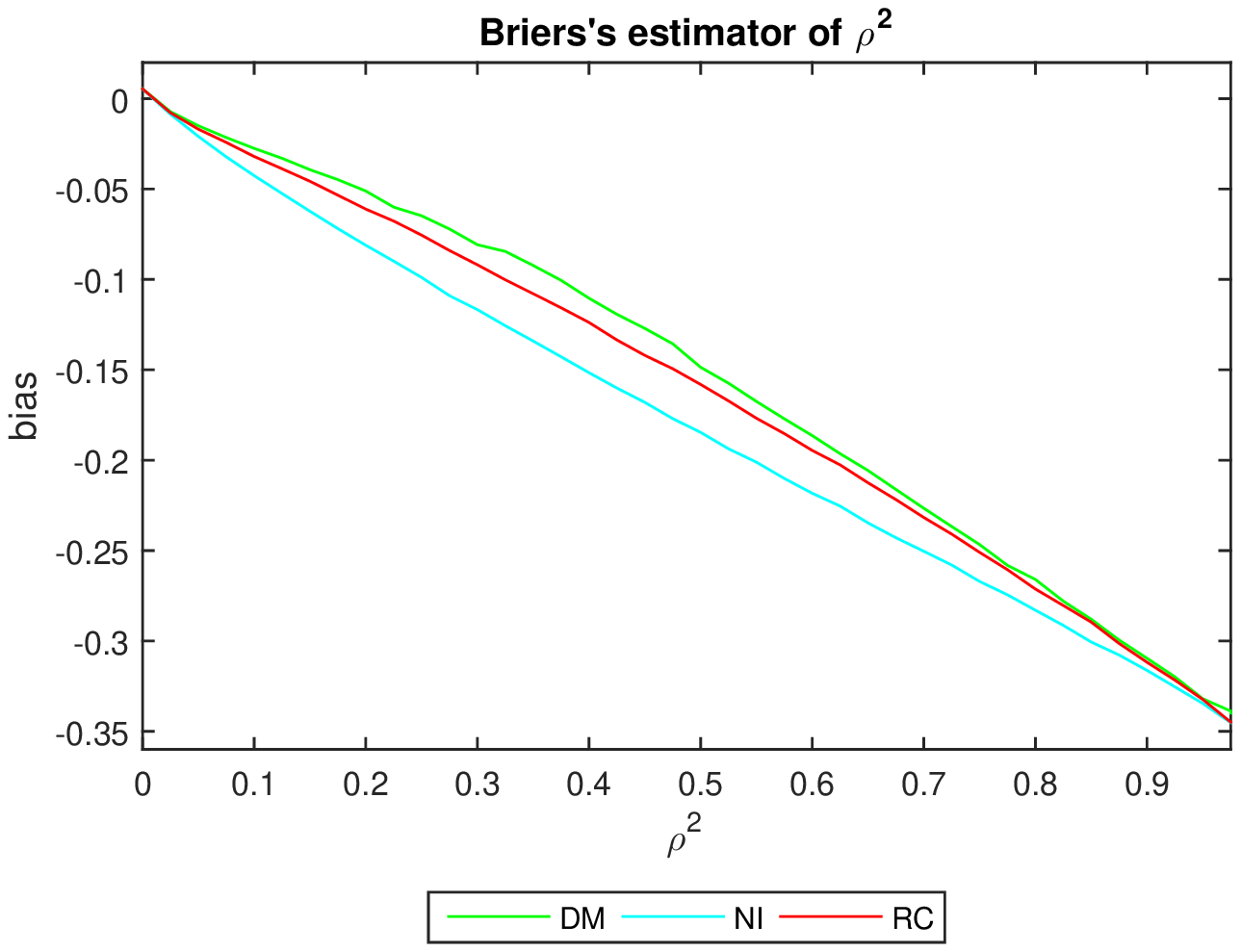}%
}
\\%
{\includegraphics[
height=2.6489in,
width=3.5284in
]%
{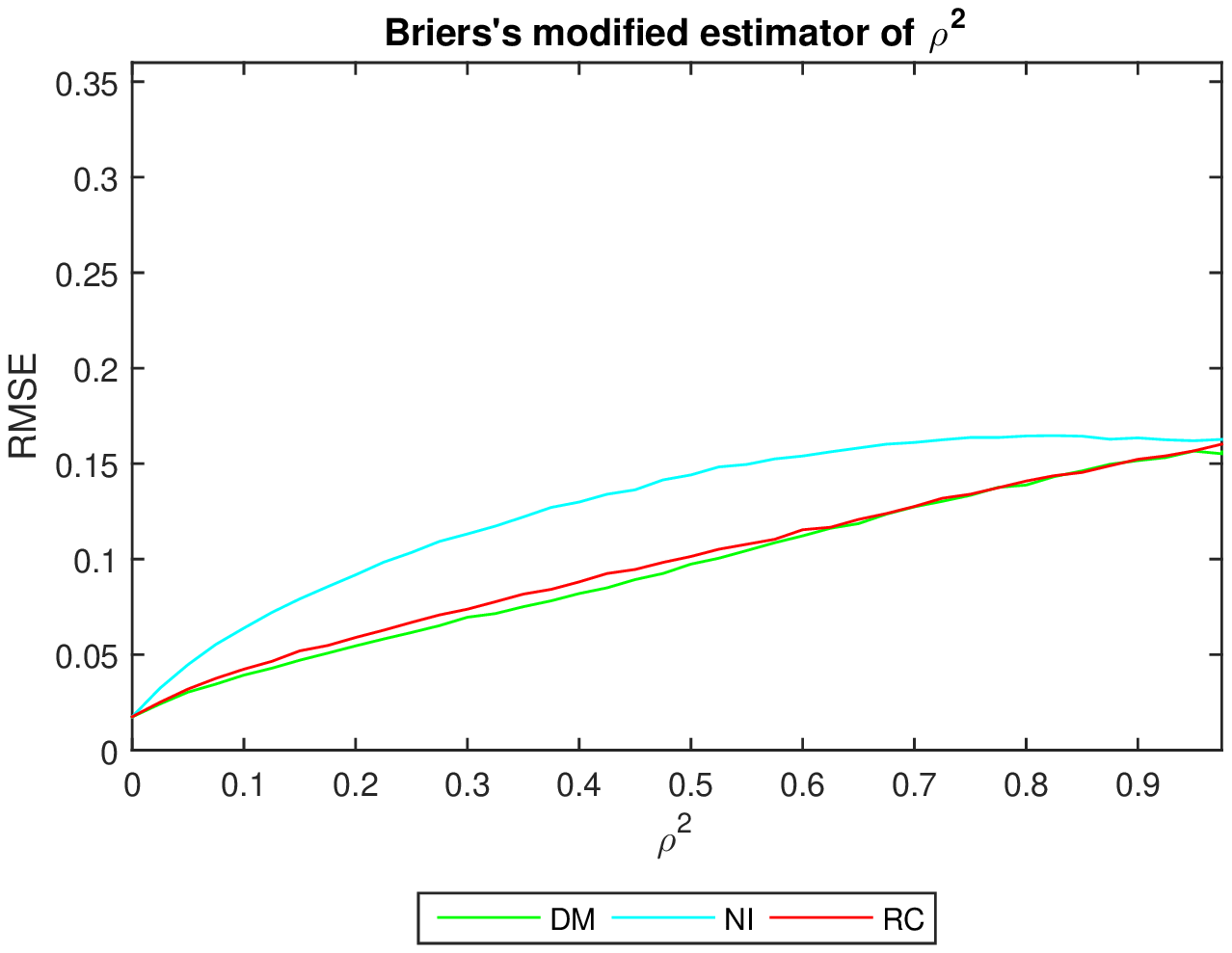}%
}
&
{\includegraphics[
height=2.6489in,
width=3.5284in
]%
{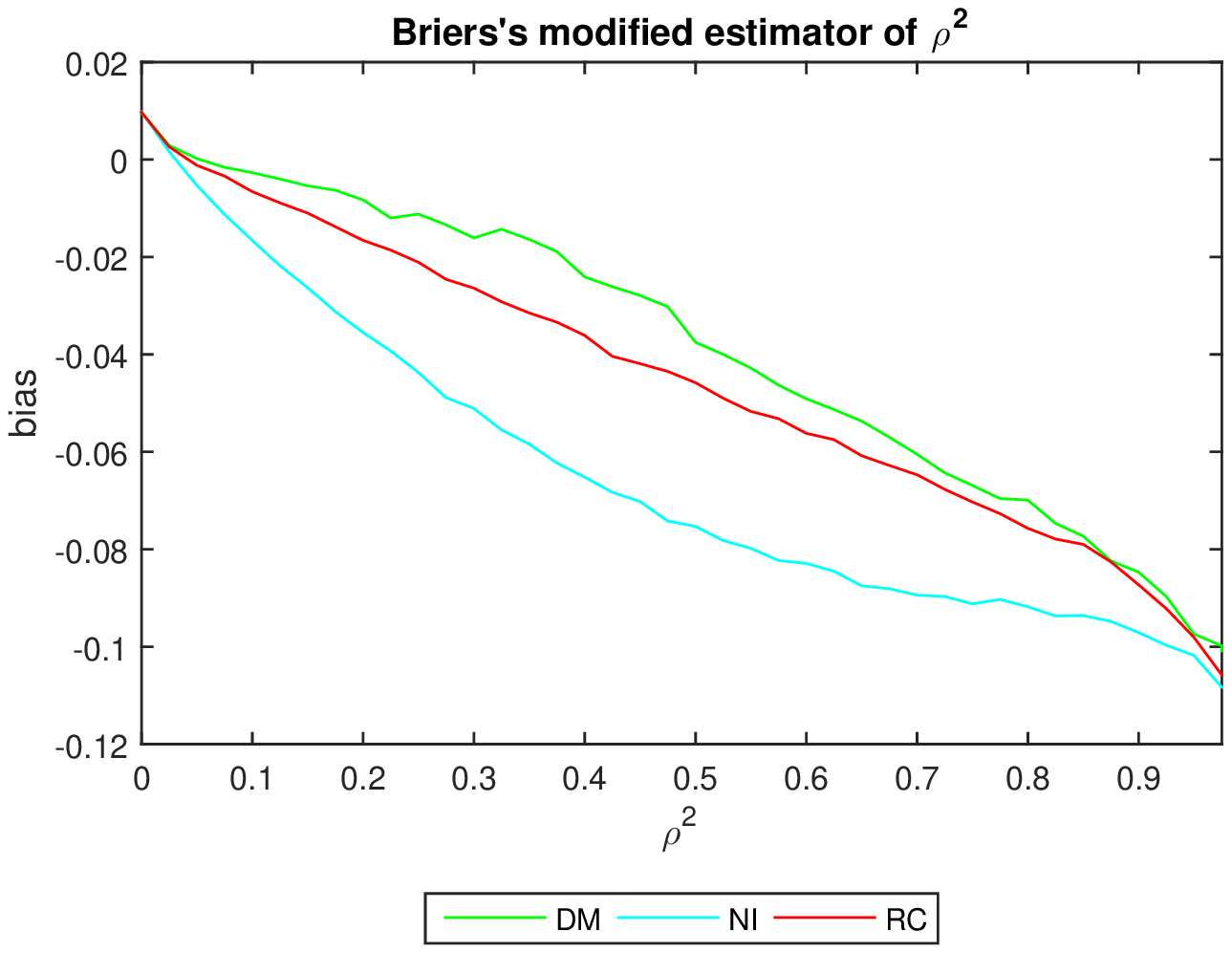}%
}
\\%
{\includegraphics[
height=2.6489in,
width=3.5284in
]%
{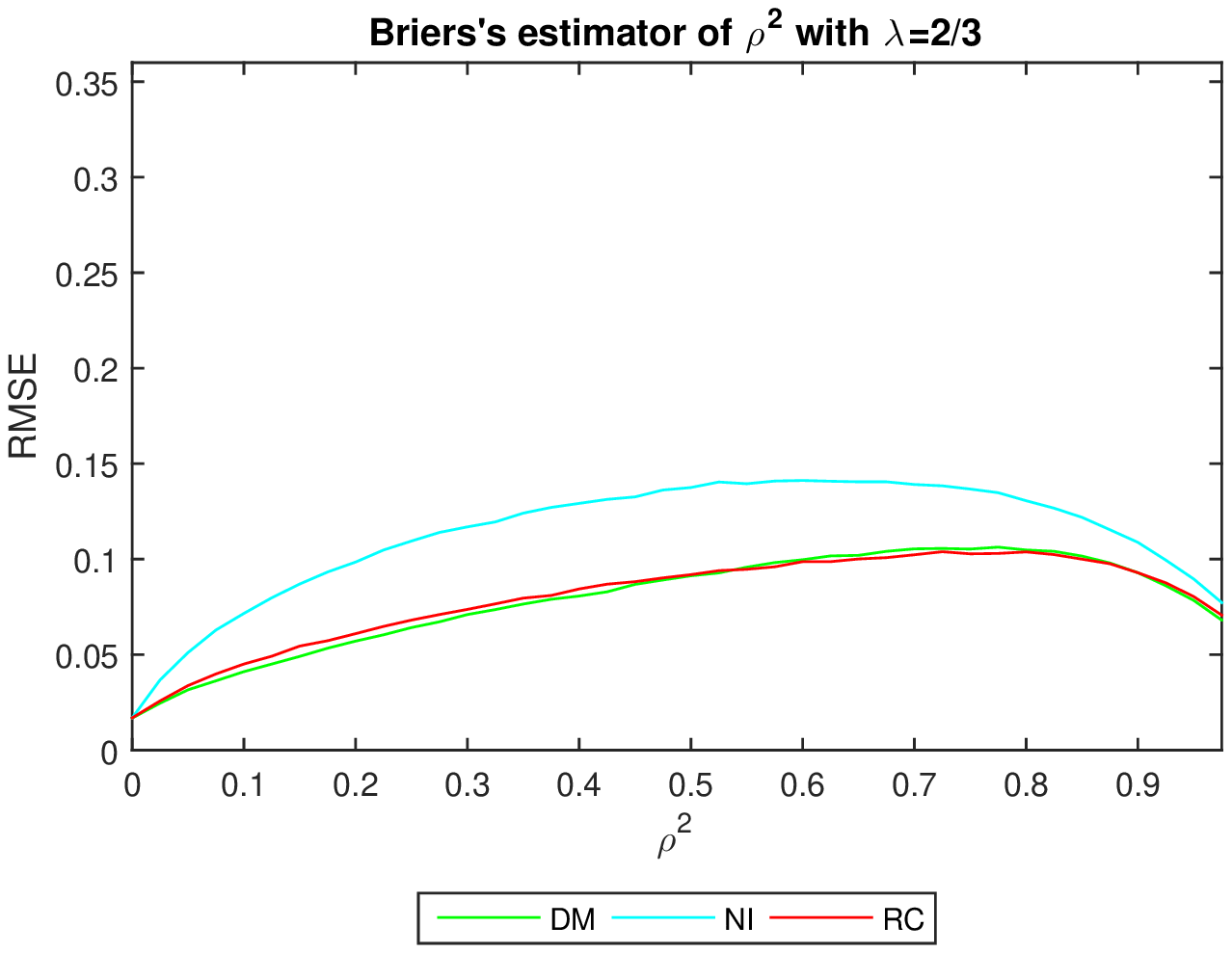}%
}
&
{\includegraphics[
height=2.6489in,
width=3.5284in
]%
{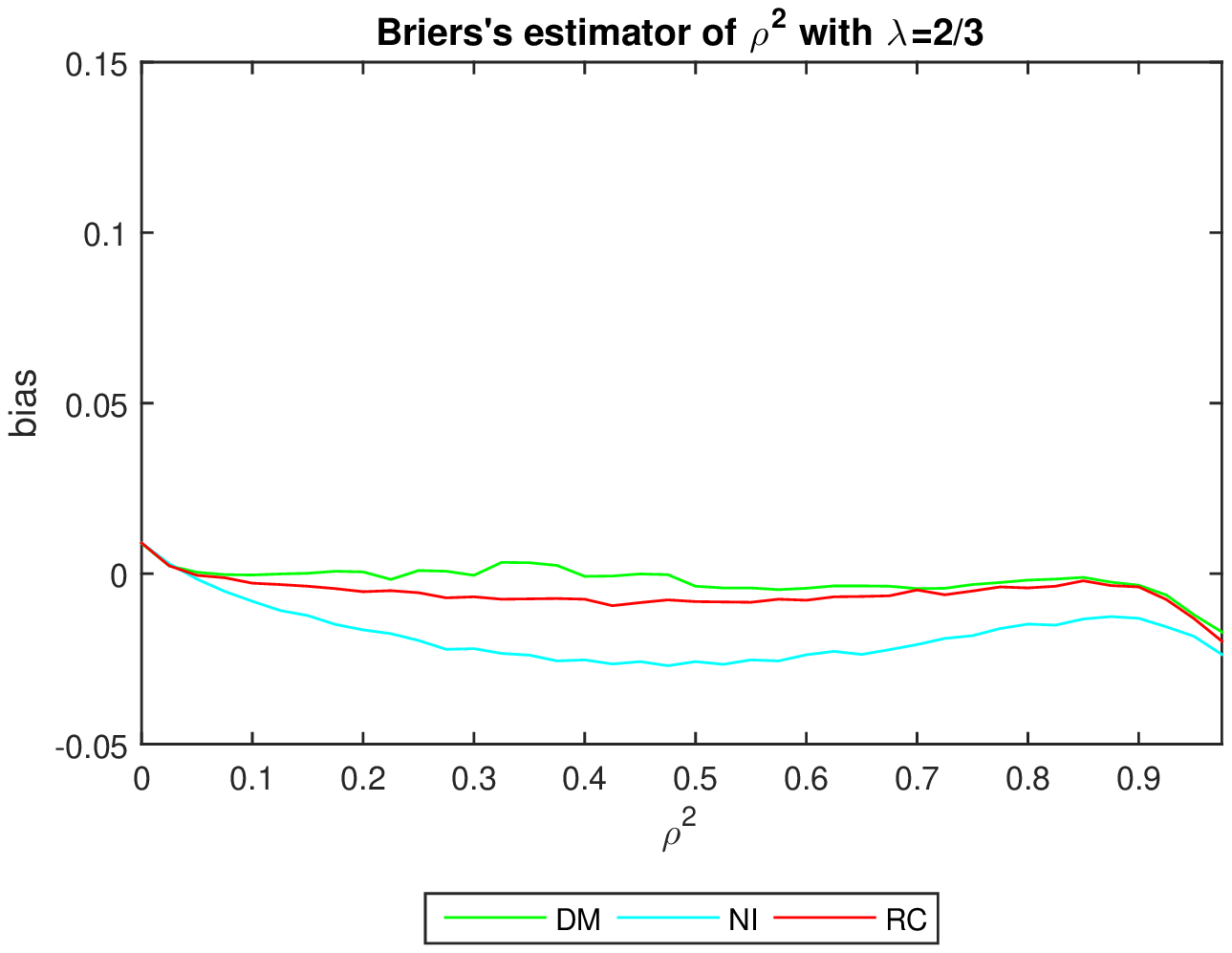}%
}
\end{tabular}
\caption{RMSE and bias for different estimators of $\rho^{2}$: $\widetilde{\rho}_{\widehat{n}^{\ast},N}^{2}$ (top), $\widetilde{\rho}_{\widehat{n}^{\ast},N,\bullet}^{2}$ (middle), $\widetilde{\rho}_{\widehat{n}^{\ast},N,\lambda}^{2}$ with $\lambda=2/3$ (bottom).\label{fig1}}%
\end{figure}%
%

\begin{figure}[htbp]  \centering
\begin{tabular}
[c]{c}%
{\includegraphics[
height=5.649in,
width=6.7654in
]%
{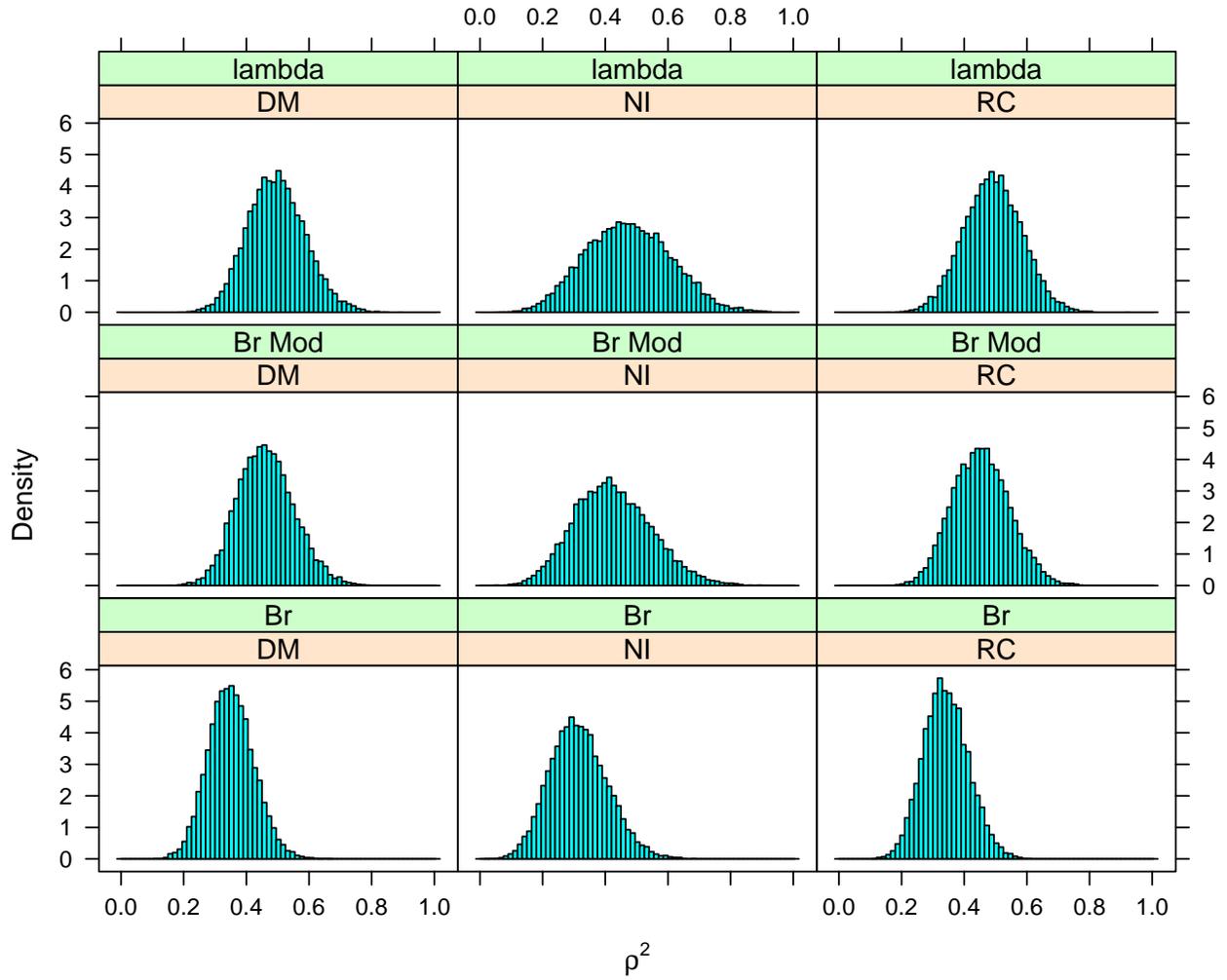}%
}
\end{tabular}
\caption{Density plots with estimates obtained from observations of three distributions, DM, NI, RC, when $\rho^{2}=0.5$: $\widetilde{\rho}_{\widehat{n}^{\ast},N}^{2}$ (below, Br), $\widetilde{\rho}_{\widehat{n}^{\ast},N,\bullet}^{2}$ (middle, Br Mod), $\widetilde{\rho}_{\widehat{n}^{\ast},N,\lambda}^{2}$ with $\lambda=2/3$ (top, lambda)\label{fig2}}%
\end{figure}%
%

\begin{figure}[htbp]  \centering
\begin{tabular}
[c]{cc}%
{\includegraphics[
height=2.6498in,
width=3.5284in
]%
{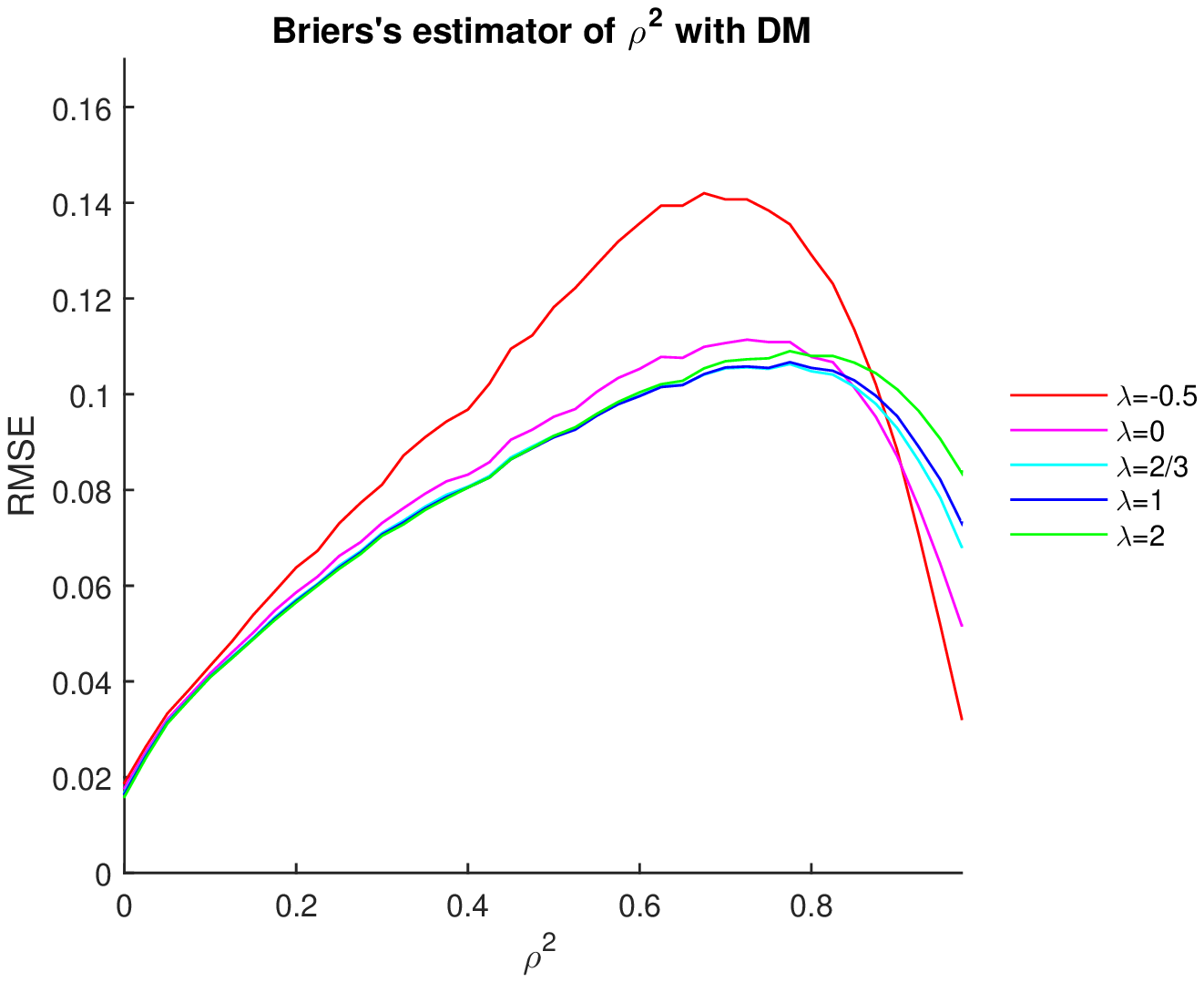}%
}
&
{\includegraphics[
height=2.6498in,
width=3.5284in
]%
{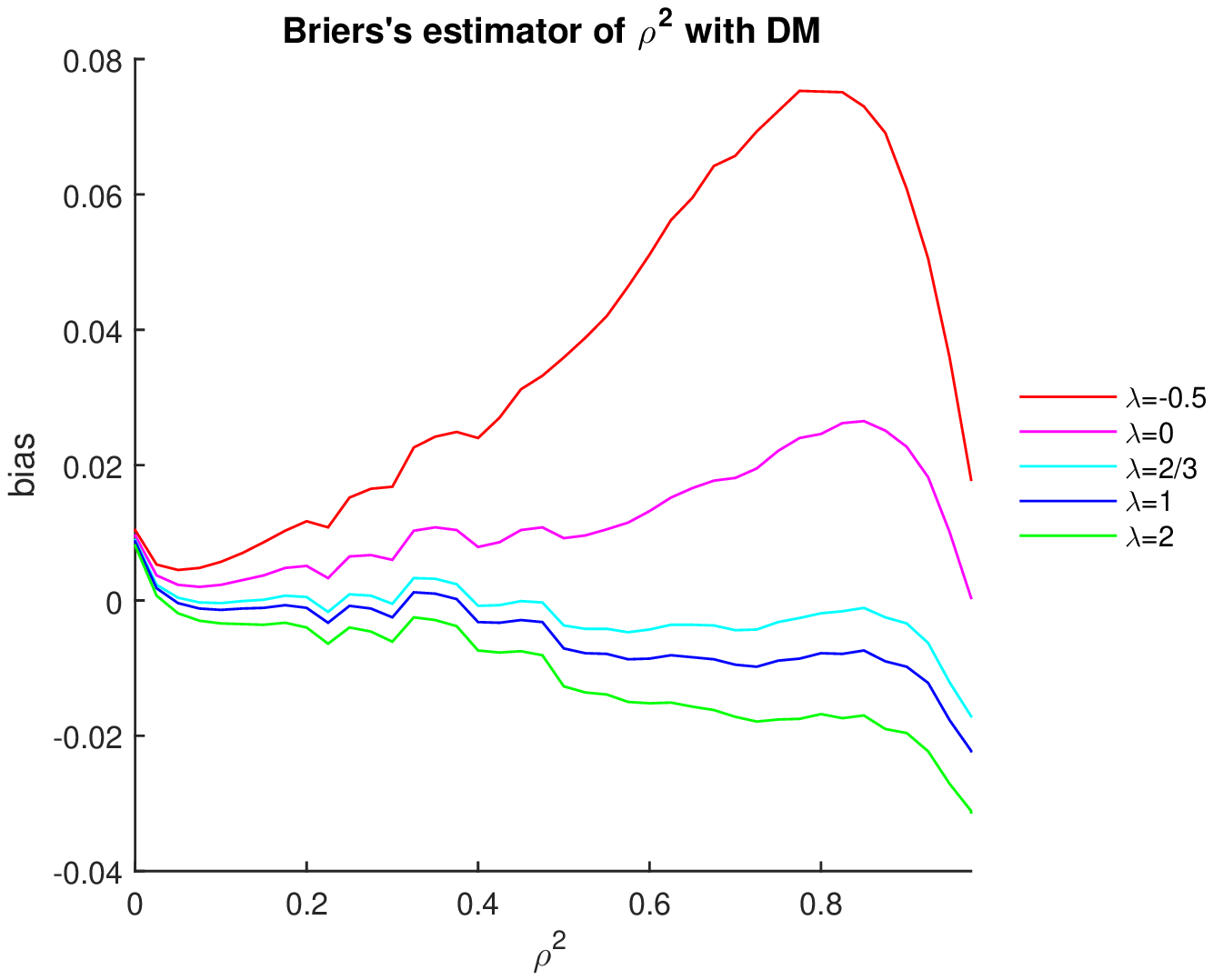}%
}
\\%
{\includegraphics[
height=2.6498in,
width=3.5284in
]%
{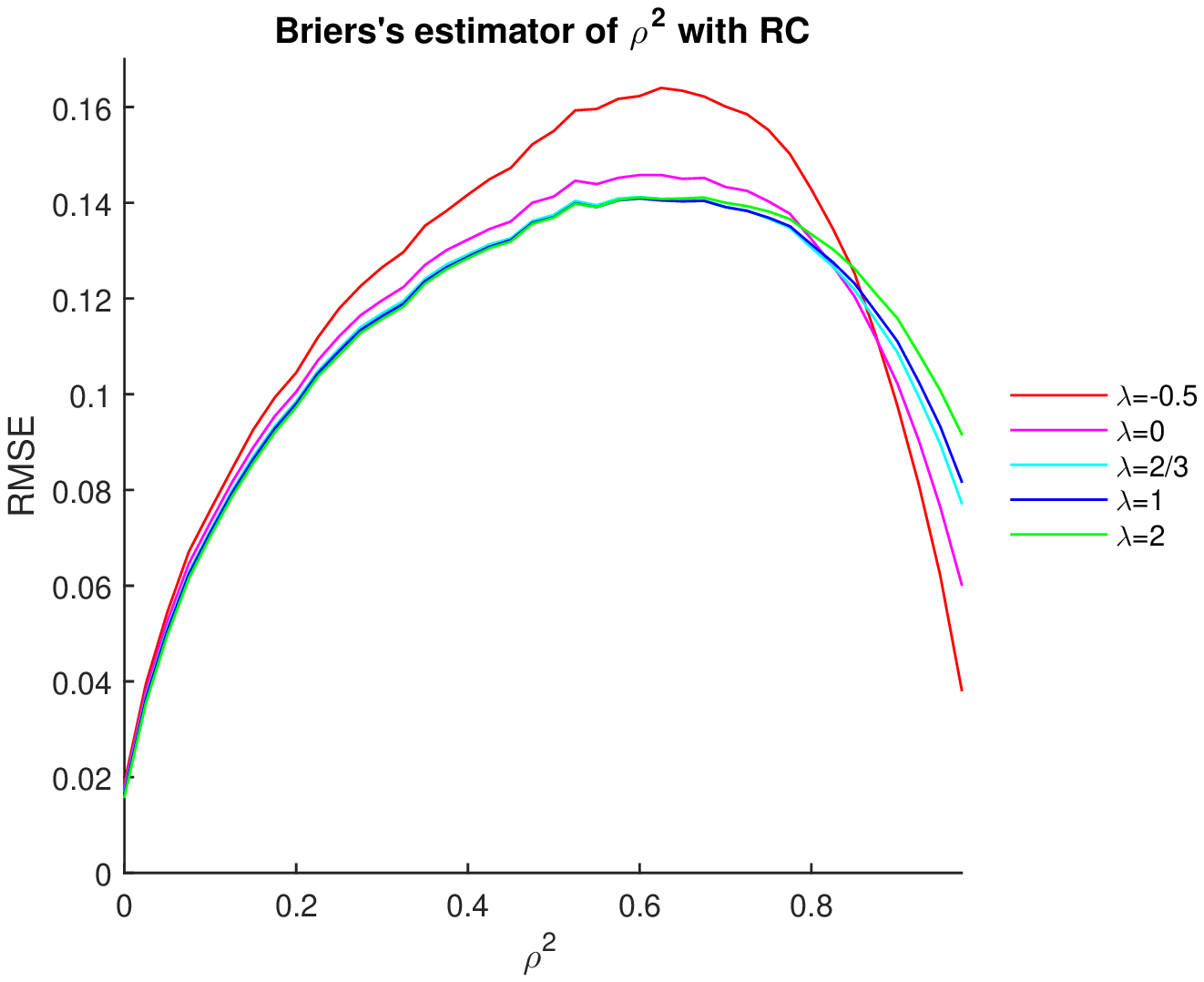}%
}
&
{\includegraphics[
height=2.6498in,
width=3.5284in
]%
{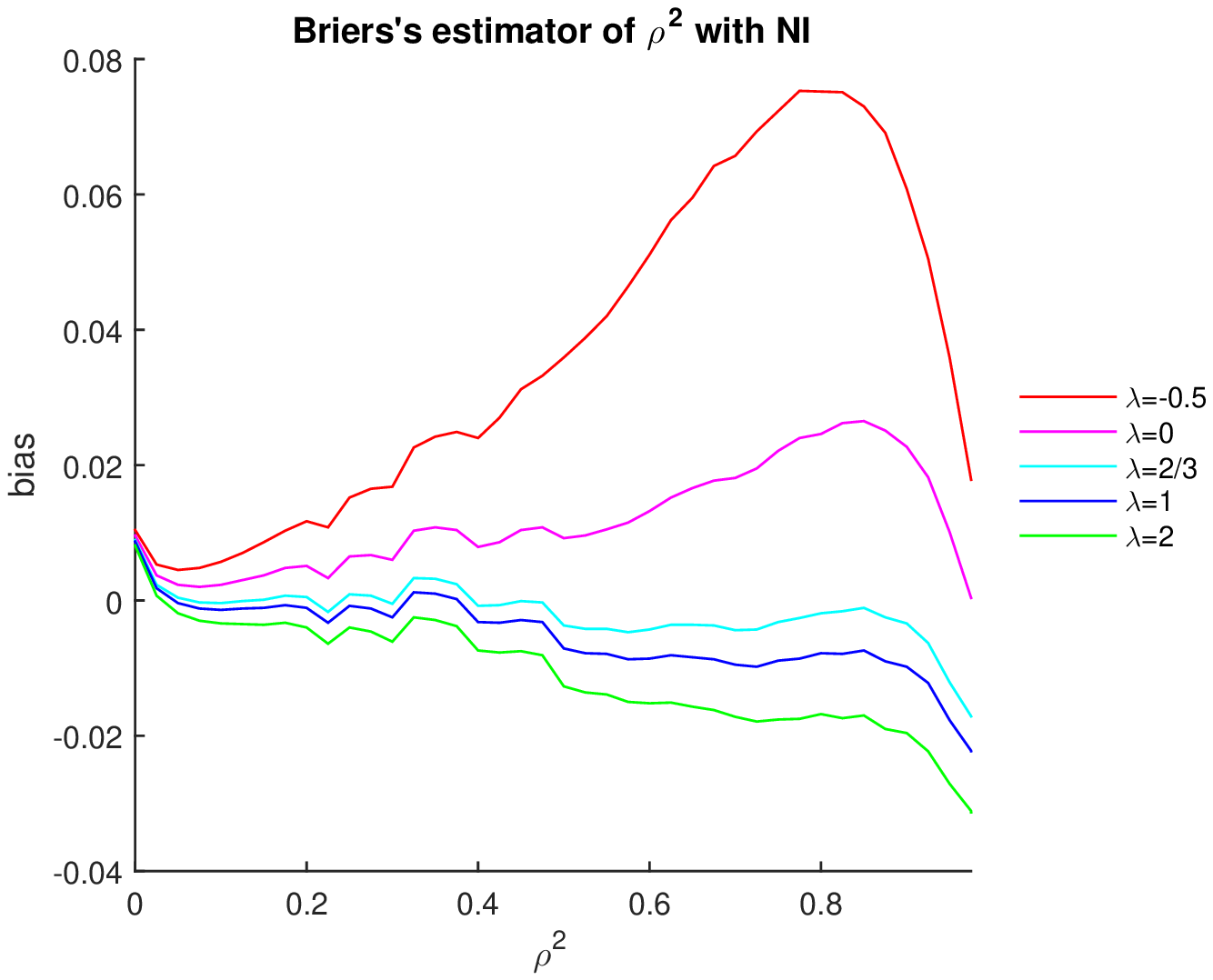}%
}
\\%
{\includegraphics[
height=2.6498in,
width=3.5284in
]%
{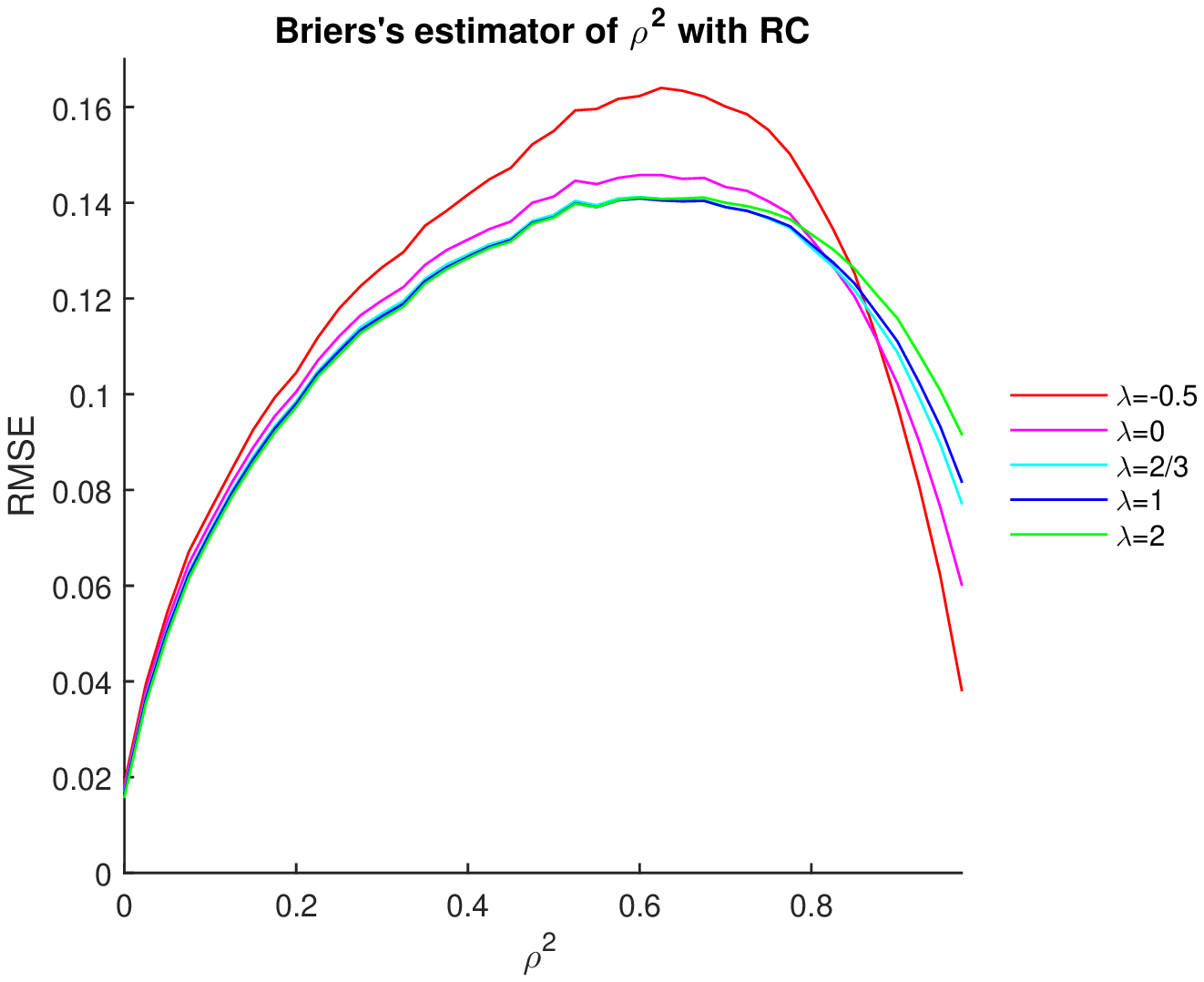}%
}
&
{\includegraphics[
height=2.6498in,
width=3.5284in
]%
{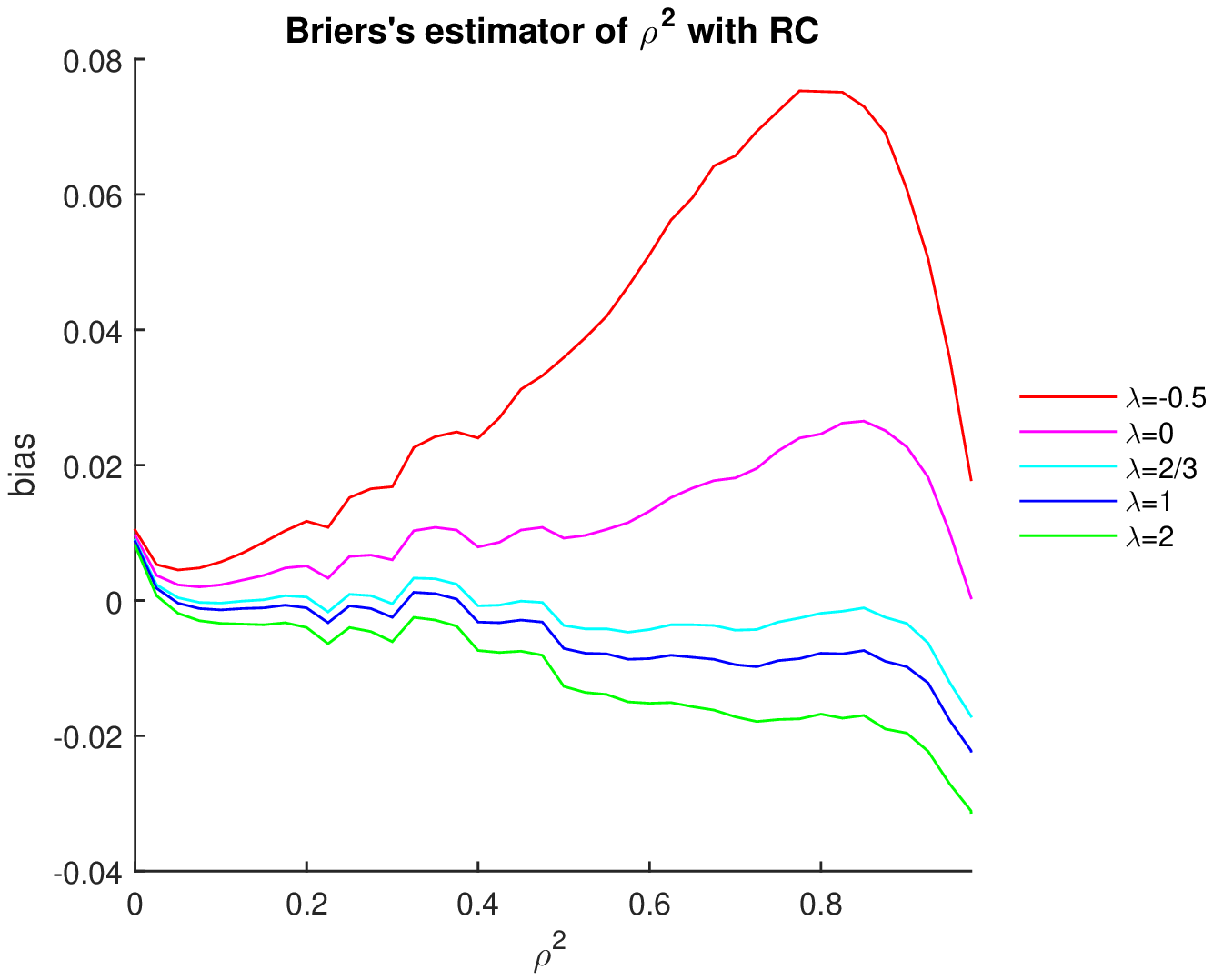}%
}
\end{tabular}
\caption{RMSE and bias of $\widetilde{\rho}_{\widehat{n}^{\ast},N,\lambda}^{2}$ with different values of $\lambda$ for DM (top), NI (middle), RC (bottom) distributions.\label{fig3}}%
\end{figure}%
\newpage

\subsection{Study on FBI data (Weir and Hill, 2002)}

Based on the FBI data study (Section \ref{FBI}), with theoretical values
obtained from the estimates of the probability vectors given in Table \ref{t3}
for loci D3S1358, vWA, FGA and D8S1179, the clustered multinomial
distributions are studied under no underlying assumption (saturated log-linear
model). Through Monte Carlo simulations, the \textrm{RMSE} and bias of the new
estimator proposed in Section \ref{Sec:new2} ($\widehat{\rho}^{2}$) and the
Weir and Hill estimator ($\overline{\rho}^{2}$) are compared in Figures
\ref{fig4.1}, \ref{fig4.2}, \ref{fig4.3}, \ref{fig4.4}, focused respectively
on the loci D3S1358, vWA, FGA and D8S1179. Since these kind of data have
usually small values of the intracluster correlation coefficient, $\rho^{2}$,
the study is\ only focussed on $\rho^{2}\in(0,0.1)$. Except for the RC
distribution, the bias of $\widehat{\rho}^{2}$ tends to be greater than the
bias of $\overline{\rho}^{2}$, however, the \textrm{RMSE}\ of $\widehat{\rho
}^{2}$ tends to be smaller than the \textrm{RMSE} of $\overline{\rho}^{2}$.
This weakness of the bias could be improved in case of being able to identify
an apropriate log-linear model.%

\begin{figure}[htbp]  \centering
\begin{tabular}
[c]{cc}%
{\includegraphics[
height=2.6498in,
width=3.5284in
]%
{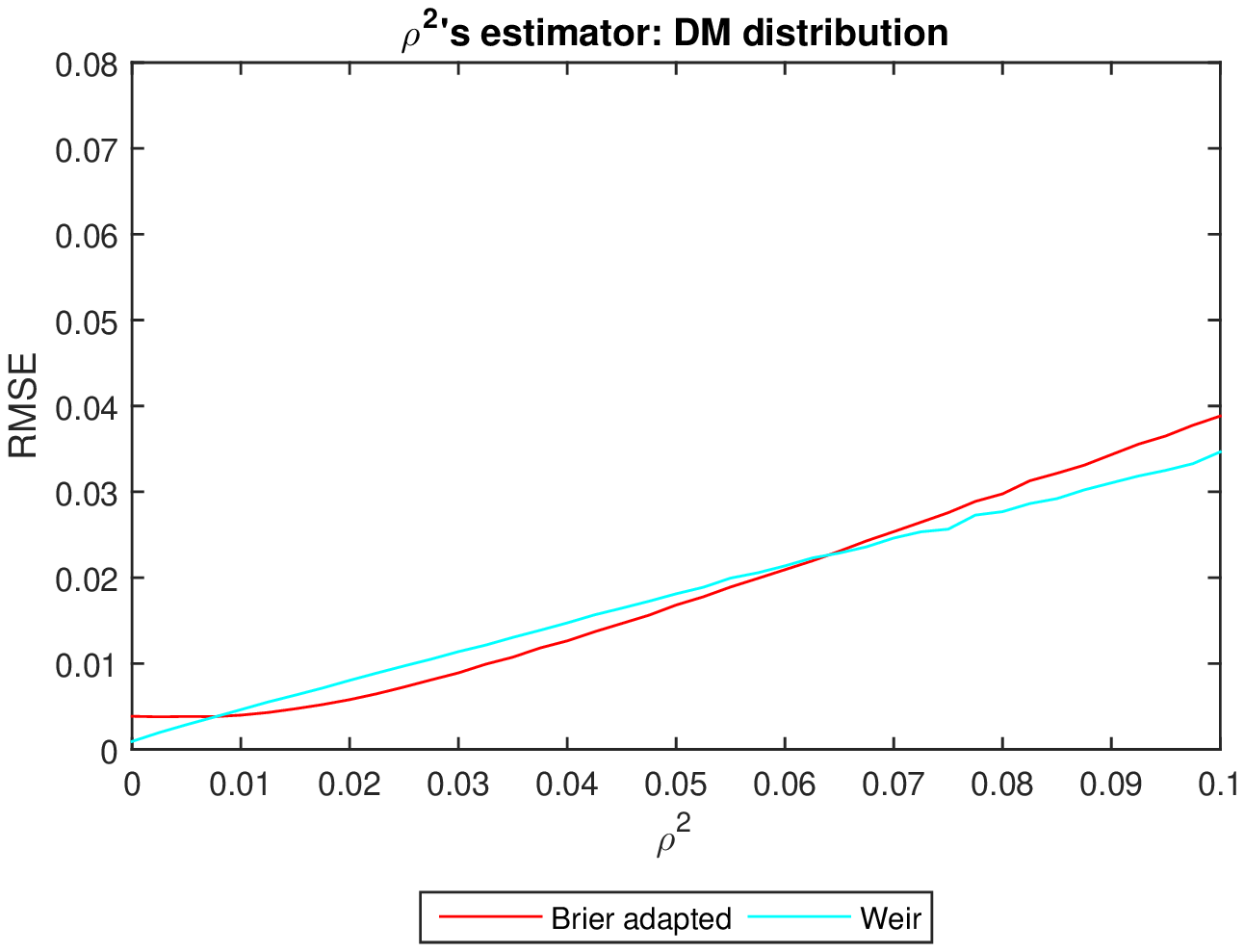}%
}
&
{\includegraphics[
height=2.6498in,
width=3.5284in
]%
{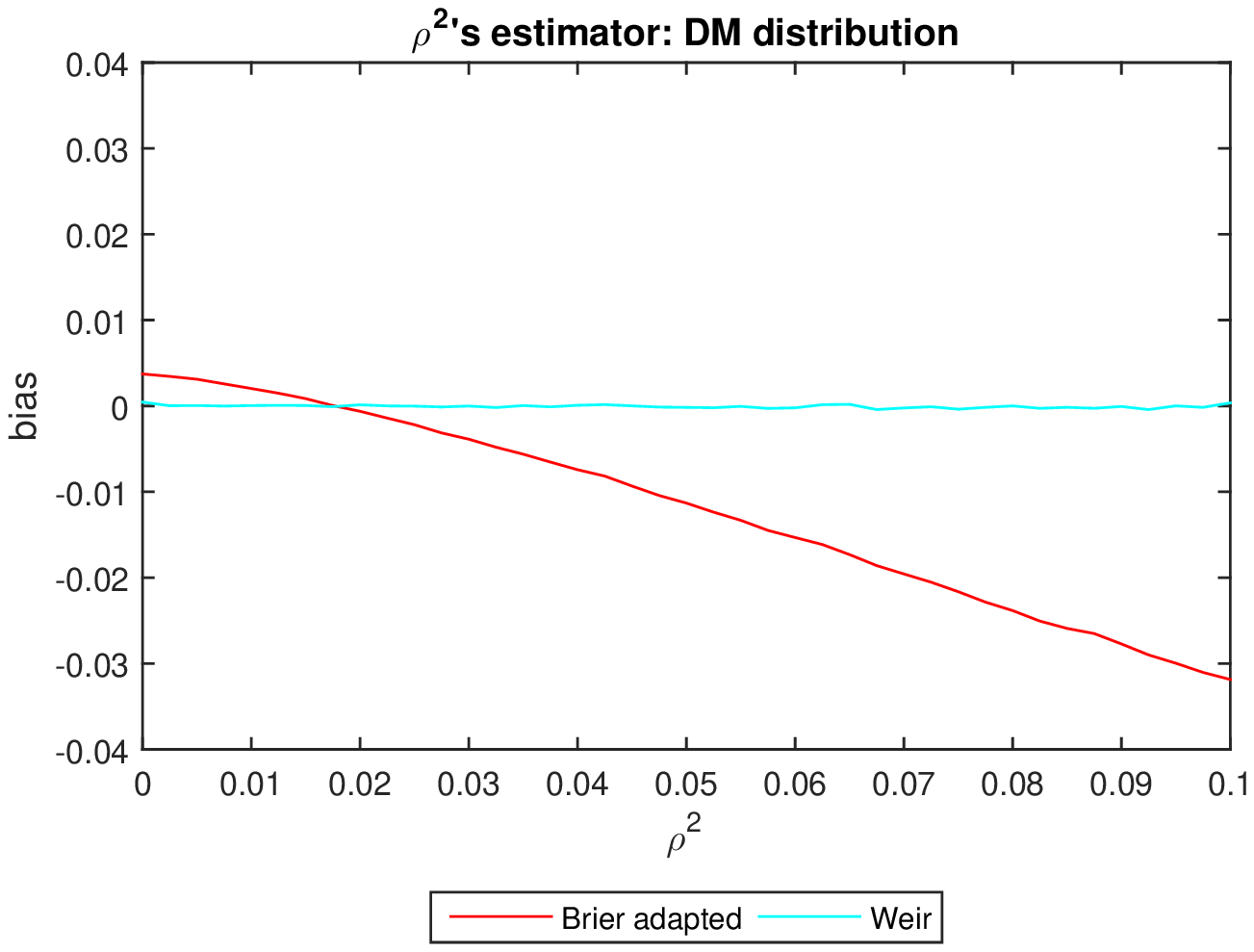}%
}
\\%
{\includegraphics[
height=2.6498in,
width=3.5284in
]%
{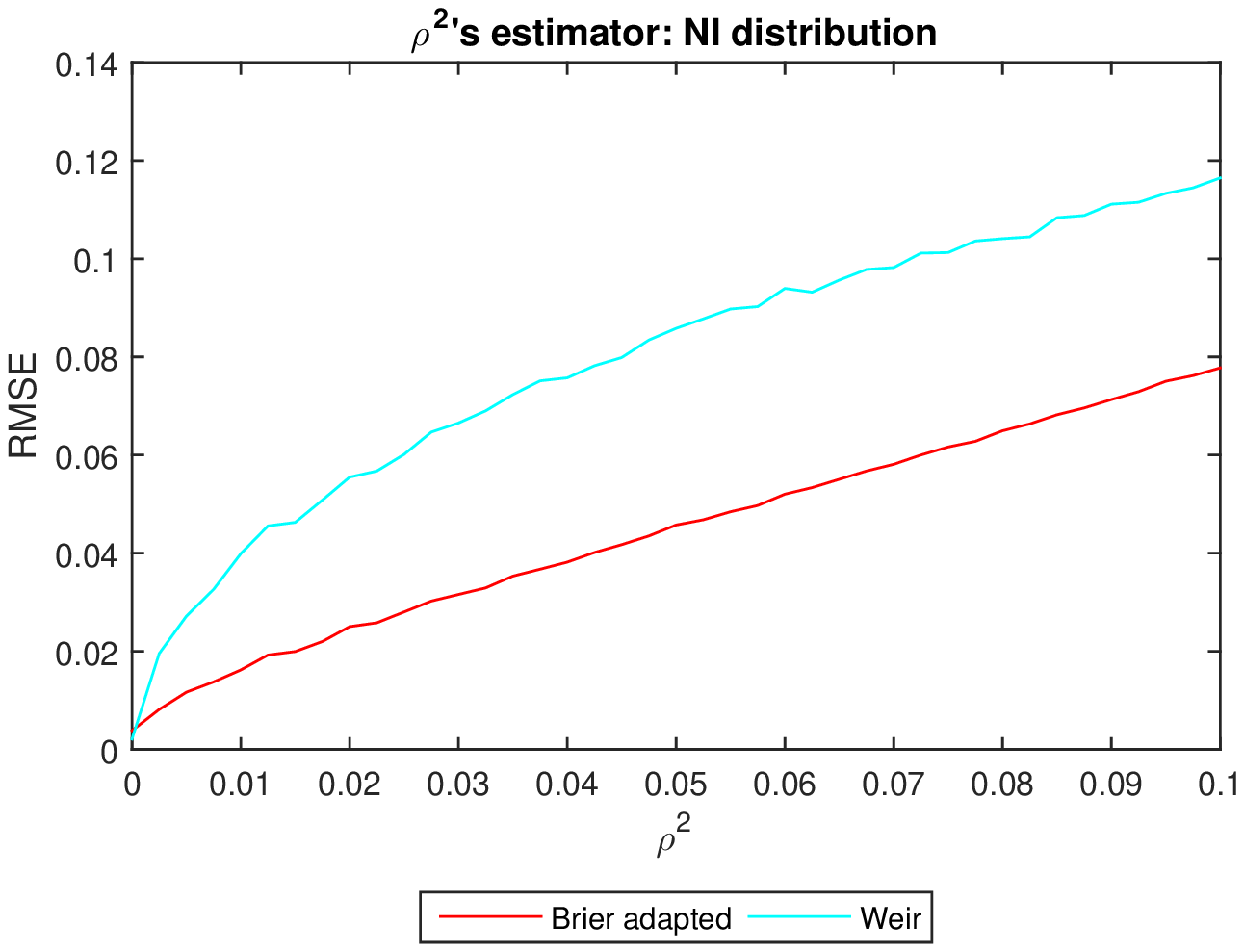}%
}
&
{\includegraphics[
height=2.6498in,
width=3.5284in
]%
{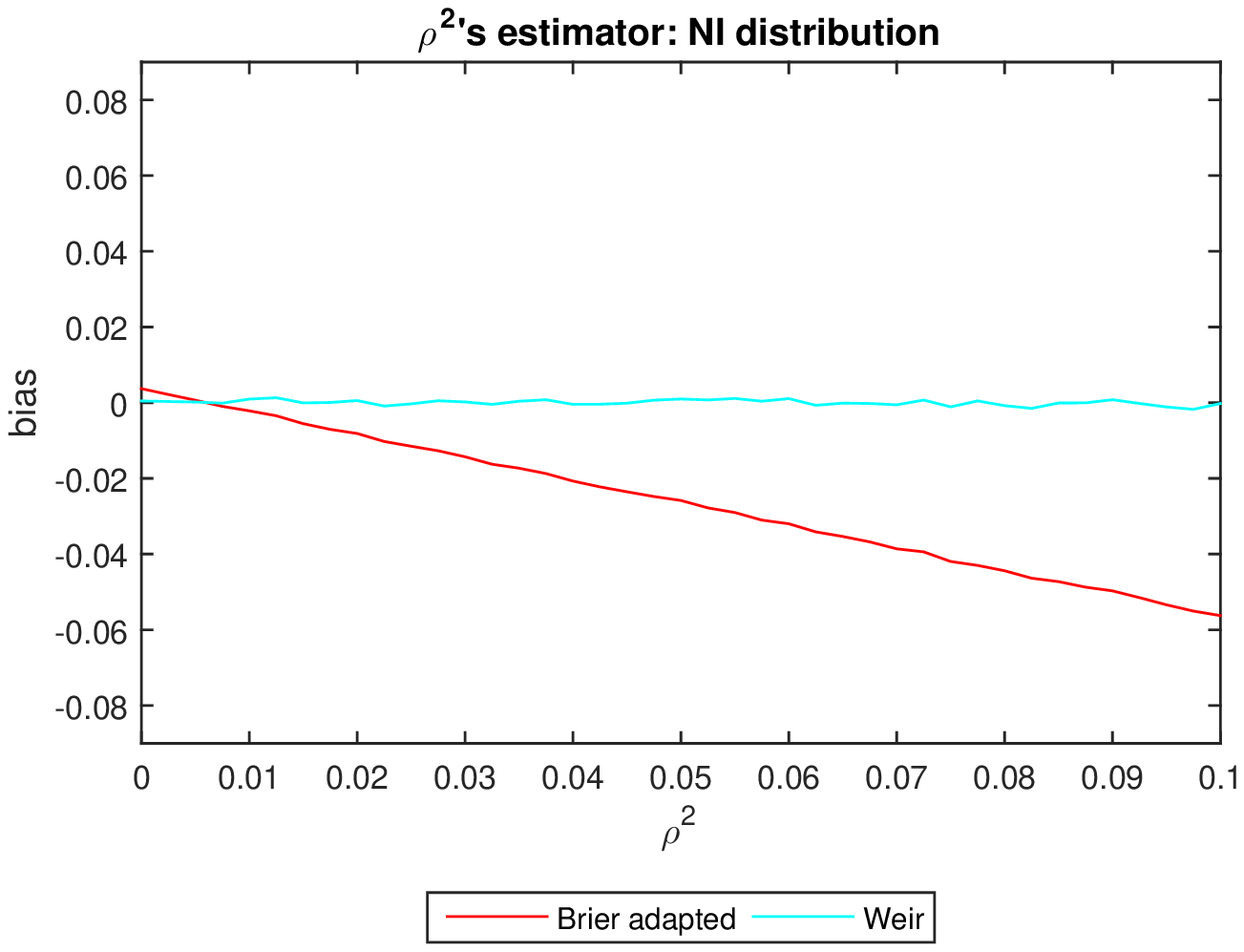}%
}
\\%
{\includegraphics[
height=2.6498in,
width=3.5284in
]%
{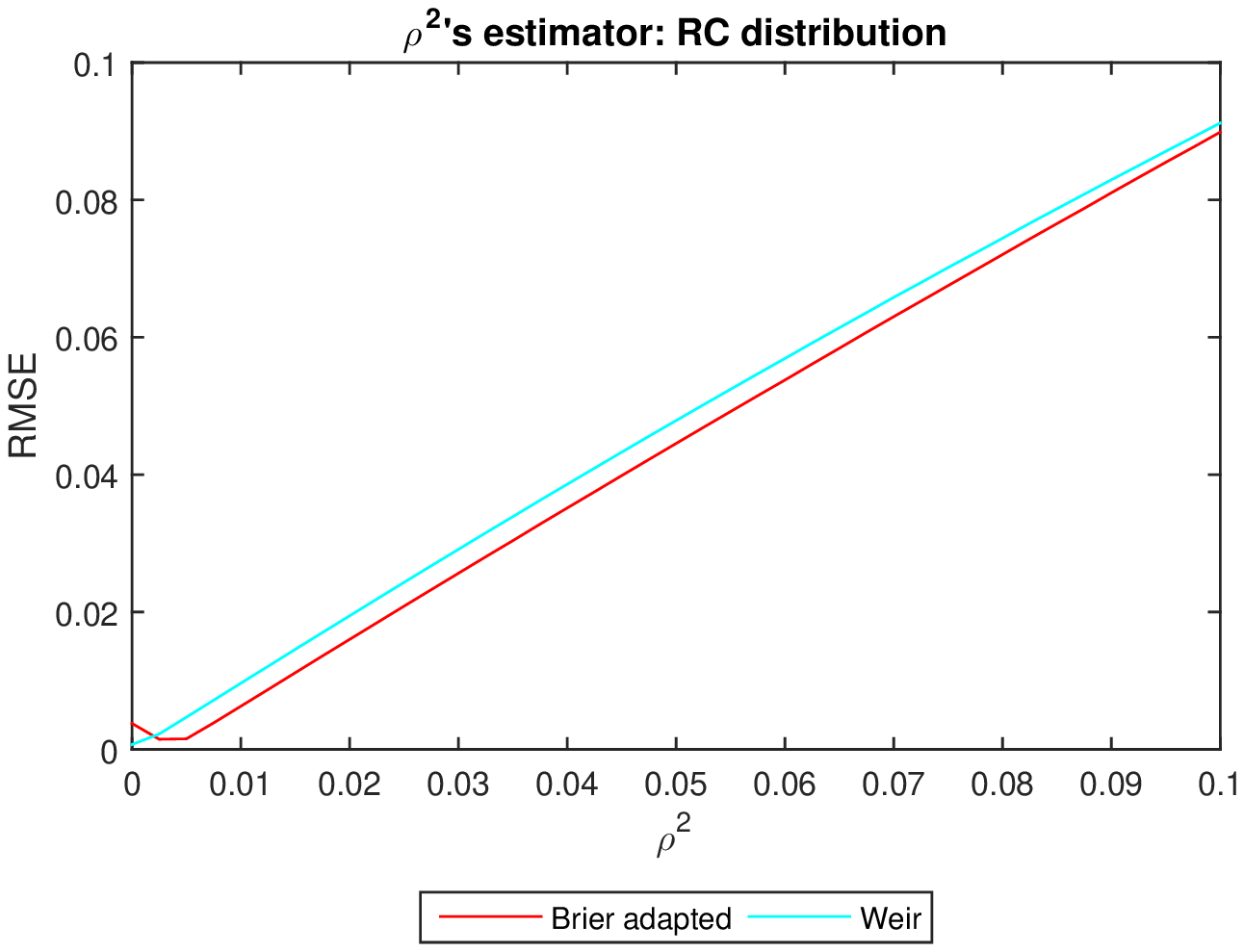}%
}
&
{\includegraphics[
height=2.6498in,
width=3.5284in
]%
{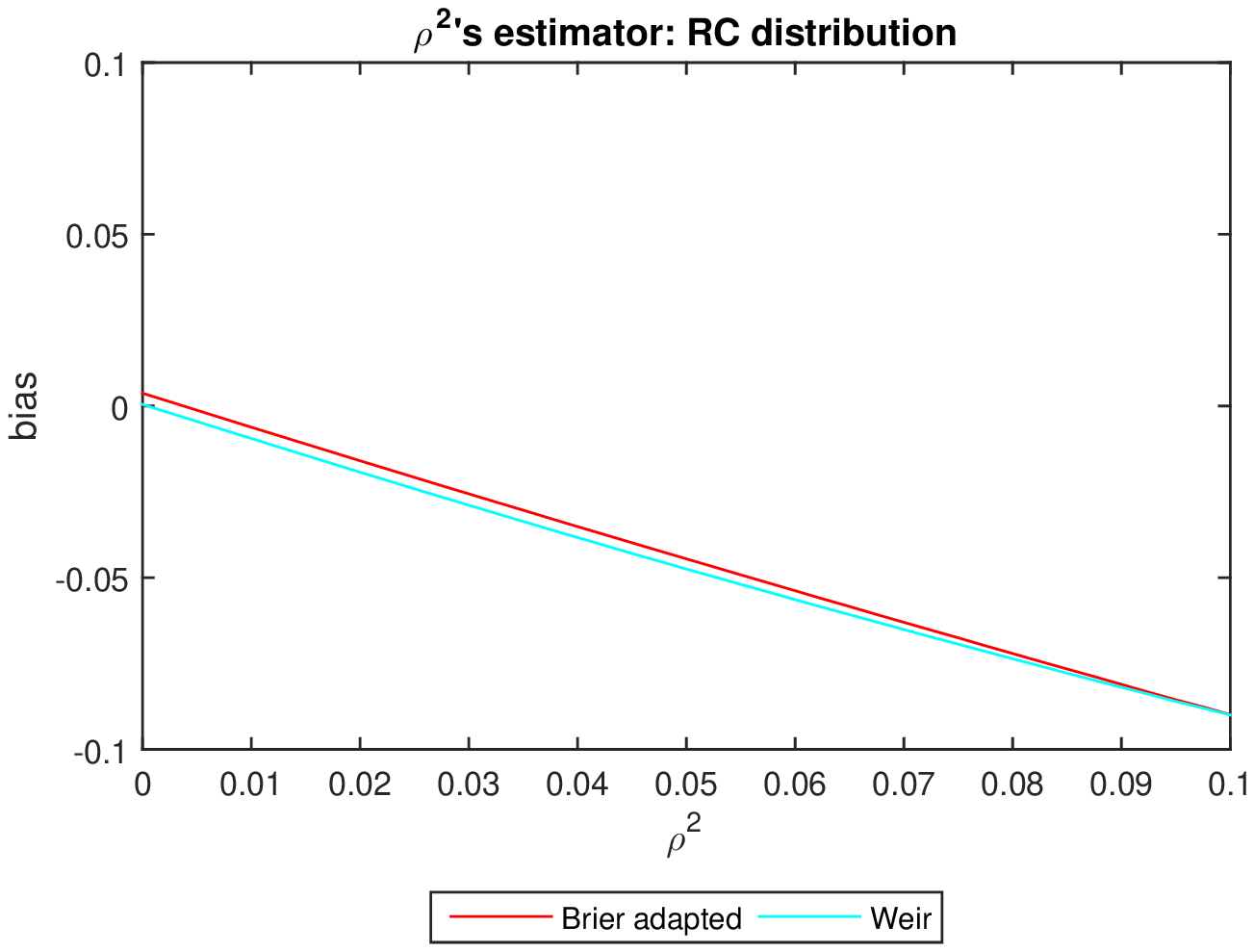}%
}
\end{tabular}
\caption{RMSE and bias of the Brier adapted $\widehat{\rho}^{2}$ and Weir's $\overline{\rho}^{2}$ for small values of $\rho^2$ when DM, NI and RC distributions are considered and the theoretical probabilities are equal to the estimates  for locus D3S1358.\label{fig4.1}}%
\end{figure}%
%

\begin{figure}[htbp]  \centering
\begin{tabular}
[c]{cc}%
{\includegraphics[
height=2.6498in,
width=3.5284in
]%
{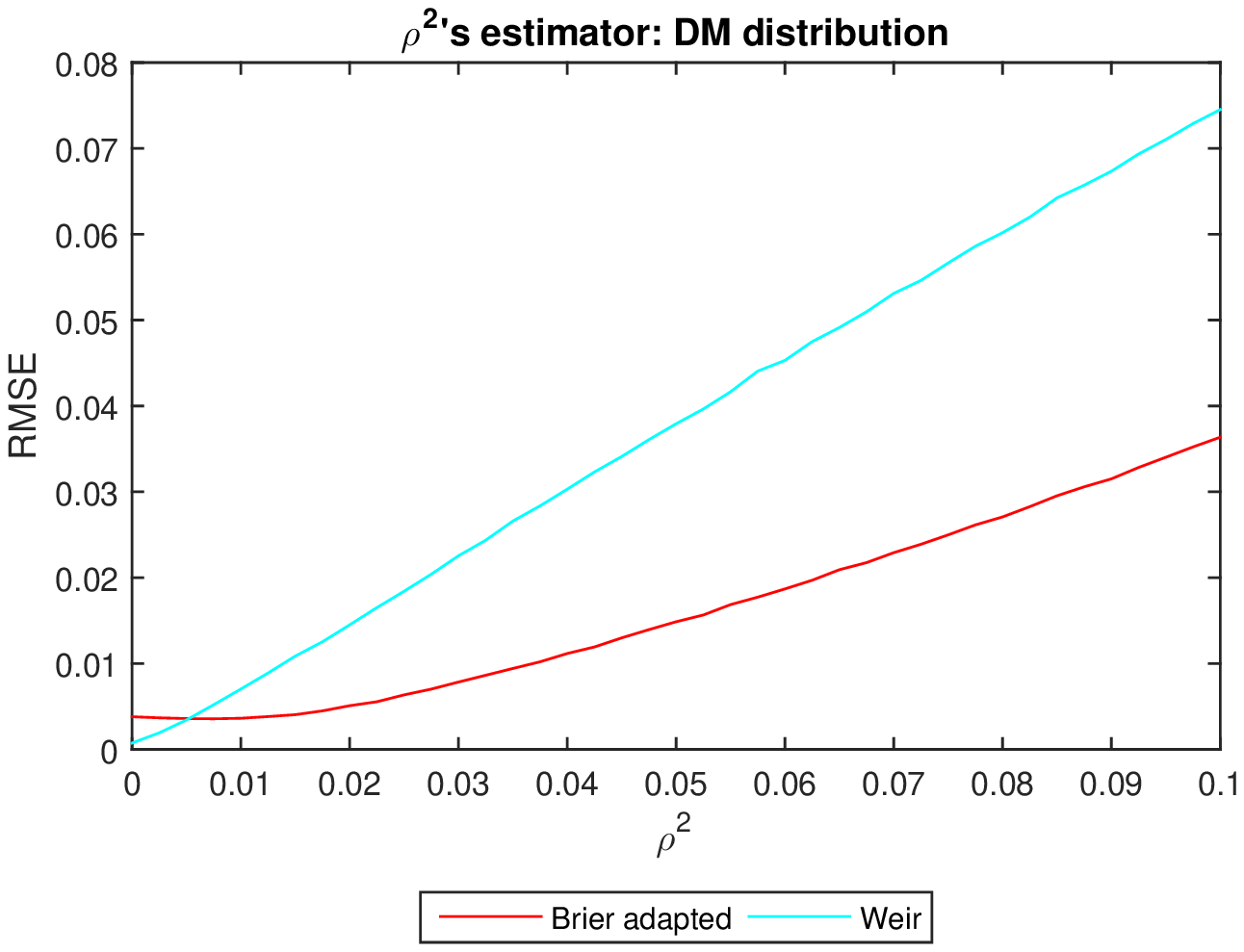}%
}
&
{\includegraphics[
height=2.6498in,
width=3.5284in
]%
{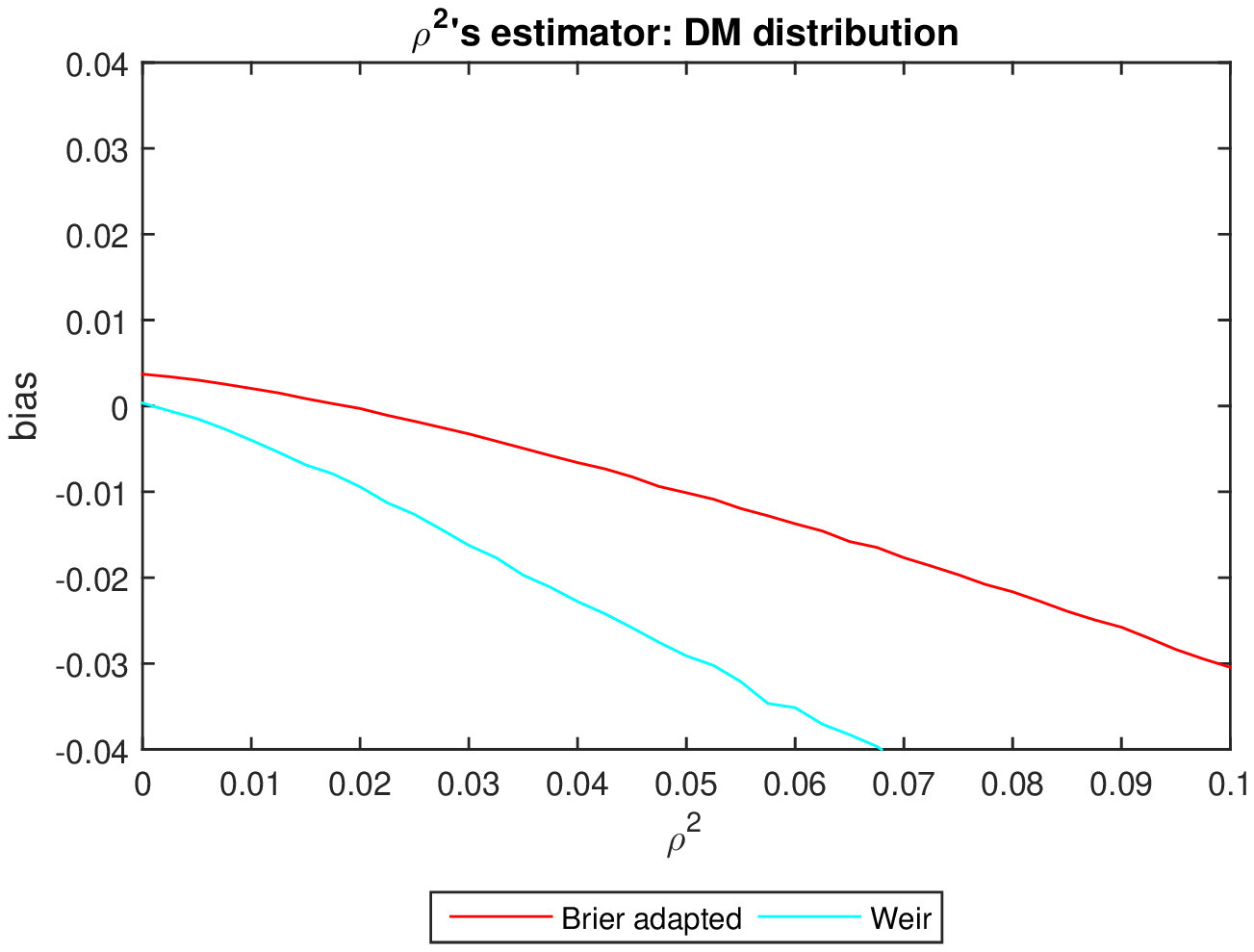}%
}
\\%
{\includegraphics[
height=2.6498in,
width=3.5284in
]%
{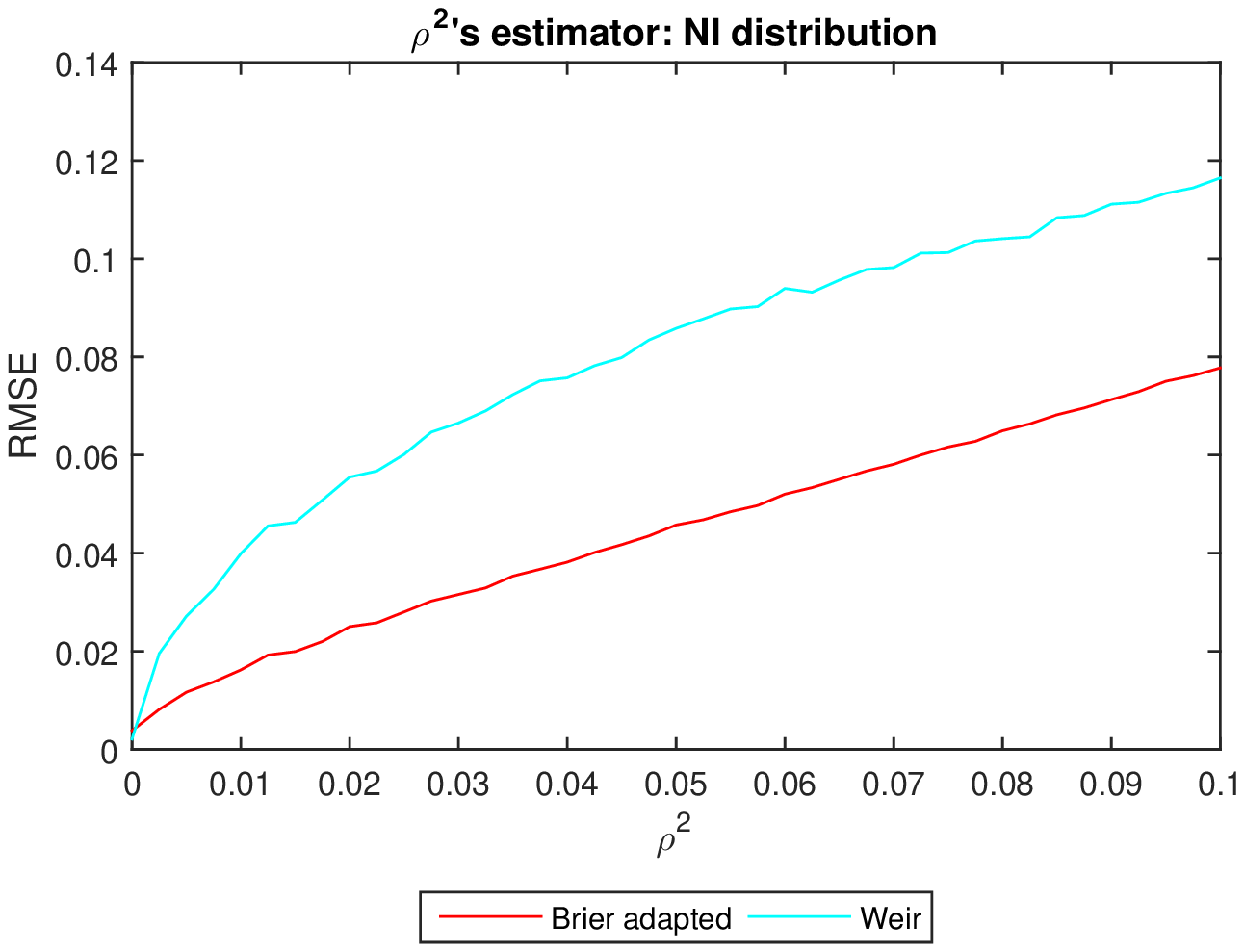}%
}
&
{\includegraphics[
height=2.6498in,
width=3.5284in
]%
{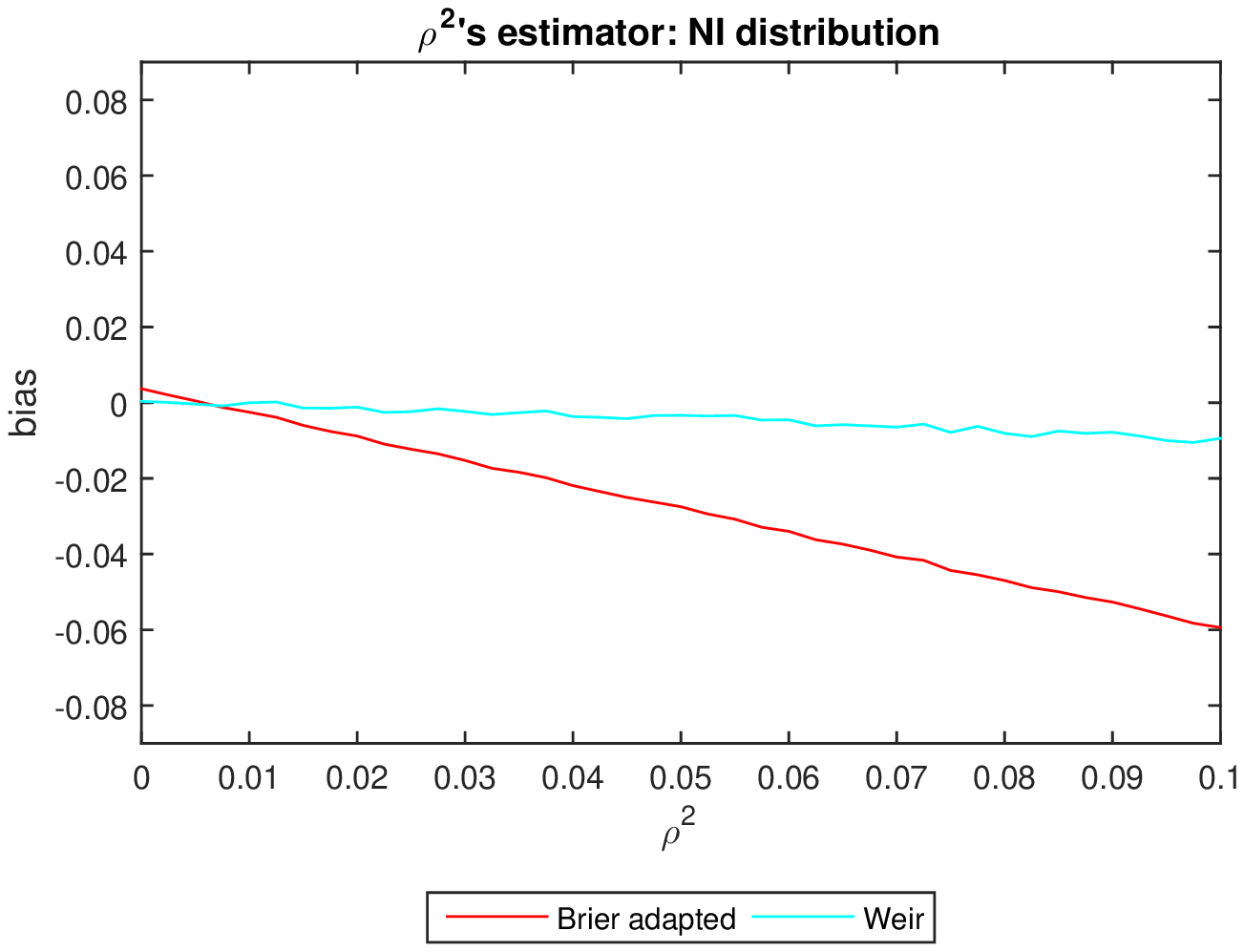}%
}
\\%
{\includegraphics[
height=2.6498in,
width=3.5284in
]%
{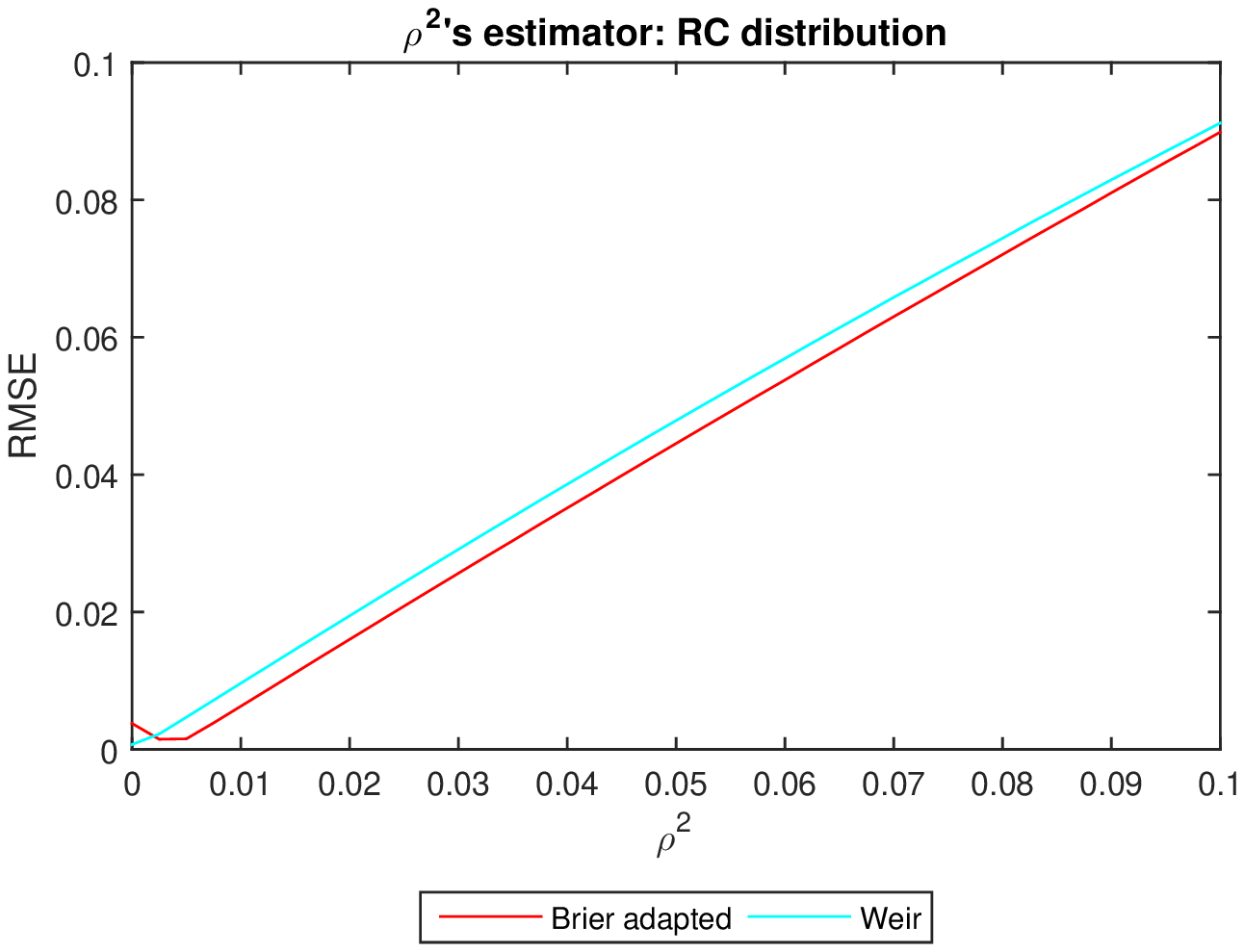}%
}
&
{\includegraphics[
height=2.6498in,
width=3.5284in
]%
{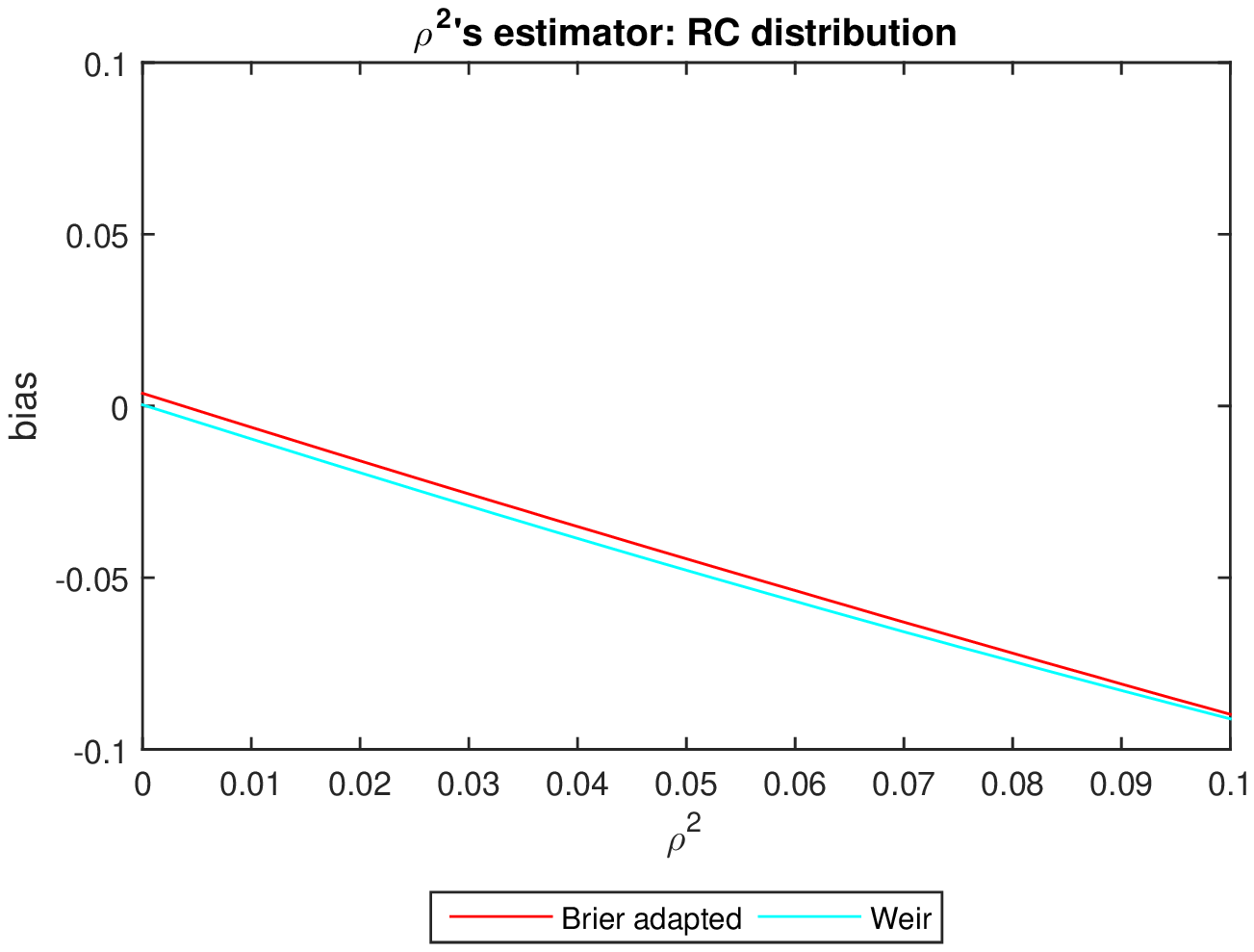}%
}
\end{tabular}
\caption{RMSE and bias of the Brier adapted $\widehat{\rho}^{2}$ and Weir's $\overline{\rho}^{2}$ for small values of $\rho^2$ when DM, NI and RC distributions are considered and the theoretical probabilities are equal to the estimates  for locus vWA.\label{fig4.2}}%
\end{figure}%
%

\begin{figure}[htbp]  \centering
\begin{tabular}
[c]{cc}%
{\includegraphics[
height=2.6498in,
width=3.5284in
]%
{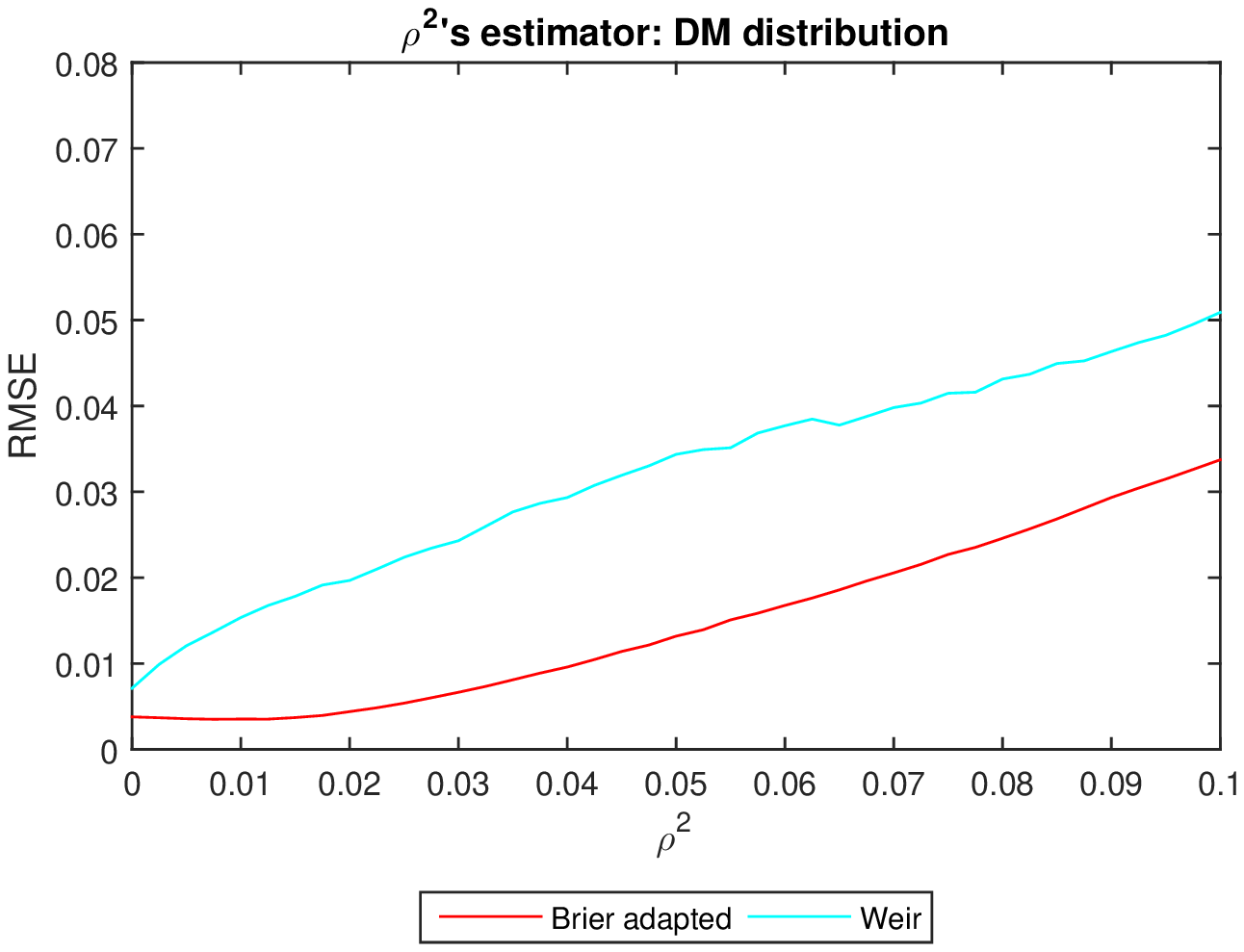}%
}
&
{\includegraphics[
height=2.6498in,
width=3.5284in
]%
{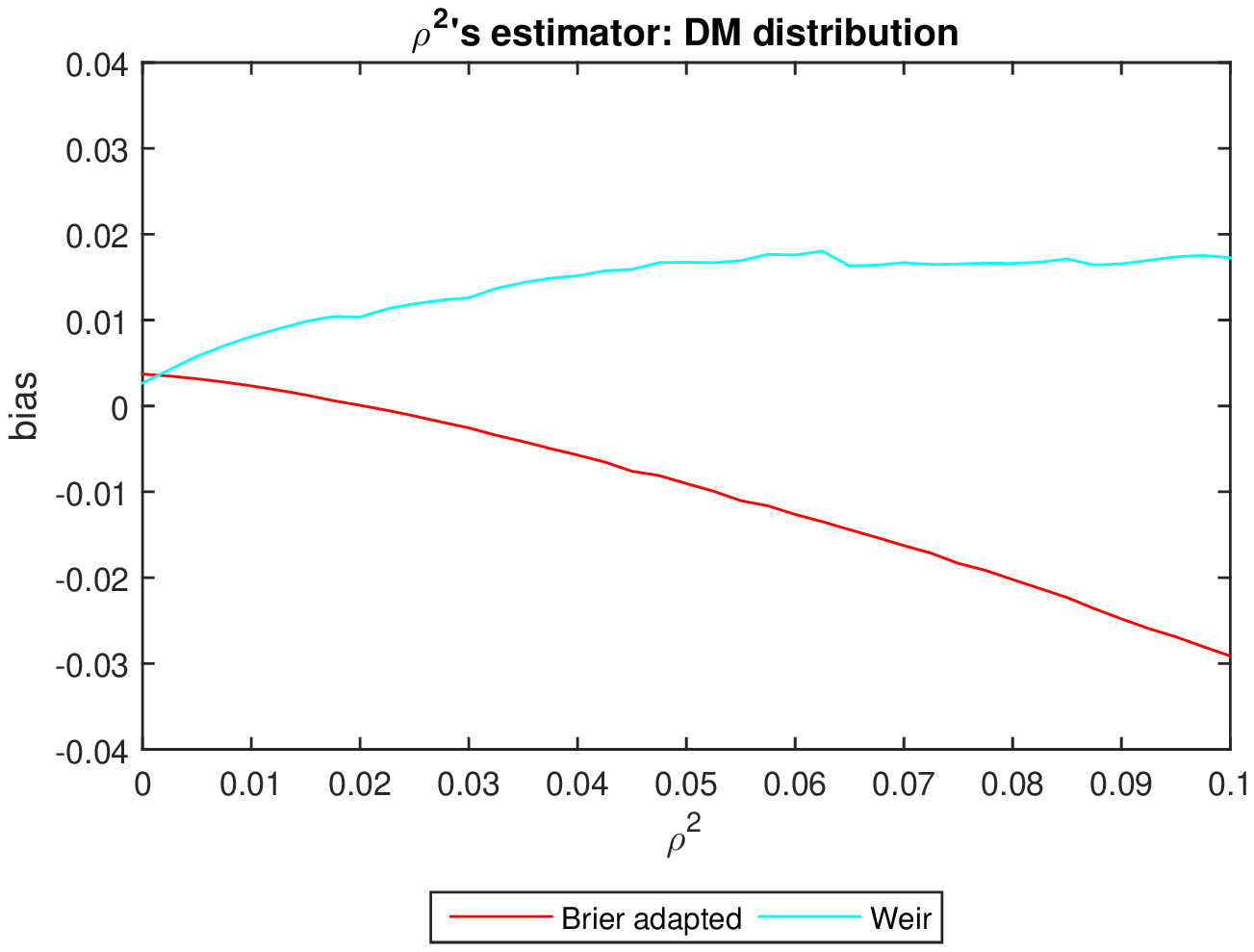}%
}
\\%
{\includegraphics[
height=2.6498in,
width=3.5284in
]%
{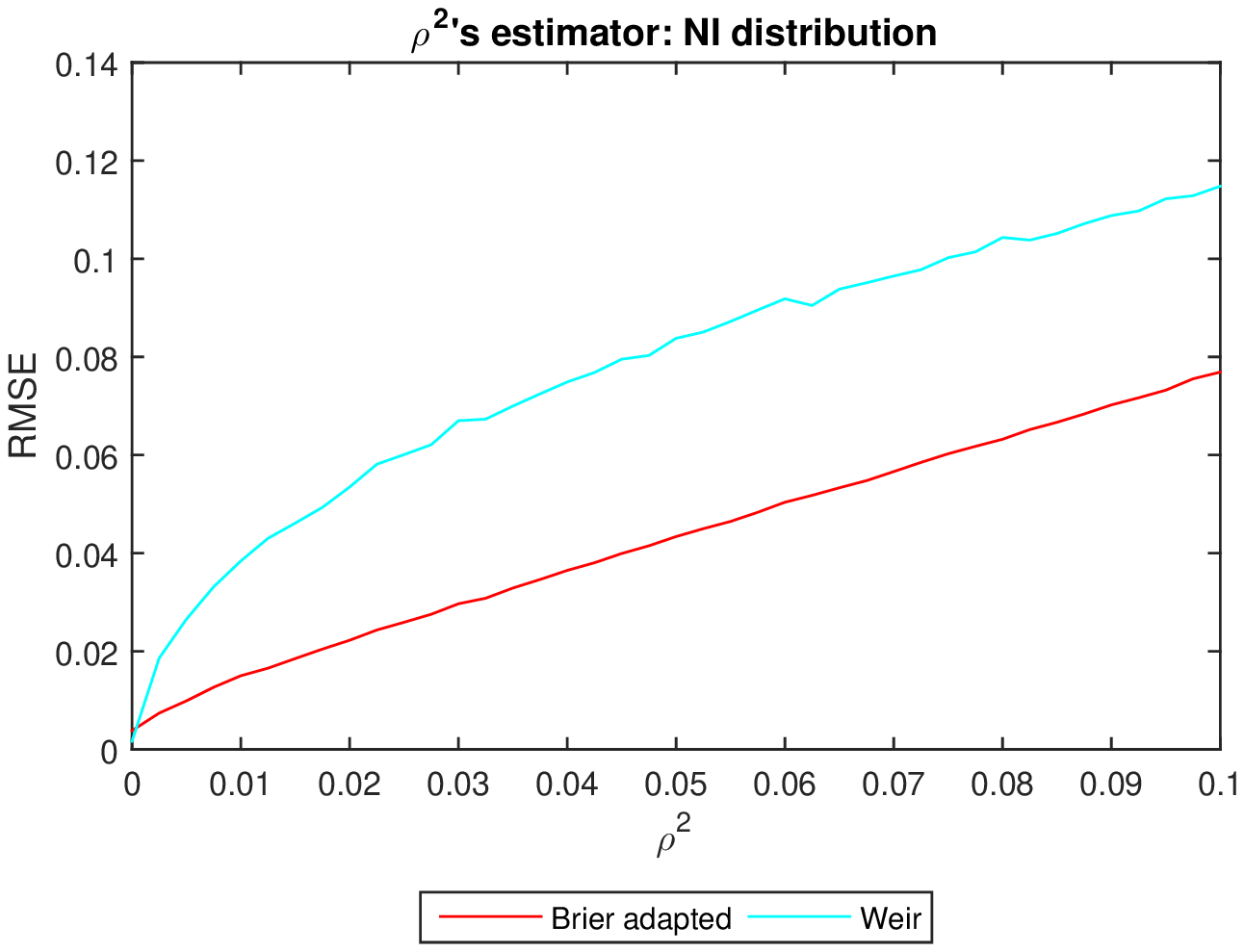}%
}
&
{\includegraphics[
height=2.6498in,
width=3.5284in
]%
{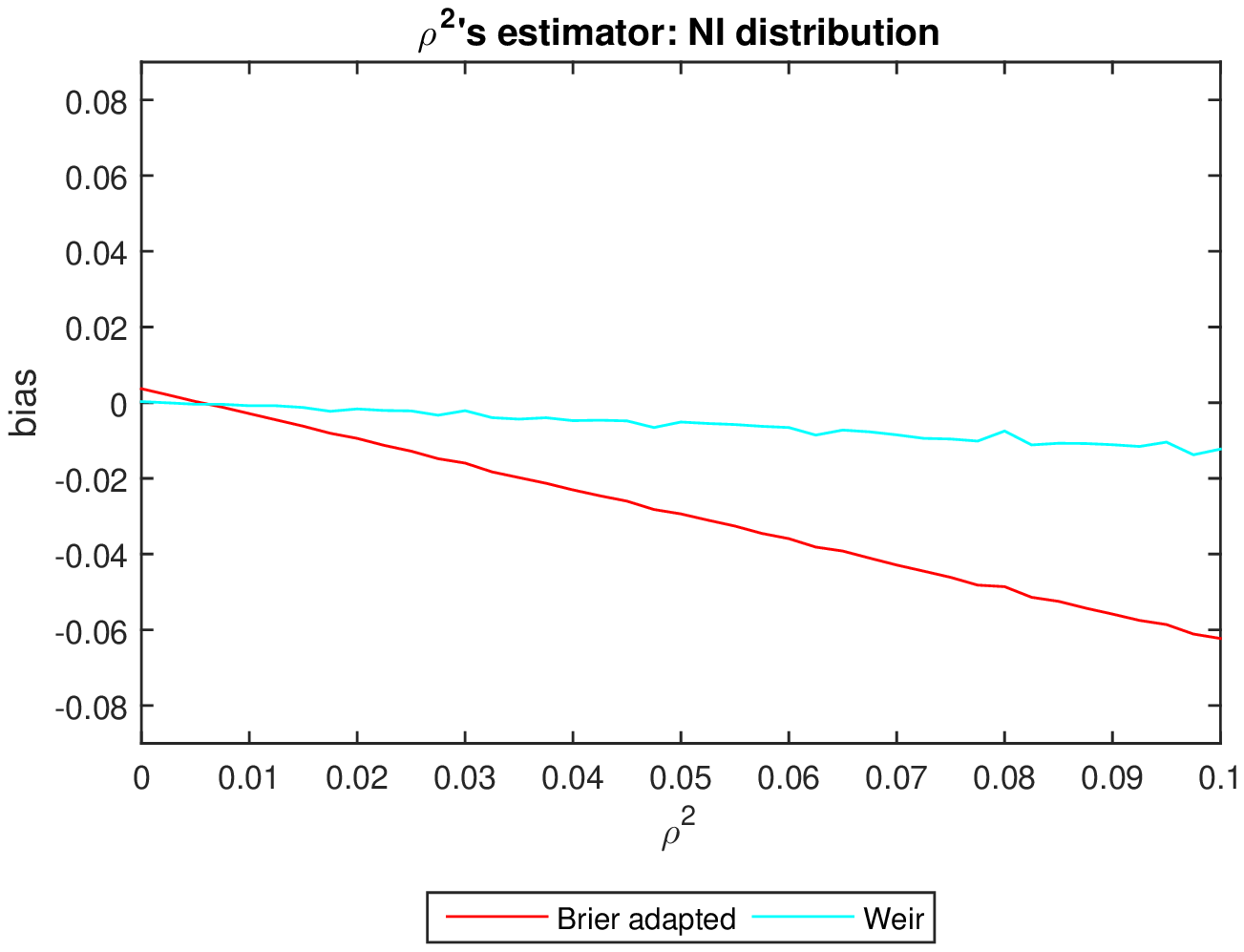}%
}
\\%
{\includegraphics[
height=2.6567in,
width=3.5284in
]%
{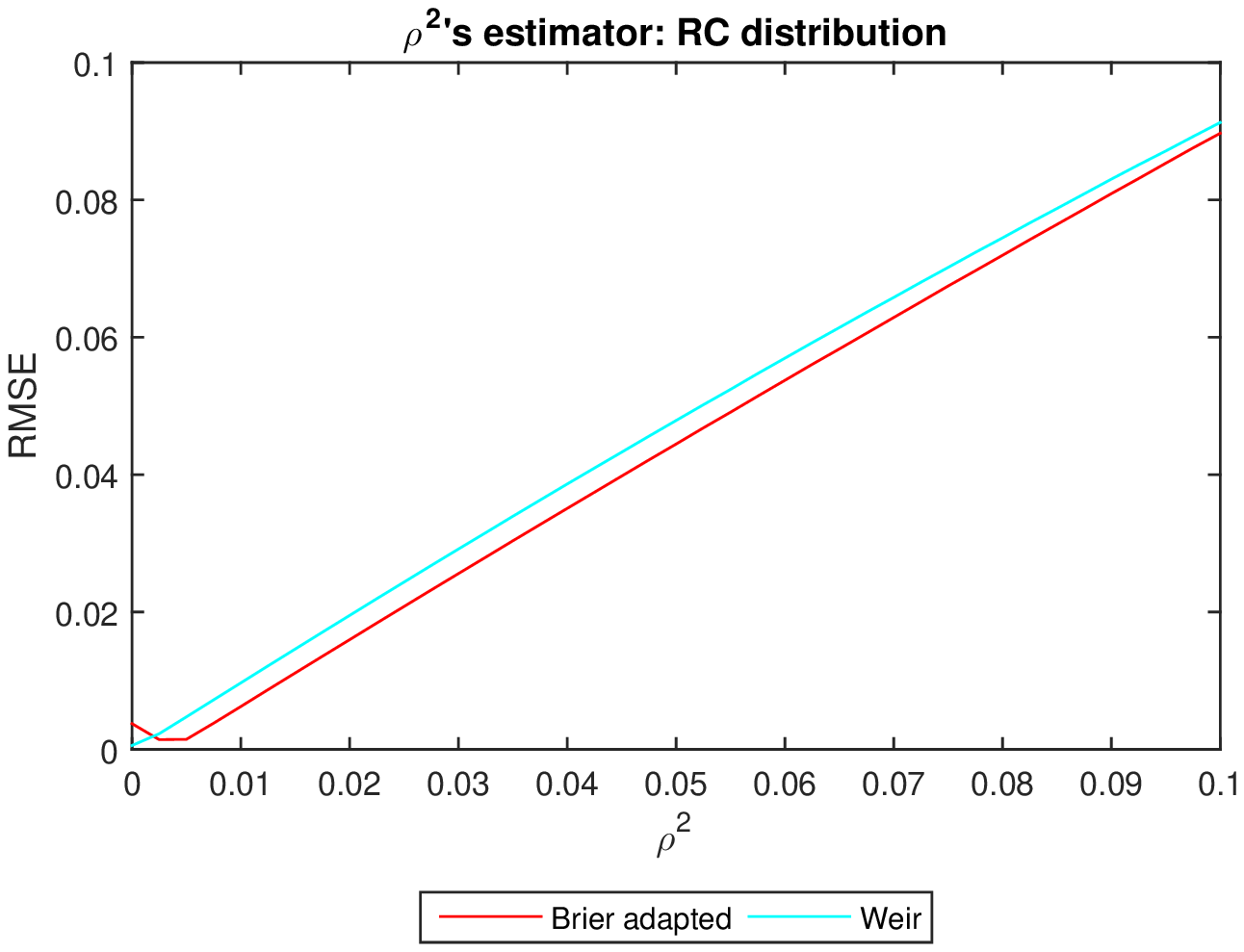}%
}
&
{\includegraphics[
height=2.6498in,
width=3.5284in
]%
{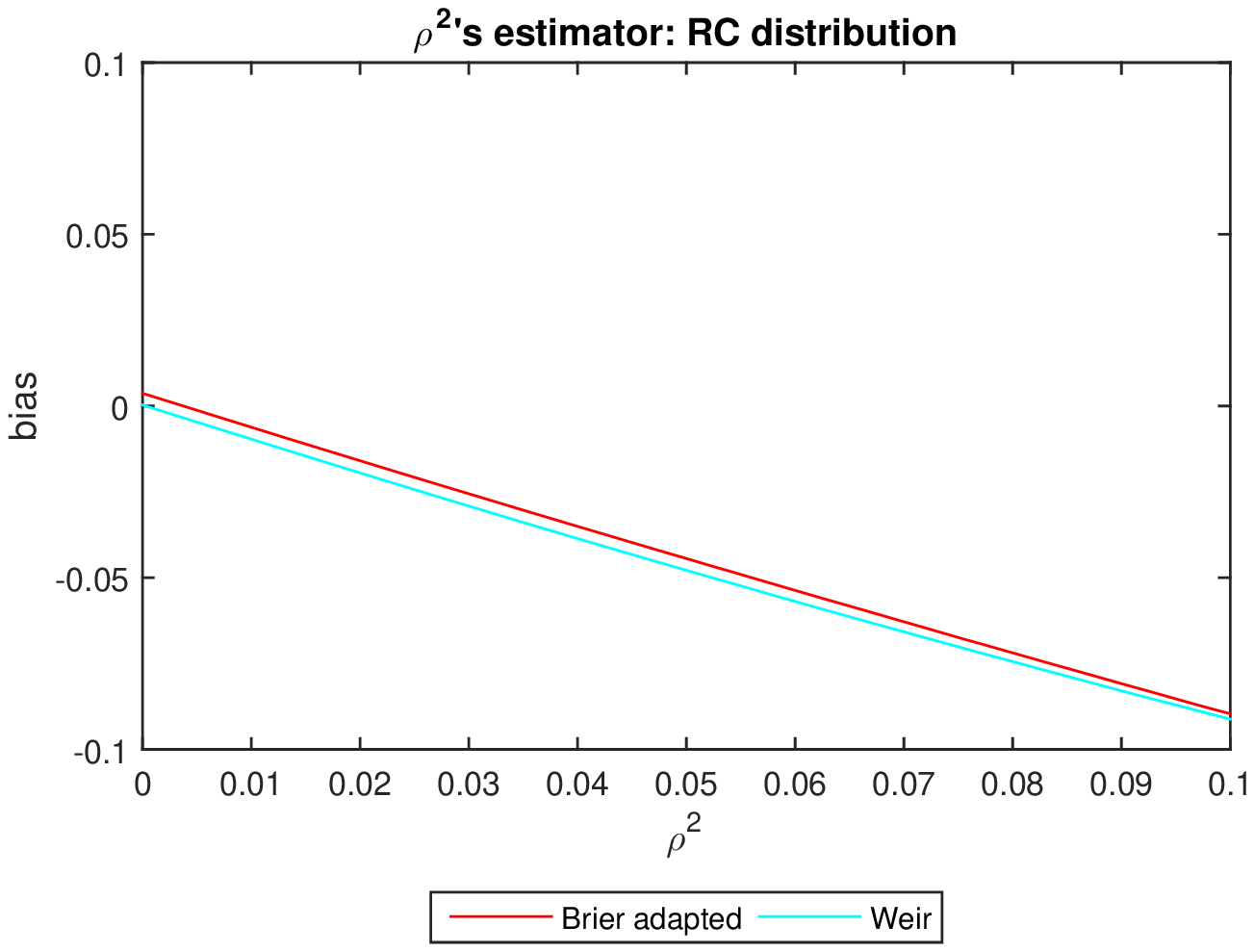}%
}
\end{tabular}
\caption{RMSE and bias of the Brier adapted $\widehat{\rho}^{2}$ and Weir's $\overline{\rho}^{2}$ for small values of $\rho^2$ when DM, NI and RC distributions are considered and the theoretical probabilities are equal to the estimates  for locus FGA.\label{fig4.3}}%
\end{figure}%
%

\begin{figure}[htbp]  \centering
\begin{tabular}
[c]{cc}%
{\includegraphics[
height=2.6498in,
width=3.5284in
]%
{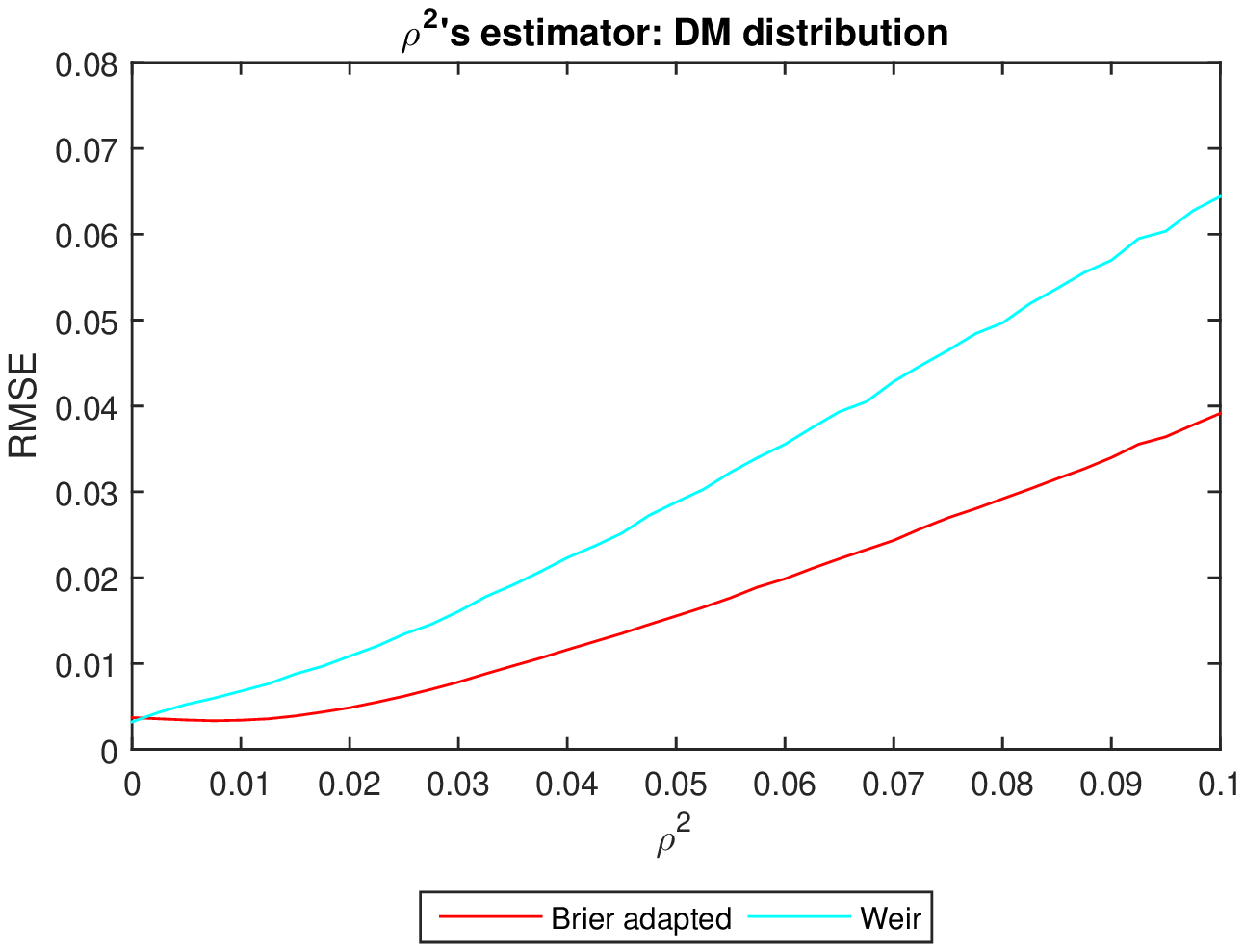}%
}
&
{\includegraphics[
height=2.6498in,
width=3.5284in
]%
{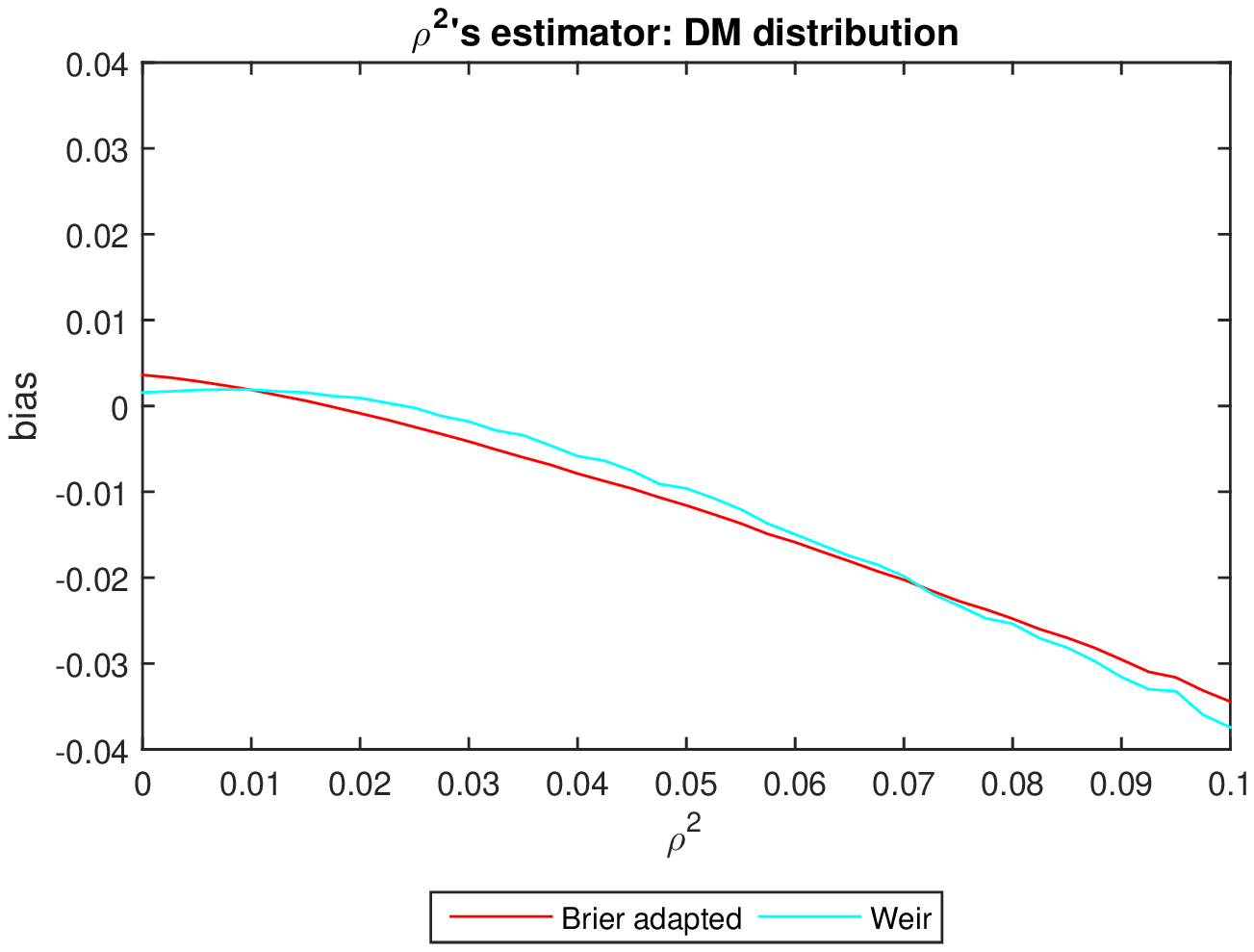}%
}
\\%
{\includegraphics[
height=2.6498in,
width=3.5284in
]%
{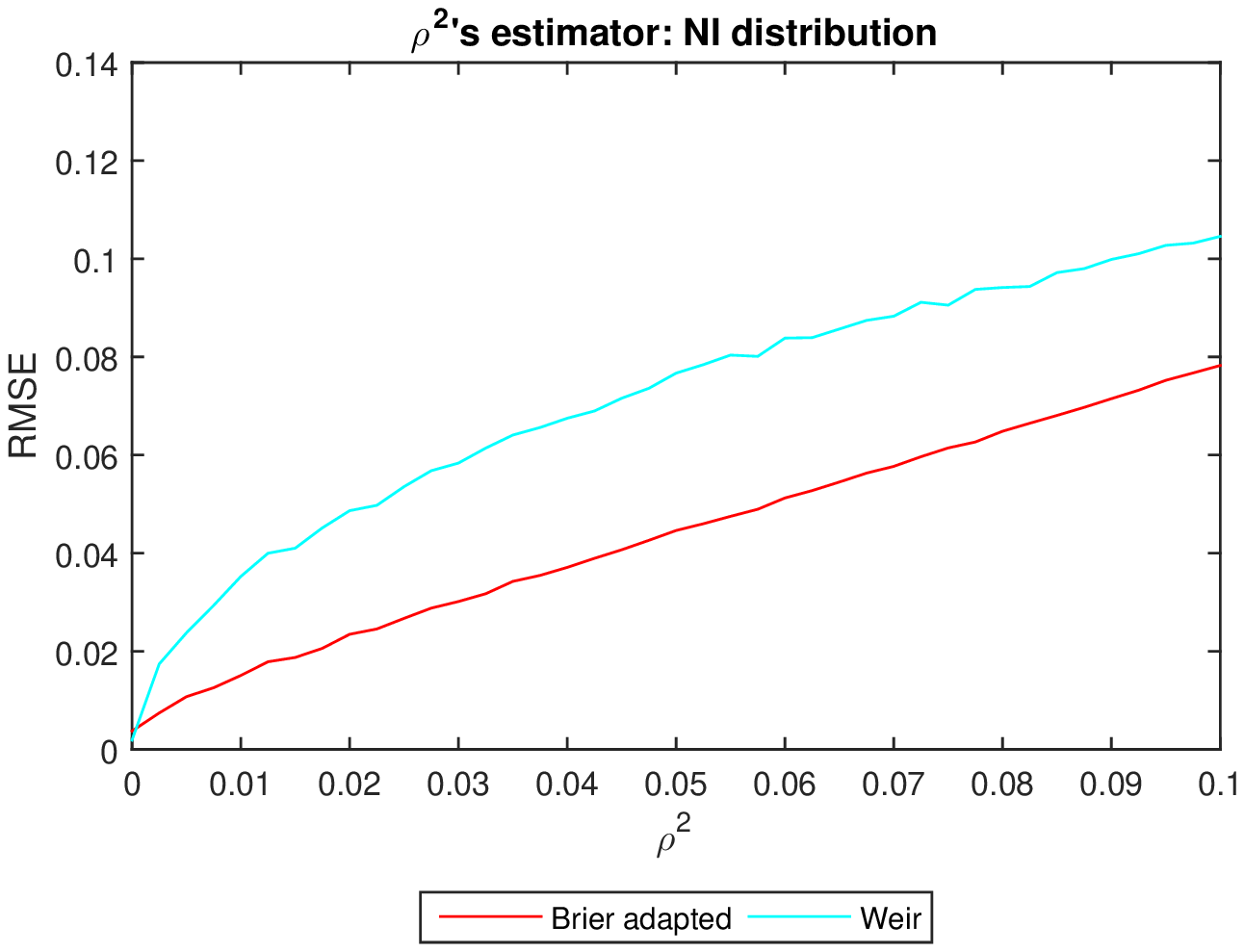}%
}
&
{\includegraphics[
height=2.6498in,
width=3.5284in
]%
{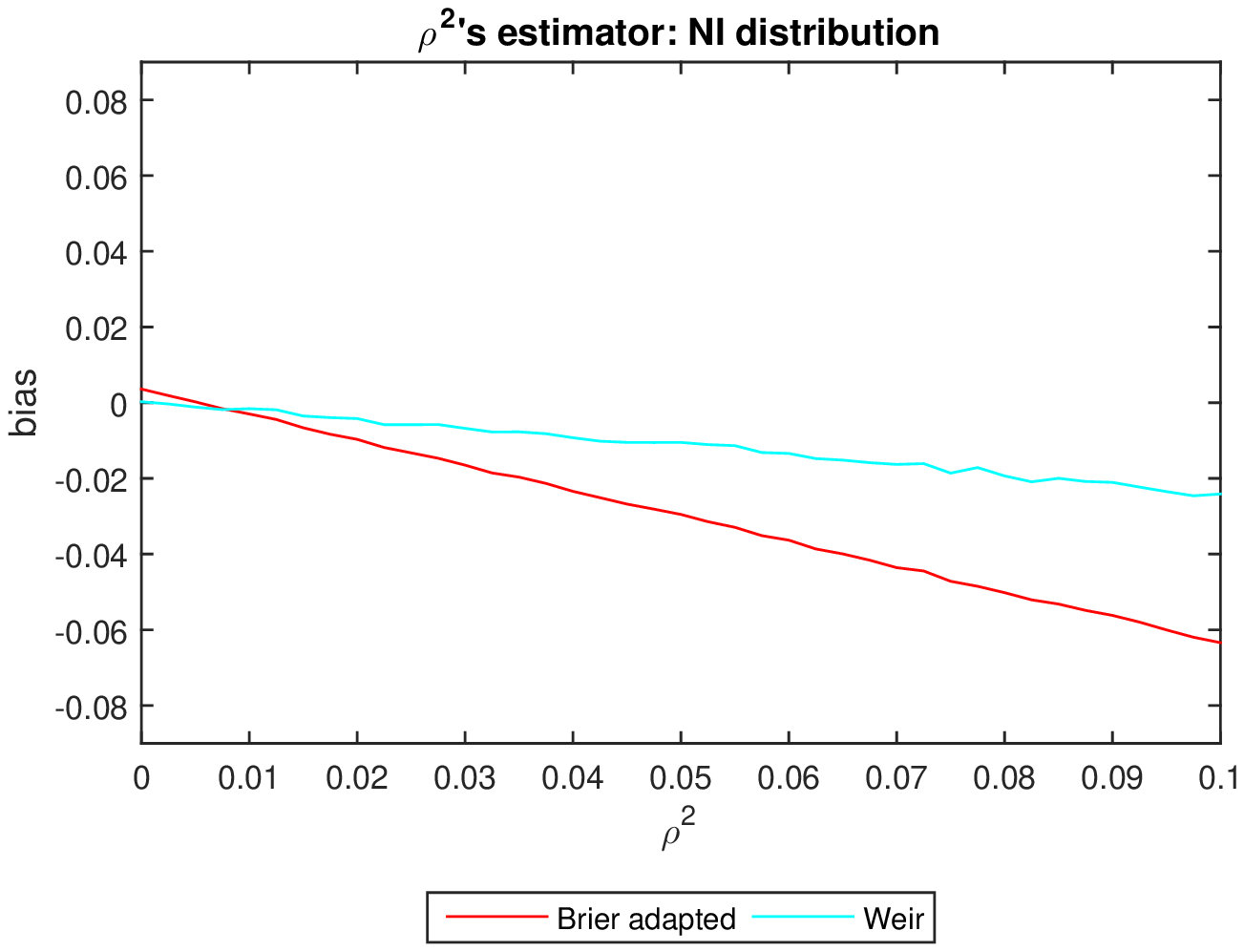}%
}
\\%
{\includegraphics[
height=2.6498in,
width=3.5284in
]%
{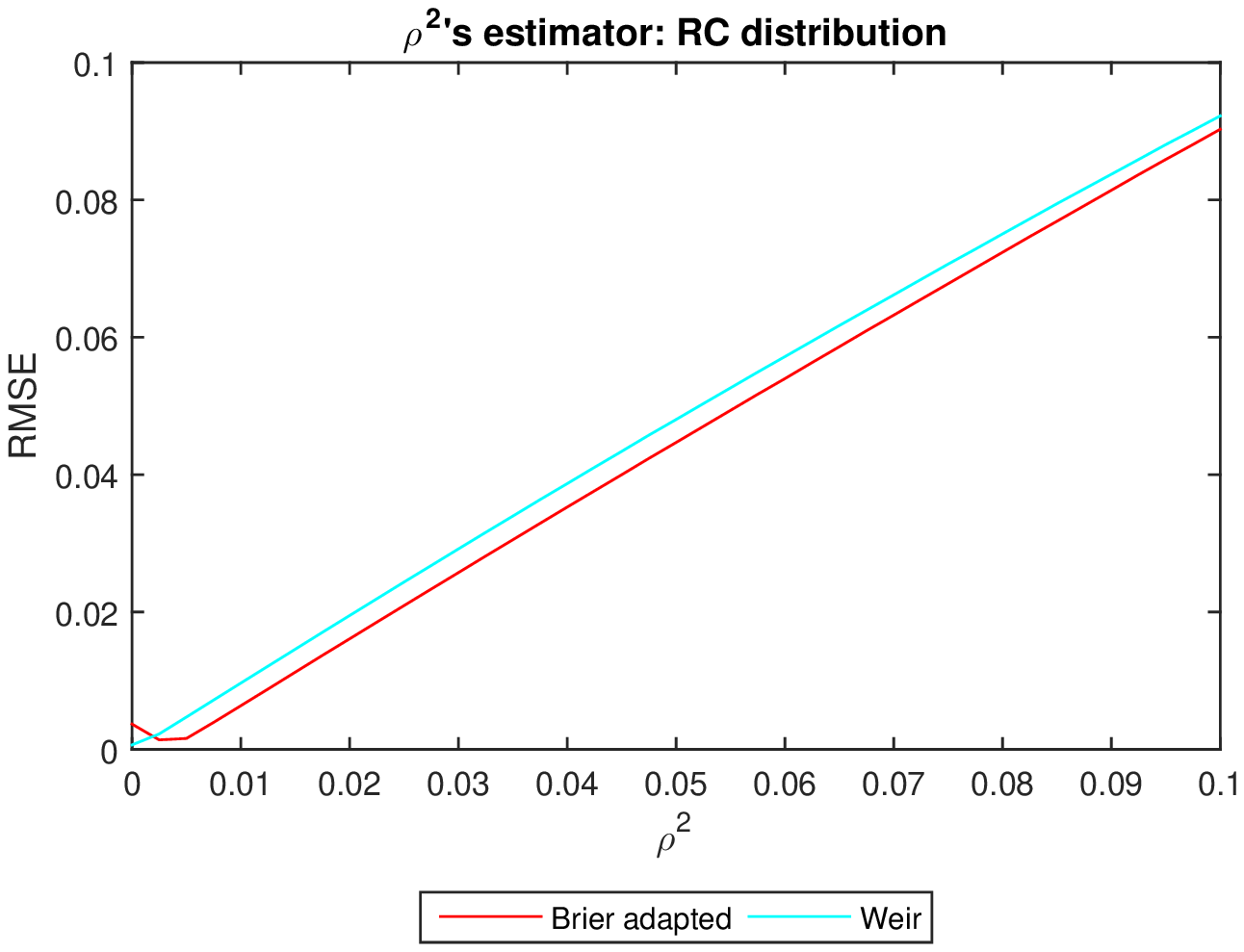}%
}
&
{\includegraphics[
height=2.6498in,
width=3.5284in
]%
{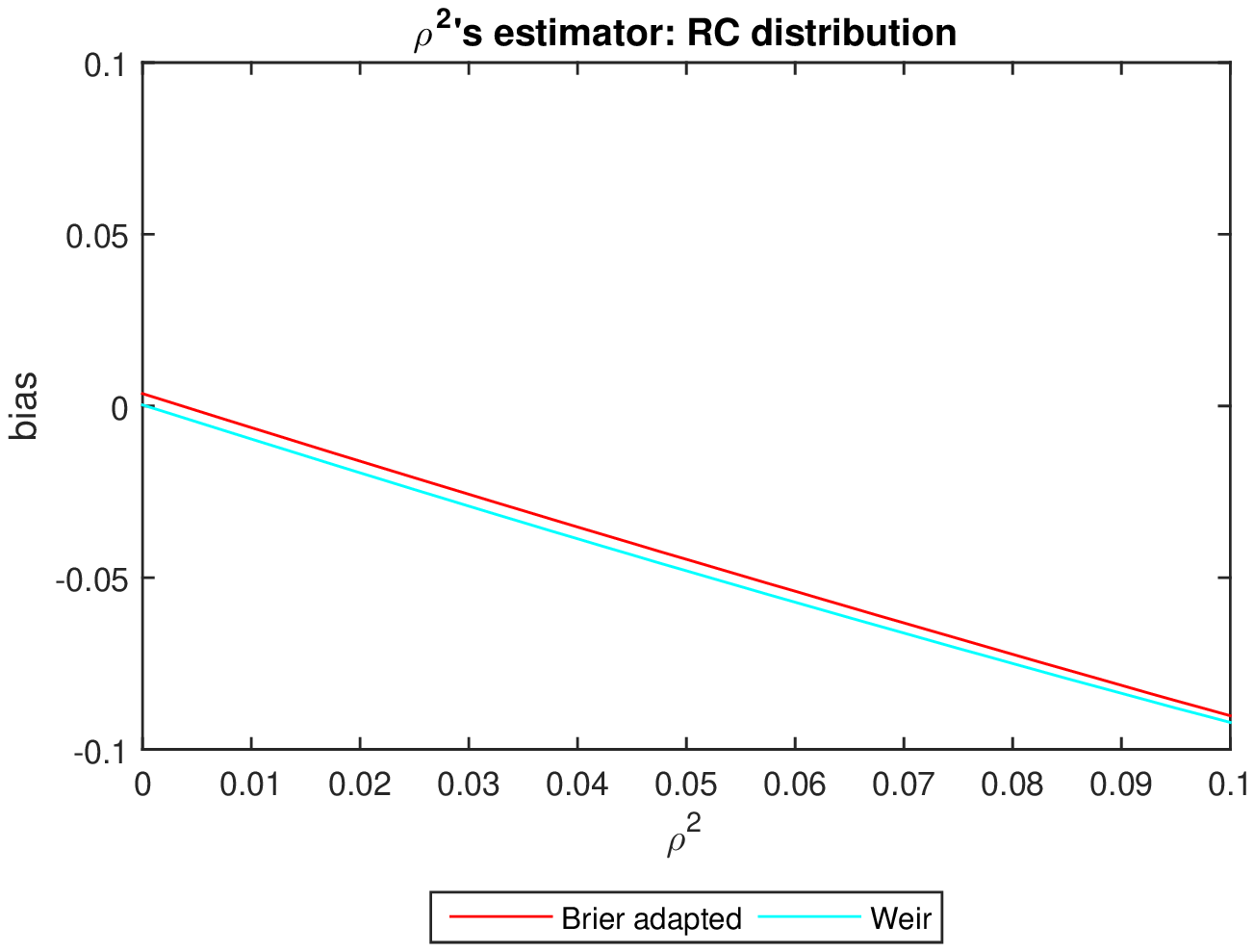}%
}
\end{tabular}
\caption{RMSE and bias of the Brier adapted $\widehat{\rho}^{2}$ and Weir's $\overline{\rho}^{2}$ for small values of $\rho^2$ when DM, NI and RC distributions are considered and  the theoretical probabilities are equal to the estimates  for locus D8S1179.\label{fig4.4}}%
\end{figure}%

\newpage

\section{Concluding remarks\label{sec7}}

This paper deals with log-linear models for studying the intracluster
correlation coefficient in clustered multinomial data. As no distributional
assumption is made, only the first two moment assumptions are considered,
quasi-likelihood methods are followed. With the saturated log-linear model the
non-parametric estimators of the intracluster correlation coefficient are
considered, and the semi-parametric estimators arise for general log-linear
models. New estimators are proposed for log-linear modeling in overdispersed
clustered multinomial data with unequal cluster sizes, valid either in a
non-paramateric and semi-parametric setting. Big differences are found in the
Monte-Carlo simulation study, when comparing the root of the mean square error
and the bias of the new estimators of the intracluster correlation coefficient
with the clasical ones. In addition, quasi minimum $\phi$-divergence
estimators are proposed and from the Monte Carlo experiments we saw that it is
possible to decrease the root of the mean square error in comparison with the
quasi-maximum likelihood estimators. These results of this paper could be
extended for any generalized linear model and the new estimators are promising
to improve the quality of the goodness-of-fit test statistics for log-linear
models in overdispersed clustered multinomial data.

The referees suggested to us to consider the interesting problems related to
mising values as well as to get the standard errors of $\widehat{\rho}%
_{n,N}^{2}.$ We know that the problem of missing values has been very well
solved in Chapter 2 of the PhD thesis of Raim (2014) for the random clumped
distribution. We think that the problem associated to missing data using the
modelization given in this paper, applying log-linear models, requires a
separate paper. The standard errors of $\widehat{\rho}_{n,N}^{2}$ requires
also a paper in the line of the paper of Weir and Hill (2002).

\textbf{Acknowledgement. }We would like to thank the referees for their
helpful comments and suggestions. This research is supported by the Spanish
Grant MTM2012-33740 from Ministerio de Economia y Competitividad.

%

\appendix

\section{Appendix}

\subsection{Zero-inflated binomial distribution\label{SecA.0}}

The binomial distribution with zero inflation in the first cell, i.e.
$n$-inflation in the second cell, is given by%
\[
\left(  \left.
\begin{pmatrix}
Y_{1}\\
Y_{2}%
\end{pmatrix}
\right\vert V=v\right)  =\left\{
\begin{array}
[c]{lll}%
\mathcal{M}\left(  n,%
\begin{pmatrix}
p_{1}(\boldsymbol{\theta})\\
p_{2}(\boldsymbol{\theta})
\end{pmatrix}
\right)  , & \text{if }v=1, & \text{with }\Pr(V=1)=w\\
n\boldsymbol{e}_{2}, & \text{if }v=0, & \text{with }\Pr(V=0)=1-w
\end{array}
\right.  .
\]
Its first order moment vector is given by%
\begin{align*}
E\left[
\begin{pmatrix}
Y_{1}\\
Y_{2}%
\end{pmatrix}
\right]   &  =E\left[  E\left[  \left.
\begin{pmatrix}
Y_{1}\\
Y_{2}%
\end{pmatrix}
\right\vert V\right]  \right]  \\
&  =E\left[  \left.
\begin{pmatrix}
Y_{1}\\
Y_{2}%
\end{pmatrix}
\right\vert V=1\right]  \Pr\left(  V=1\right)  +E\left[  \left.
n\boldsymbol{e}_{2}\right\vert V=0\right]  \Pr\left(  V=0\right)  \\
&  =n%
\begin{pmatrix}
wp_{1}(\boldsymbol{\theta})\\
1-wp_{1}(\boldsymbol{\theta})
\end{pmatrix}
.
\end{align*}
The derivation for the the second order moment matrix calculation is given by%
\begin{align*}
E\left[  Var\left[  \left.
\begin{pmatrix}
Y_{1}\\
Y_{2}%
\end{pmatrix}
\right\vert V\right]  \right]   &  =Var\left[  \left.
\begin{pmatrix}
Y_{1}\\
Y_{2}%
\end{pmatrix}
\right\vert V=1\right]  \Pr\left(  V=1\right)  +Var\left[  \left.
n\boldsymbol{e}_{2}\right\vert V=0\right]  \Pr\left(  V=0\right)  \\
&  =Var\left[  \mathcal{M}\left(  n,%
\begin{pmatrix}
p_{1}(\boldsymbol{\theta})\\
p_{2}(\boldsymbol{\theta})
\end{pmatrix}
\right)  \right]  w\\
&  =nwp_{1}(\boldsymbol{\theta})\left(  1-p_{1}(\boldsymbol{\theta})\right)
\begin{pmatrix}
1 & -1\\
-1 & 1
\end{pmatrix}
,
\end{align*}%
\begin{align*}
&  Var\left[  E\left[  \left.
\begin{pmatrix}
Y_{1}\\
Y_{2}%
\end{pmatrix}
\right\vert V\right]  \right]  \\
&  =E\left[  E\left[  \left.
\begin{pmatrix}
Y_{1}\\
Y_{2}%
\end{pmatrix}
\right\vert V\right]  E^{T}\left[  \left.
\begin{pmatrix}
Y_{1}\\
Y_{2}%
\end{pmatrix}
\right\vert V\right]  \right]  -E\left[  E\left[  \left.
\begin{pmatrix}
Y_{1}\\
Y_{2}%
\end{pmatrix}
\right\vert V\right]  \right]  E^{T}\left[  E\left[  \left.
\begin{pmatrix}
Y_{1}\\
Y_{2}%
\end{pmatrix}
\right\vert V\right]  \right]  \\
&  =E\left[  \left.
\begin{pmatrix}
Y_{1}\\
Y_{2}%
\end{pmatrix}
\right\vert V=1\right]  E^{T}\left[  \left.
\begin{pmatrix}
Y_{1}\\
Y_{2}%
\end{pmatrix}
\right\vert V=1\right]  w+E\left[  \left.
\begin{pmatrix}
Y_{1}\\
Y_{2}%
\end{pmatrix}
\right\vert V=0\right]  E^{T}\left[  \left.
\begin{pmatrix}
Y_{1}\\
Y_{2}%
\end{pmatrix}
\right\vert V=0\right]  (1-w)\\
&  -n^{2}%
\begin{pmatrix}
wp_{1}(\boldsymbol{\theta})\\
1-wp_{1}(\boldsymbol{\theta})
\end{pmatrix}%
\begin{pmatrix}
wp_{1}(\boldsymbol{\theta}) & 1-wp_{1}(\boldsymbol{\theta})
\end{pmatrix}
\\
&  =n^{2}(1-w)wp_{1}^{2}(\boldsymbol{\theta})%
\begin{pmatrix}
1 & -1\\
-1 & 1
\end{pmatrix}
,
\end{align*}
and hence%
\begin{align*}
Var\left[
\begin{pmatrix}
Y_{1}\\
Y_{2}%
\end{pmatrix}
\right]   &  =E\left[  Var\left[  \left.
\begin{pmatrix}
Y_{1}\\
Y_{2}%
\end{pmatrix}
\right\vert V\right]  \right]  +Var\left[  E\left[  \left.
\begin{pmatrix}
Y_{1}\\
Y_{2}%
\end{pmatrix}
\right\vert V\right]  \right]  \\
&  =nwp_{1}(\boldsymbol{\theta})\left[  \left(  1-p_{1}(\boldsymbol{\theta
})\right)  +n(1-w)p_{1}(\boldsymbol{\theta})\right]
\begin{pmatrix}
1 & -1\\
-1 & 1
\end{pmatrix}
\\
&  =nwp_{1}(\boldsymbol{\theta})(1-wp_{1}(\boldsymbol{\theta}))(1+\rho
^{2}(n-1))%
\begin{pmatrix}
1 & -1\\
-1 & 1
\end{pmatrix}
,
\end{align*}
where%
\[
\rho^{2}=\frac{(1-w)p_{1}(\boldsymbol{\theta})}{1-wp_{1}(\boldsymbol{\theta}%
)},\quad\text{for any }w\in(0,1).
\]
This result matches the one given in Morel and Neerchal (2012, page 83). Let
\[
\left(  \left.  \boldsymbol{Y}\right\vert V=v\right)  =\left\{
\begin{array}
[c]{lll}%
\mathcal{M}\left(  n,\boldsymbol{p}(\boldsymbol{\theta})\right)  , & \text{if
}v=1, & \text{with }\Pr(V=1)=w\\
n\boldsymbol{e}_{M}, & \text{if }v=0, & \text{with }\Pr(V=0)=1-w
\end{array}
\right.
\]
be the multinomial distribution with zero inflation in the first $M-1$ cells,
i.e. $n$-inflation in the $M$-th cell.

For $M\geq3$, a univariate homogeneous intraclass correlation coefficient,
$\rho^{2}$, seems not to be an appropriate measure to characterize the
variability of this distribution, since the intraclass correlation along the
cells seems to be heterogeous. The reason for this is that for $M\geq3$ there
is not an expression for the variance-covariance matrix of the multinomial
distribution defined as a matrix not depending on parameters multiplied by a
scalar with all the information about the parameters of the distribution.

\subsection{Proof of Theorem \ref{Dist}\label{SecA.1}}

Let%
\[
\boldsymbol{S}_{\boldsymbol{Y}}=\frac{1}{N-1}\sum_{\ell=1}^{N}\left(
\boldsymbol{Y}^{(\ell)}-n\widehat{\boldsymbol{p}}\right)  \left(
\boldsymbol{Y}^{(\ell)}-n\widehat{\boldsymbol{p}}\right)  ^{T},
\]
the matrix of quasi-variances and quasi-covariances of the simple random
sample $\boldsymbol{Y}^{(1)},...,\boldsymbol{Y}^{(N)}$ and
\begin{align*}
\overline{\boldsymbol{S}}_{\boldsymbol{Y}} &  =\mathrm{diag}(\boldsymbol{S}%
_{\boldsymbol{Y}})=%
\begin{pmatrix}
S_{Y_{1}}^{2} &  & \\
& \ddots & \\
&  & S_{Y_{M}}^{2}%
\end{pmatrix}
,\\
S_{Y_{r}}^{2} &  =\frac{1}{N-1}\sum_{\ell=1}^{N}(Y^{(\ell,r)}-n\widehat{p}%
_{r})^{2}.
\end{align*}
It is well-known that each diagonal element of $\overline{\boldsymbol{S}%
}_{\boldsymbol{Y}}$\ is a consistent estimator of each diagonal element of
$\vartheta_{n}n\boldsymbol{\Sigma}_{\boldsymbol{p}(\boldsymbol{\theta})}$,
i.e.%
\[
\mathrm{E}\left[  \overline{\boldsymbol{S}}_{\boldsymbol{Y}}\right]
=\mathrm{diag}\{\mathrm{E}\left[  \boldsymbol{S}_{\boldsymbol{Y}}\right]
\}=\mathrm{diag}\{\mathrm{Var}[\boldsymbol{Y}^{(\ell)}]\}=\mathrm{diag}%
\{\vartheta_{n}n\boldsymbol{\Sigma}_{\boldsymbol{p}(\boldsymbol{\theta})}\},
\]
and%
\begin{align}
&  S_{Y_{r}}^{2}\overset{P}{\underset{N\rightarrow\infty}{\longrightarrow}%
}\vartheta_{n}np_{r}(\boldsymbol{\theta})\left(  1-p_{r}(\boldsymbol{\theta
})\right)  ,\quad r=1,...,M,\label{cons}\\
&  \text{or}\quad\overline{\boldsymbol{S}}_{\boldsymbol{Y}}%
\overset{P}{\underset{N\rightarrow\infty}{\longrightarrow}}\mathrm{diag}%
(\vartheta_{n}n\boldsymbol{\Sigma}_{\boldsymbol{p}(\boldsymbol{\theta}%
)}).\nonumber
\end{align}
It is not difficult to establish that%
\begin{equation}
\mathrm{trace}(\overline{\boldsymbol{S}}_{\boldsymbol{Y}})=\sum_{r=1}%
^{M}S_{Y_{r}}^{2}=\mathrm{trace}(\boldsymbol{S}_{\boldsymbol{Y}})=\frac
{1}{N-1}\sum_{\ell=1}^{N}\left(  \boldsymbol{Y}^{(\ell)}%
-n\widehat{\boldsymbol{p}}\right)  ^{T}\left(  \boldsymbol{Y}^{(\ell
)}-n\widehat{\boldsymbol{p}}\right)  ,\label{trace}%
\end{equation}
which is consistent for $\mathrm{trace}(\vartheta_{n}n\boldsymbol{\Sigma
}_{\boldsymbol{p}(\boldsymbol{\theta})})=\vartheta_{n}n\sum_{r=1}^{M}%
p_{r}(\boldsymbol{\theta})\left(  1-p_{r}(\boldsymbol{\theta})\right)  $. We
know that the chi-square test-statistic $X^{2}(\widetilde{\boldsymbol{Y}})$,
given in (\ref{chi}), has an asymptotic $\mathcal{\chi}_{(N-1)(M-1)}^{2}$
distribution for fixed values of number of clusters $N$ and an increasing
cluster size, $n$, under the assumption of inter-cluster level homogeneity.
However, this distribution is not a useful device for the proof. Based on the
expression of the chi-square test-statistic, $X^{2}(\widetilde{\boldsymbol{Y}%
})$, in terms of the variance-covariance matrix, as well as the same steps to
obtain the expression and consistency of (\ref{trace}), we are going to
establish (\ref{cons2}). We have%
\[
\mathrm{trace}(\overline{\boldsymbol{S}}_{\boldsymbol{Y}}\tfrac{1}%
{n}\boldsymbol{D}_{\boldsymbol{p}(\boldsymbol{\theta})}^{-1})=\frac{1}%
{N-1}\sum_{\ell=1}^{N}\left(  \boldsymbol{Y}^{(\ell)}-n\widehat{\boldsymbol{p}%
}\right)  ^{T}\tfrac{1}{n}\boldsymbol{D}_{\boldsymbol{p}(\boldsymbol{\theta}%
)}^{-1}\left(  \boldsymbol{Y}^{(\ell)}-n\widehat{\boldsymbol{p}}\right)
\]
and%
\begin{align*}
\mathrm{E}\left[  \mathrm{trace}(\overline{\boldsymbol{S}}_{\boldsymbol{Y}%
}\tfrac{1}{n}\boldsymbol{D}_{\boldsymbol{p}(\boldsymbol{\theta})}%
^{-1})\right]   &  =\mathrm{traceE}\left[  \overline{\boldsymbol{S}%
}_{\boldsymbol{Y}}\tfrac{1}{n}\boldsymbol{D}_{\boldsymbol{p}%
(\boldsymbol{\theta})}^{-1}\right]  =\mathrm{trace}\left(  \mathrm{E}\left[
\overline{\boldsymbol{S}}_{\boldsymbol{Y}}\right]  \tfrac{1}{n}\boldsymbol{D}%
_{\boldsymbol{p}(\boldsymbol{\theta})}^{-1}\right)  =\mathrm{trace}\left(
\vartheta_{n}n\boldsymbol{\Sigma}_{\boldsymbol{p}(\boldsymbol{\theta})}%
\tfrac{1}{n}\boldsymbol{D}_{\boldsymbol{p}(\boldsymbol{\theta})}^{-1}\right)
\\
&  =\vartheta_{n}\mathrm{trace}\left(  \boldsymbol{\Sigma}_{\boldsymbol{p}%
(\boldsymbol{\theta})}\boldsymbol{D}_{\boldsymbol{p}(\boldsymbol{\theta}%
)}^{-1}\right)  =\vartheta_{n}\mathrm{trace}\left(  \left(  \boldsymbol{D}%
_{\boldsymbol{p}(\boldsymbol{\theta})}-\boldsymbol{p}(\boldsymbol{\theta
})\boldsymbol{p}^{T}(\boldsymbol{\theta})\right)  \boldsymbol{D}%
_{\boldsymbol{p}(\boldsymbol{\theta})}^{-1}\right)  \\
&  =\vartheta_{n}\left[  \mathrm{trace}(\boldsymbol{I}_{M})-\mathrm{trace}%
(\boldsymbol{p}(\boldsymbol{\theta})\boldsymbol{1}_{M}^{T})\right]
=\vartheta_{n}(M-1).
\end{align*}
Hence,%
\[
\mathrm{E}\left[  \frac{1}{M-1}\mathrm{trace}(\overline{\boldsymbol{S}%
}_{\boldsymbol{Y}}\tfrac{1}{n}\boldsymbol{D}_{\boldsymbol{p}%
(\boldsymbol{\theta})}^{-1})\right]  =\mathrm{E}\left[  \frac{1}%
{(N-1)(M-1)}\sum_{\ell=1}^{N}\left(  \boldsymbol{Y}^{(\ell)}%
-n\widehat{\boldsymbol{p}}\right)  ^{T}\tfrac{1}{n}\boldsymbol{D}%
_{\boldsymbol{p}(\boldsymbol{\theta})}^{-1}\left(  \boldsymbol{Y}^{(\ell
)}-n\widehat{\boldsymbol{p}}\right)  \right]  =\vartheta_{n},
\]
and taking into account that $\widehat{\boldsymbol{p}}$ is a consistent
estimator of $\boldsymbol{p}(\boldsymbol{\theta})$, as $N\rightarrow\infty$,
as well as (\ref{cons}),
\[
\frac{1}{M-1}\mathrm{trace}(\overline{\boldsymbol{S}}_{\boldsymbol{Y}}%
\tfrac{1}{n}\boldsymbol{D}_{\widehat{\boldsymbol{p}}}^{-1})=\frac
{1}{(N-1)(M-1)}\sum_{\ell=1}^{N}\left(  \boldsymbol{Y}^{(\ell)}%
-n\widehat{\boldsymbol{p}}\right)  ^{T}\tfrac{1}{n}\boldsymbol{D}%
_{\widehat{\boldsymbol{p}}}^{-1}\left(  \boldsymbol{Y}^{(\ell)}%
-n\widehat{\boldsymbol{p}}\right)  =\frac{X^{2}(\widetilde{\boldsymbol{Y}}%
)}{(N-1)(M-1)}%
\]
tends in probability to $\vartheta_{n}$, as $N\rightarrow\infty$. In other
words,%
\[
\frac{X^{2}(\widetilde{\boldsymbol{Y}})}{(N-1)(M-1)}=\frac{1}{(M-1)n}%
\sum_{r=1}^{M}\frac{1}{\widehat{p}_{r}}S_{Y_{r}}^{2}%
\overset{P}{\underset{N\rightarrow\infty}{\longrightarrow}}\frac{\vartheta
_{n}n}{(M-1)n}\sum_{r=1}^{M}\frac{p_{r}(\boldsymbol{\theta})}{p_{r}%
(\boldsymbol{\theta})}\left(  1-p_{r}(\boldsymbol{\theta})\right)
=\vartheta_{n}.
\]
In addition, taking into account (\ref{DF}), the right hand size of
(\ref{cons2}) follows. Finally, we like to mention that even though
$X^{2}(\widetilde{\boldsymbol{Y}})$ and $\vartheta_{n}(N-1)(M-1)$ have the
same expectation for a fixed value of $N$, this proof is not trivial since
$\vartheta_{n}(N-1)(M-1)$ as well as $X^{2}(\widetilde{\boldsymbol{Y}})$ tend
to infinite as $N\rightarrow\infty$.

\subsection{Proof of Theorem \ref{Th1}\label{SecA.2}}

By applying the Central Limit Theorem it holds (\ref{CLT}). Hence, from Pardo
(2006, formula (7.10)), for the minimum phi-divergence estimator of
$\boldsymbol{\theta}$ of a log-linear model it holds%
\begin{equation}
\sqrt{N}(\widehat{\boldsymbol{\theta}}_{\phi}-\boldsymbol{\theta}_{0})=\left(
\boldsymbol{\boldsymbol{W}}^{T}\boldsymbol{\Sigma\boldsymbol{_{\boldsymbol{p}%
\left(  \theta_{0}\right)  }}W}\right)  ^{-1}\boldsymbol{W}^{T}%
\boldsymbol{\Sigma}_{p\left(  \boldsymbol{\theta}_{0}\right)  }\boldsymbol{D}%
_{\boldsymbol{p}\left(  \theta_{0}\right)  }^{-1}\sqrt{N}\left(
\widehat{\boldsymbol{p}}-\boldsymbol{p}\left(  \boldsymbol{\theta}_{0}\right)
\right)  +o_{p}\left(  \boldsymbol{1}_{M_{0}}\right)  , \label{8}%
\end{equation}
and the variance-covariance matrix of $\sqrt{N}(\widehat{\boldsymbol{\theta}%
}_{\phi}-\boldsymbol{\theta}_{0})$ is%
\begin{align}
&  \tfrac{\vartheta_{n}}{n}\left(  \boldsymbol{\boldsymbol{W}}^{T}%
\boldsymbol{\Sigma\boldsymbol{_{\boldsymbol{p}\left(  \theta_{0}\right)  }}%
W}\right)  ^{-1}\boldsymbol{W}^{T}\boldsymbol{\Sigma}_{p\left(
\boldsymbol{\theta}_{0}\right)  }\boldsymbol{D}_{\boldsymbol{p}\left(
\theta_{0}\right)  }^{-1}\boldsymbol{\Sigma}_{\boldsymbol{p}\left(
\boldsymbol{\theta}_{0}\right)  }\boldsymbol{D}_{\boldsymbol{p}\left(
\theta_{0}\right)  }^{-1}\boldsymbol{\Sigma}_{\boldsymbol{p}\left(
\boldsymbol{\theta}_{0}\right)  }\boldsymbol{W}\left(
\boldsymbol{\boldsymbol{W}}^{T}\boldsymbol{\Sigma\boldsymbol{_{\boldsymbol{p}%
\left(  \theta_{0}\right)  }}W}\right)  ^{-1}\nonumber\\
&  =\tfrac{\vartheta_{n}}{n}\left(  \boldsymbol{\boldsymbol{W}}^{T}%
\boldsymbol{\Sigma\boldsymbol{_{\boldsymbol{p}\left(  \theta_{0}\right)  }}%
W}\right)  ^{-1}. \label{8Var}%
\end{align}
The last equality comes from%
\[
\boldsymbol{\Sigma}_{\boldsymbol{p}\left(  \boldsymbol{\theta}_{0}\right)
}\boldsymbol{D}_{\boldsymbol{p}\left(  \theta_{0}\right)  }^{-1}%
\boldsymbol{\Sigma}_{\boldsymbol{p}\left(  \boldsymbol{\theta}_{0}\right)
}=\boldsymbol{\Sigma}_{\boldsymbol{p}\left(  \boldsymbol{\theta}_{0}\right)
}.
\]
From the Taylor expansion of $\boldsymbol{p}(\widehat{\boldsymbol{\theta}%
}_{\phi})$ around $\boldsymbol{p}(\boldsymbol{\theta}_{0})$ we obtain%
\begin{equation}
\sqrt{N}(\boldsymbol{p}(\widehat{\boldsymbol{\theta}}_{\phi})-\boldsymbol{p}%
(\boldsymbol{\theta}_{0}))=\boldsymbol{\Sigma\boldsymbol{_{\boldsymbol{p}%
\left(  \theta_{0}\right)  }}W}\sqrt{N}(\widehat{\boldsymbol{\theta}}_{\phi
}-\boldsymbol{\theta}_{0})+o_{p}\left(  \boldsymbol{1}_{M}\right)  ,
\label{8b}%
\end{equation}
and the variance-covariance matrix of $\sqrt{N}(\boldsymbol{p}%
(\widehat{\boldsymbol{\theta}}_{\phi})-\boldsymbol{p}(\boldsymbol{\theta}%
_{0}))$ is%
\begin{equation}
\tfrac{\vartheta_{n}}{n}\boldsymbol{\Sigma\boldsymbol{_{\boldsymbol{p}\left(
\theta_{0}\right)  }}W}\left(  \boldsymbol{\boldsymbol{W}}^{T}%
\boldsymbol{\Sigma\boldsymbol{_{\boldsymbol{p}\left(  \theta_{0}\right)  }}%
W}\right)  ^{-1}\boldsymbol{W}^{T}\boldsymbol{\Sigma}_{p\left(
\boldsymbol{\theta}_{0}\right)  }. \label{8Varb}%
\end{equation}
Since $\sqrt{N}\left(  \widehat{\boldsymbol{p}}-\boldsymbol{p}\left(
\boldsymbol{\theta}_{0}\right)  \right)  $ is normal and centred, from
(\ref{8}) and (\ref{8Var}), (\ref{9}) is obtained. Similarly, since $\sqrt
{N}(\widehat{\boldsymbol{\theta}}_{\phi}-\boldsymbol{\theta}_{0})$ is normal
and centred, from (\ref{8b}) and (\ref{8Varb}), (\ref{10}) is obtained.

\subsection{Derivation of Formula (\ref{CLT2})\label{SecA.3}}

Multiplying (\ref{TCL0}) by $\left.  \sqrt{N_{g}}n_{g}\right/  \sum
\limits_{h=1}^{G}n_{h}N_{h}$%
\[
w_{g}(\widehat{\boldsymbol{p}}^{(g)}-\boldsymbol{p}\left(  \boldsymbol{\theta
}_{0}\right)  )\overset{\mathcal{L}}{\underset{N_{g}\rightarrow\infty
}{\longrightarrow}}\mathcal{N}\left(  \boldsymbol{0}_{M},\tfrac{n_{g}%
N_{g}\vartheta_{n_{g}}}{\left(  \sum\nolimits_{h=1}^{G}n_{h}N_{h}\right)
^{2}}\boldsymbol{\Sigma}_{\boldsymbol{p}\left(  \boldsymbol{\theta}%
_{0}\right)  }\right)  ,
\]
hence summing up from $g=1$ to $G$\ and by the independence of clusters%
\[
\sum\limits_{g=1}^{G}w_{g}(\widehat{\boldsymbol{p}}^{\left(  g\right)
}-\boldsymbol{p}\left(  \boldsymbol{\theta}_{0}\right)  )=\left(
\widehat{\boldsymbol{p}}-\boldsymbol{p}\left(  \boldsymbol{\theta}_{0}\right)
\right)  \overset{\mathcal{L}}{\underset{N_{g}\rightarrow\infty
,\;g=1,...,G}{\longrightarrow}}\mathcal{N}\left(  \boldsymbol{0}_{M}%
,\tfrac{\sum\nolimits_{g=1}^{G}n_{g}N_{g}\vartheta_{n_{g}}}{\left(
\sum\nolimits_{h=1}^{G}n_{h}N_{h}\right)  ^{2}}\boldsymbol{\Sigma
}_{\boldsymbol{p}\left(  \boldsymbol{\theta}_{0}\right)  }\right)  .
\]
Finally multiplying the previous expression by $\left.  \sum\nolimits_{h=1}%
^{G}n_{h}N_{h}\right/  \sqrt{\sum\nolimits_{g=1}^{G}n_{g}N_{g}\vartheta
_{n_{g}}}$, the desired expression is obtained.

\subsection{Algorithms for Dirichlet-multinomial, n-inflated and
random-clumped distributions\label{Alg}}

The usual parameters of the $M$-dimensional random variable $\boldsymbol{Y}%
=(Y_{1},...,Y_{M})^{T}$\ with Dirichlet-multinomial distribution are
$\boldsymbol{\alpha}=\left(  \alpha_{11},...,\alpha_{M1}\right)  ^{T}$, where
$\alpha_{r1}=\frac{1-\rho^{2}}{\rho^{2}}p_{r}\left(  \boldsymbol{\theta
}\right)  $, $r=1,...,M$. For convenience it is considered with parameters
$\boldsymbol{\beta}=%
\begin{pmatrix}
\rho^{2}\\
\boldsymbol{p}(\boldsymbol{\theta})
\end{pmatrix}
$, $\boldsymbol{p}\left(  \boldsymbol{\theta}\right)  =\left(  p_{1}\left(
\boldsymbol{\theta}\right)  ,...,p_{M}\left(  \boldsymbol{\theta}\right)
\right)  ^{T}$, and is generated as follows:\medskip\texttt{\newline\noindent
STEP 1. Generate }$B_{1}\sim Beta(\alpha_{11},\alpha_{12})$, with $\alpha
_{11}=\frac{1-\rho^{2}}{\rho^{2}}p_{1}\left(  \boldsymbol{\theta}\right)  $,
$\alpha_{12}=\frac{1-\rho^{2}}{\rho^{2}}(1-p_{1}\left(  \boldsymbol{\theta
}\right)  )$.\smallskip\texttt{\newline STEP 2. Generate }$\left(  Y_{1}%
|B_{1}=b_{1}\right)  \sim Bin(n,b_{1})$.\texttt{\smallskip\newline STEP 3. For
}$r=2,...,M-1$\texttt{ do:\smallskip\newline\hspace*{1.5cm}Generate }%
$B_{r}\sim Beta(\alpha_{r1},\alpha_{r2})$\texttt{, }with $\alpha_{r1}%
=\frac{1-\rho^{2}}{\rho^{2}}p_{r}\left(  \boldsymbol{\theta}\right)  $,
$\alpha_{r2}=\frac{1-\rho^{2}}{\rho^{2}}\left(  1-\sum\limits_{h=1}^{r}%
p_{h}\left(  \boldsymbol{\theta}\right)  \right)  $.\texttt{\smallskip
\newline\hspace*{1.5cm}Generate }$\left(  Y_{r}|Y_{1}=y_{1},...,Y_{r-1}%
=y_{r-1},B_{r}=b_{r}\right)  \sim Bin\left(  n-\sum\limits_{h=1}^{r-1}%
y_{h},b_{r}\right)  $.\smallskip\texttt{\newline STEP 4. Do }$\left(
Y_{M}|Y_{1}=y_{1},...,Y_{M-1}=y_{M-1}\right)  =n-\sum\limits_{h=1}^{M-1}y_{h}%
$.\smallskip

The random variable $\boldsymbol{Y}=(Y_{1},...,Y_{M})^{T}$\ of the
$n$-inflated multinomial distribution with parameters $\boldsymbol{\beta}$,
$\boldsymbol{p}\left(  \boldsymbol{\theta}\right)  $, is generated as
follows:\medskip\texttt{\newline\noindent STEP 1. Generate }$V\sim
Ber(\rho^{2})$.\smallskip\texttt{\newline STEP 2. Generate}%
\[
\mathtt{\ }\left(  \boldsymbol{Y|}V=v\right)  =\left\{
\begin{array}
[c]{ll}%
\mathcal{M}(n,\boldsymbol{p}\left(  \boldsymbol{\theta}\right)  ), & \text{if
}v=0\\
n\mathcal{M}(1,\boldsymbol{p}\left(  \boldsymbol{\theta}\right)  ), & \text{if
}v=1
\end{array}
\right.  .
\]
\smallskip

The random variable $\boldsymbol{Y}=(Y_{1},...,Y_{M})^{T}$\ of the random
clumped distribution with parameters $\boldsymbol{\beta}$, $\boldsymbol{p}%
\left(  \boldsymbol{\theta}\right)  $, is generated as follows:\medskip
\texttt{\newline\noindent STEP 1. Generate }$\boldsymbol{Y}_{0}=(Y_{01}%
,...,Y_{0M})^{T}\sim\mathcal{M}(1,\boldsymbol{p}\left(  \boldsymbol{\theta
}\right)  )$.\smallskip\texttt{\newline STEP 2. Generate }$K_{1}\sim
Bin(n,\rho)$.\texttt{\smallskip\newline STEP 3. Generate }$\left(
\boldsymbol{Y}_{1}|K_{1}=k_{1}\right)  =\left(  (Y_{11},...,Y_{1M})^{T}%
|K_{1}=k_{1}\right)  \sim\mathcal{M}(n-k_{1},\boldsymbol{p}\left(
\boldsymbol{\theta}\right)  )$.\smallskip\texttt{\newline STEP 4. Do }$\left(
\boldsymbol{Y|}K_{1}=k_{1}\right)  \boldsymbol{=Y}_{0}k_{1}+\left(
\boldsymbol{Y}_{1}|K_{1}=k_{1}\right)  $.\smallskip\texttt{\newline}For the
details about the equivalence of this algorithm and (\ref{model2}), see Morel
and Nagaraj (1993).

It is interesting to note that there exists the package "Modeling
overdispersion in $R$" useful to generate the distributions considered in this
Appendix. For more details see Raim et al (2015).
\end{document}